\documentclass[12pt]{article}
\usepackage{amsmath}
\usepackage{graphicx}
\usepackage{enumerate}
\usepackage{natbib}
\usepackage{url} 

\newcommand{\blind}{1}

\usepackage{setspace}
\usepackage{enumitem}
\usepackage{amssymb}

\usepackage{amsmath,amsfonts,bm}

\usepackage{amsthm}
\usepackage{microtype}
\usepackage{subfigure}
\usepackage{booktabs} 
\usepackage{multirow} 
\usepackage[toc,page]{appendix}
\usepackage{comment}
\usepackage{caption}
\usepackage{bbm}
\usepackage[dvipsnames]{xcolor}
\usepackage{xr}

\newcommand{\Mean}{{\mathbb{E}}}

\newcommand{\Cov}{{\mbox{Cov}}}

\newcommand{\diag}{{\mbox{diag}}}
\newcommand{\prob}{{\mathbb{P}}}
\DeclareMathOperator*{\argmin}{arg\,min}

\newcommand{\vc}{\textrm{vec}}

\newtheorem{thm}{Theorem}
\newtheorem{prop}{Proposition}
\newtheorem{lemma}{Lemma}

\newtheorem{asmp}{Assumption}

\usepackage{hyperref}
\newcommand{\change}[1]{{\leavevmode\color{blue}{#1}}}

\newcommand{\RNum}[1]{\uppercase\expandafter{\romannumeral #1\relax}}

\newcommand{\QDE}{\textrm{CQDE}}
\newcommand{\QIE}{\textrm{CQIE}}
\newcommand{\QTE}{\textrm{CQTE}}
\newcommand{\R}{\mathbb{R}}
\newcommand{\M}{\mathbb{M}}
\usepackage{hyperref}

\usepackage{algorithm}
\usepackage{algcompatible}

\def\R{\mathbb{R}}

\begin{document}

\def\spacingset#1{\renewcommand{\baselinestretch}%
{#1}\small\normalsize} \spacingset{1}


\if1\blind
{
  \title{\bf Evaluating Dynamic Conditional Quantile Treatment Effects with Applications in Ridesharing}
  \author{Ting Li$^{*1}$, Chengchun Shi$^{*2}$, Zhaohua Lu$^{3}$, Yi Li$^{3}$,  and Hongtu Zhu$^4$ \thanks{
			The first two authors contribute equally to this paper. 
			Address for correspondence:
			Hongtu Zhu, Ph.D., E-mail: 
			htzhu@email.unc.edu. This work was partially finished when Dr  Zhu worked at Didi Chuxing.    
			The content is solely the responsibility
			of the authors and does not necessarily represent the official views
			of   any other funding agency.} 
 \hspace{.2cm} \\
	\textit{$^{1}$Shanghai University of Finance and Economics} \\
	\textit{$^{2}$London School of Economics and Political Science}  \\
	\textit{$^{3}$DiDi Chuxing}\\
	\textit{$^{4}$University of North Carolina at Chapel Hill} }
    \date{}
  \maketitle
} \fi

\if0\blind
{
  \bigskip
  \bigskip
  \bigskip
  \begin{center}
    {\LARGE\bf  Evaluating Dynamic Conditional Quantile Treatment Effects with Applications in Ridesharing }
\end{center}
  \medskip
} \fi

\bigskip
\begin{abstract}
Many modern tech companies, such as Google, Uber, and Didi, utilize online experiments (also known as A/B testing) to evaluate new policies against existing ones. While most studies concentrate on average treatment effects, situations with skewed and heavy-tailed outcome distributions may benefit from alternative criteria, such as quantiles. However, assessing dynamic quantile treatment effects (QTE) remains a challenge, particularly when dealing with data from ride-sourcing platforms that involve sequential decision-making across time and space.
In this paper, we establish a formal framework to calculate QTE conditional on characteristics independent of the treatment. Under specific model assumptions, we demonstrate that the dynamic conditional QTE (CQTE) equals the sum of individual CQTEs across time, even though the conditional quantile of cumulative rewards may not necessarily equate to the sum of conditional quantiles of individual rewards. This crucial insight significantly streamlines the estimation and inference processes for our target causal estimand. 
We then introduce two varying coefficient decision process (VCDP) models and devise an innovative method to test the dynamic CQTE. Moreover, we expand our approach to accommodate data from spatiotemporal dependent experiments and examine both conditional quantile direct and indirect effects. To showcase the practical utility of our method, we apply it to three real-world datasets from a ride-sourcing platform. Theoretical findings and comprehensive simulation studies further substantiate our proposal.
\end{abstract}

\noindent%
\vfill

\newpage
\spacingset{1.5} 

\section{Introduction}

Online experiments, often referred to as A/B testing in computer science literature, are widely utilized by technology companies (e.g., Google, Netflix, Microsoft) to assess the effectiveness of new products or policies in comparison to existing ones. These companies have developed in-house A/B testing platforms for evaluating treatment effects and providing valuable experimental insights. 
Take ridesourcing platforms like Uber, Lyft, and Didi as examples. These platforms operate within intricate spatiotemporal ecosystems, dynamically matching passengers with drivers \citep[see, for instance,][]{WangYang2019,Qin2020,zhou2021graph}. They implement online experiments to explore various order dispatch policies and customer recommendation initiatives. These innovative products hold the potential to enhance passenger engagement and satisfaction, diminish pickup waiting times, and boost driver earnings, ultimately leading to a more efficient and user-friendly transportation system.

In this study, we address the fundamental question of how to evaluate the difference between the \textit{quantile} return of a new product (treatment) and that of an existing one. Although the average treatment effect (ATE) is widely used in the literature to quantify the difference between two policies \citep[][]{imbens2015causal,wang2018bounded,kong2022identifiability}, it only considers the average effect and does not account for variability around the expectation. In applications with skewed and heavy-tailed outcome distributions, decision-makers are more interested in the quantile treatment effect (QTE), which offers a more comprehensive characterization of distributional effects beyond the mean and is robust to heavy-tailed errors \citep[see e.g., ][]{abadie2002instrumental,chernozhukov2006instrumental,chen2019causal}.
For example, in ridesourcing platforms, policymakers may want to determine which policy more effectively raises the lower tail of driver income. Furthermore, developing valid inferential tools for QTE can reveal how treatment effects differ by quantile and provide valuable information about the entire distribution.



Addressing the problem mentioned earlier presents two significant challenges. The first challenge involves efficiently inferring the dynamic QTE (quantile treatment effect), which is defined as the difference between the quantiles of cumulative outcomes under the new and old policies, in long horizon settings with weak signals.
In contrast to single-stage decision-making, policy makers for ridesourcing platforms assign treatments sequentially over time and across various locations. Existing estimators, such as those based on (augmented) inverse probability weighting   \citep[see e.g.,][Section 4]{wang2018quantile}, are subject to the curse of horizon, as described by  \citep{liu2018breaking}.  
  This means their variances increase exponentially with respect to the horizon (i.e., the number of decision stages).   
  Such approaches are inadequate in our context, where the horizon typically spans 24 or 48 stages and most policies improve key metrics by only 0.5\% to 2\% \citep{qin2022reinforcement}. 
  Furthermore, unlike the average cumulative outcome, which can be broken down into the sum of individual outcome expectations, the quantile of cumulative outcomes generally does not equal the sum of individual quantiles. This makes estimating our causal effect extremely challenging. Existing efficient evaluation methods designed for mean return, such as those proposed by \citet{kallus2022efficiently} and \citet{liao2020batch}, cannot be easily adapted to our situation.

The second challenge arises from handling the interference effect caused by temporal and spatial proximities in spatiotemporal dependent experiments. This interference effect results in a treatment applied at one location influencing not only its own outcome, but also the outcomes at other locations. Additionally, the current treatment is likely to affect both present and future outcomes. Neglecting these effects would produce a biased QTE estimator. As far as we are aware, there is no existing test capable of concurrently addressing both challenges.

\subsection{Related work}

A/B testing has been extensively researched in the literature, as evidenced by the works of \cite{yang2017framework} and \cite{zhou2020cluster}, among other references. In contrast to most existing A/B testing methods that focus on the Average Treatment Effect (ATE), Quantile Treatment Effects (QTE) have received less attention. Among the few available studies, \cite{liu2019large} proposed a scalable method to test QTE and construct associated confidence intervals. Moreover, \cite{wang2021conq} developed a nonparametric method to estimate QTEs at a continuous range of quantile locations, including  point-wise confidence intervals. 
More broadly, the estimation and inference of (conditional) QTEs have been considered in the causal inference literature, as seen in the works of \cite{chernozhukov2006instrumental}, \cite{firpo2007efficient}, and \cite{blanco2020bounds}. However, these methods predominantly address single-stage decision-making. To the best of our knowledge, this paper represents the first attempt to explore QTE in temporally and/or spatially dependent experiments.

Our paper is closely related to the rapidly expanding body of literature on off-policy evaluation in sequential decision-making. The majority of existing studies primarily concentrate on inferring the expected return under a fixed target policy or a data-dependent estimated optimal policy \citep{zhang2013robust,  shi2020breaking, kallus2022efficiently}. 
In recent years, several papers have explored policy evaluation beyond averages \citep{wang2018quantile,kallus2019localized,qi2022robustness,xu2022quantile}. These works propose using (augmented) inverse probability weighted estimators to evaluate specific robust metrics under a given target policy. As noted previously, these methods are subject to the curse of horizon and become less effective in long-horizon settings. Most notably, policy evaluation in spatiotemporal dependent experiments remains unexplored in the aforementioned studies.



Recent proposals have investigated causal inference with temporal or spatial interference, including studies by  \cite{Savje2020} and 
\cite{hu2021average}, among others. However, these methods primarily focus on the average effect.
Furthermore, our paper is closely related to the literature on distributional reinforcement learning \citep[see e.g.,][]{zhou2020non}.  Despite this connection, these studies primarily concentrate on the policy learning problem, and uncertainty quantification of a target policy's quantile value remains unexplored.

Lastly, our paper is connected to a line of research on quantitative analysis of ridesharing across various fields such as economics, operations research, statistics, and computer science \citep[see e.g.,][]{shi2022multi,zhao2022dynamic}. Nevertheless, quantile policy evaluation has not been examined in these papers.

\subsection{Contributions}

Our proposal offers three valuable contributions to existing literature. 
First, we present a framework for deducing dynamic conditional Quantile Treatment Effects (QTE), defined as dynamic QTE dependent on market features, irrespective of treatment history. While unconditional QTE may be of interest, as previously noted, it assumes a highly complex form in long horizon settings and is extremely challenging to identify when the signal is weak. In contrast, we demonstrate that under certain modeling assumptions, the proposed dynamic conditional QTE (CQTE) is equal to the sum of individual CQTE at each spatiotemporal unit, even though the conditional quantile of cumulative rewards does not necessarily equate to the sum of conditional quantiles of individual rewards. This finding significantly streamlines the estimation and inference processes for our causal estimand, making our proposal easily implementable in practice. Additionally, the estimated CQTE can exhibit a smaller variance compared to that of the unconditional counterpart.

Second, we introduce an innovative framework to test dynamic CQTE while accounting for the interference effect. We propose two Varying Coefficient Decision Process (VCDP) models, enabling the application of classical quantile regression \citep{koenker2001quantile}  for parameter estimation and subsequent inference. We then develop a two-step method for estimating CQTE, along with a bootstrap-assisted procedure for testing CQTE. We further extend our proposal to analyze spatiotemporally dependent data and to test Conditional Quantile Direct Effects (CQDE) and Conditional Quantile Indirect Effects (CQIE).

Third, we thoroughly examine the theoretical and finite sample properties of our methods. Theoretically, we prove the consistency of our proposed test procedure, allowing the horizon to diverge with the sample size. Notably, classical weak convergence theorems \citep{van1996weak}   necessitate a fixed horizon and are not directly applicable. Empirically, we apply our proposed method to real datasets obtained from a leading ridesourcing platform to assess the dynamic quantile treatment effects of new policies. 

\subsection{Organization of the paper}

The paper's structure is as follows: Section \ref{sec:data} describes data from online randomized experiments. Section \ref{sec:temporal design} covers temporally dependent experiments, the proposed model, and estimation and inference procedures. Section \ref{sec:ST method} extends the proposal to spatiotemporally dependent data. Section \ref{sec:DEIE} decomposes CQTE into CQDE and CQIE. Section \ref{sec:asymptotic results} presents the proposed CQTE test's asymptotic results. Section \ref{sec:real data} evaluates ridesourcing dispatching and repositioning policies, and Section \ref{sec:simulation} assesses our methods' finite sample performance using real-data-based simulations. 

\section{Data Description}\label{sec:data}

The purpose of this paper is to analyze three real datasets collected from Didi Chuxing, one of the world's leading ride-sharing companies. One dataset was collected during a time-dependent A/B experiment conducted in a city from December 10, 2021 to December 23, 2021. The goal of this experiment was to evaluate the performance of a newly designed order dispatching policy, which aimed to increase the number of fulfilled ride requests and boost drivers' total revenue. To protect privacy, we will not disclose the city name and the specific policy used. 
During the experiment, each day was divided into 24 equally spaced non-overlapping time intervals. The new policy (B) and the old policy (A) were alternated and assigned to these intervals every day. On the first day, we used an alternating sequence of AB$\ldots$AB, and on the second day, we used BA$\ldots$BA. Thereafter, we switched between A and B every two days, ensuring that each policy was used with equal probability at each time interval, meeting the positivity assumption. For more details, see Section \ref{sec:CQTE}. It is worth noting that such an alternating-time-interval design is commonly used in industries to reduce the variance of treatment effect estimators, as discussed in Section 6.2 of \citet{shi2022dynamic}. For further information, please refer to the article by Lyft on experimentation in a ride-sharing marketplace at \url{https://eng.lyft.com/experimentation-in-a-ridesharing-marketplace-b39db027a66e}.

The second dataset comes from a spatiotemporal-dependent experiment conducted in another city between February 19, 2020, and March 13, 2020. Each day is divided into 48 non-overlapping, equal time intervals, and the city is partitioned into 12 distinct, non-overlapping regions. On the first day of the experiment, the initial policy in each region is independently set to either the new or old policy with a $50\%$ probability. The temporal alternation design for time-dependent experiments is then applied in each region.

In addition to the two datasets from A/B experiments mentioned earlier, we also analyze a third dataset collected from an A/A experiment. In this case, the two policies being compared are identical, and the treatment effect is zero. The experiment took place in a specific city from July 13, 2021 to September 17, 2021. This analysis serves as a sanity check to examine the size property of the proposed test. We expect that our test will not reject the null hypothesis when applied to this dataset, as the true effect is zero.

The ridesharing system dynamically connects passengers and drivers in real-time. All three datasets include the number of call orders and the total online time of drivers for each time interval. These metrics represent the supply and demand in this two-sided market. The platform's outcomes include the drivers' total income, the answer rate (the number of call orders responded to), and the completion rate (the number of call orders completed) for each time interval. In our study, we are interested in determining whether the new policy improves drivers' total income at various quantile levels. 


The datasets exhibit four distinct characteristics. First, the horizon duration is typically much longer (e.g., 24 or 48) than the experiment duration, while the treatment effect is usually weak (e.g., 0.5\%-2\%). Second, both supply and demand are spatiotemporal networks that interact across time and location, as observed in panels (a) and (b) of Figure \ref{fig:AB_temporal}, which display drivers' online time and the number of call orders. Third, the outcome of interest follows a non-normal and heavy-tailed distribution, illustrated in panels (c) and (d) of Figure \ref{fig:AB_temporal}. Finally, there are interference effects over time and space, demonstrated in Figure \ref{fig:interference}, with temporal interference effects occurring when past actions impact future outcomes.


\begin{figure}[!t]
	\centering
\subfigure[]{  \includegraphics[height=3.5cm, width=3.6cm]{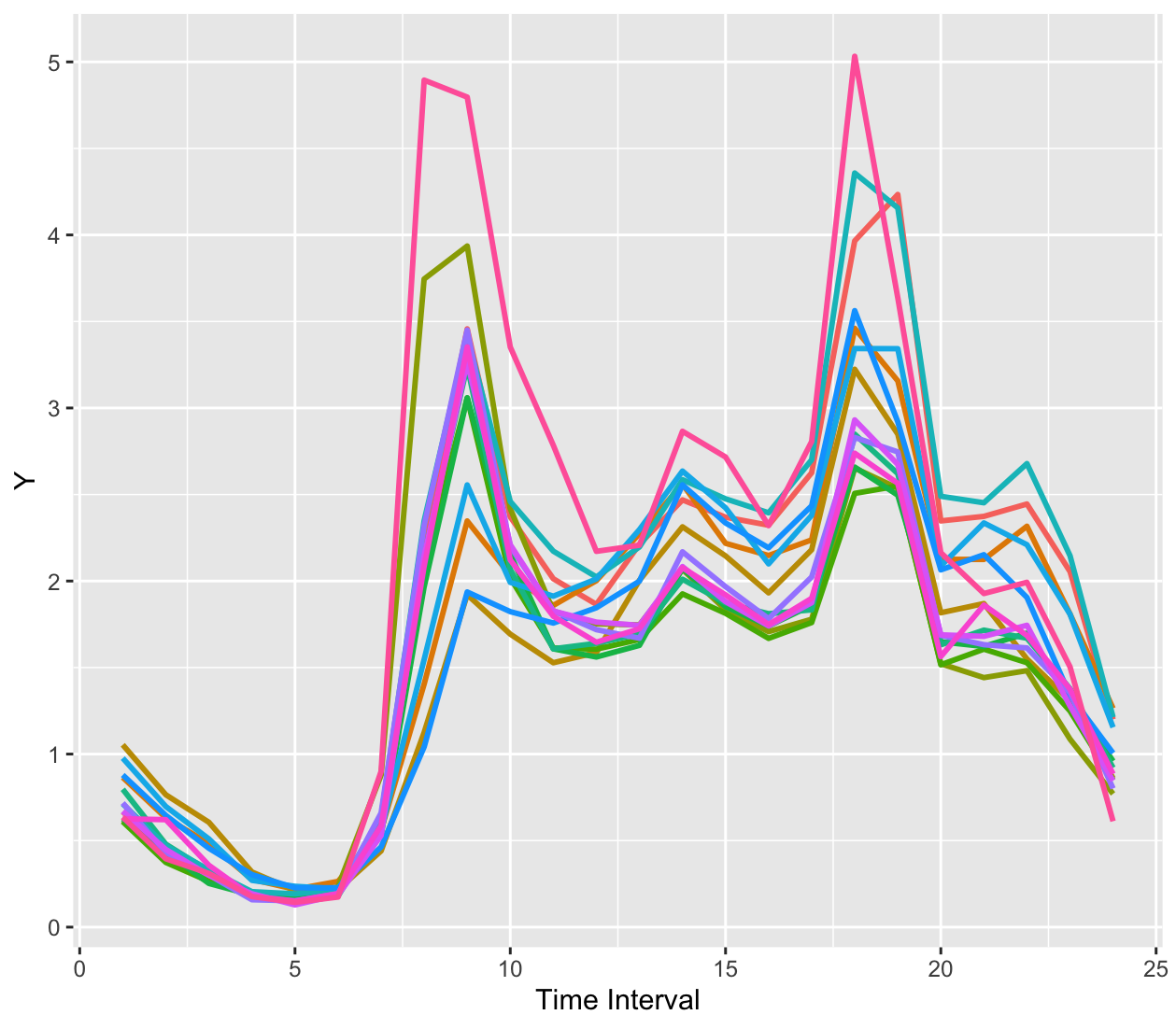} }
\subfigure[]{   \includegraphics[height=3.5cm, width=3.6cm]{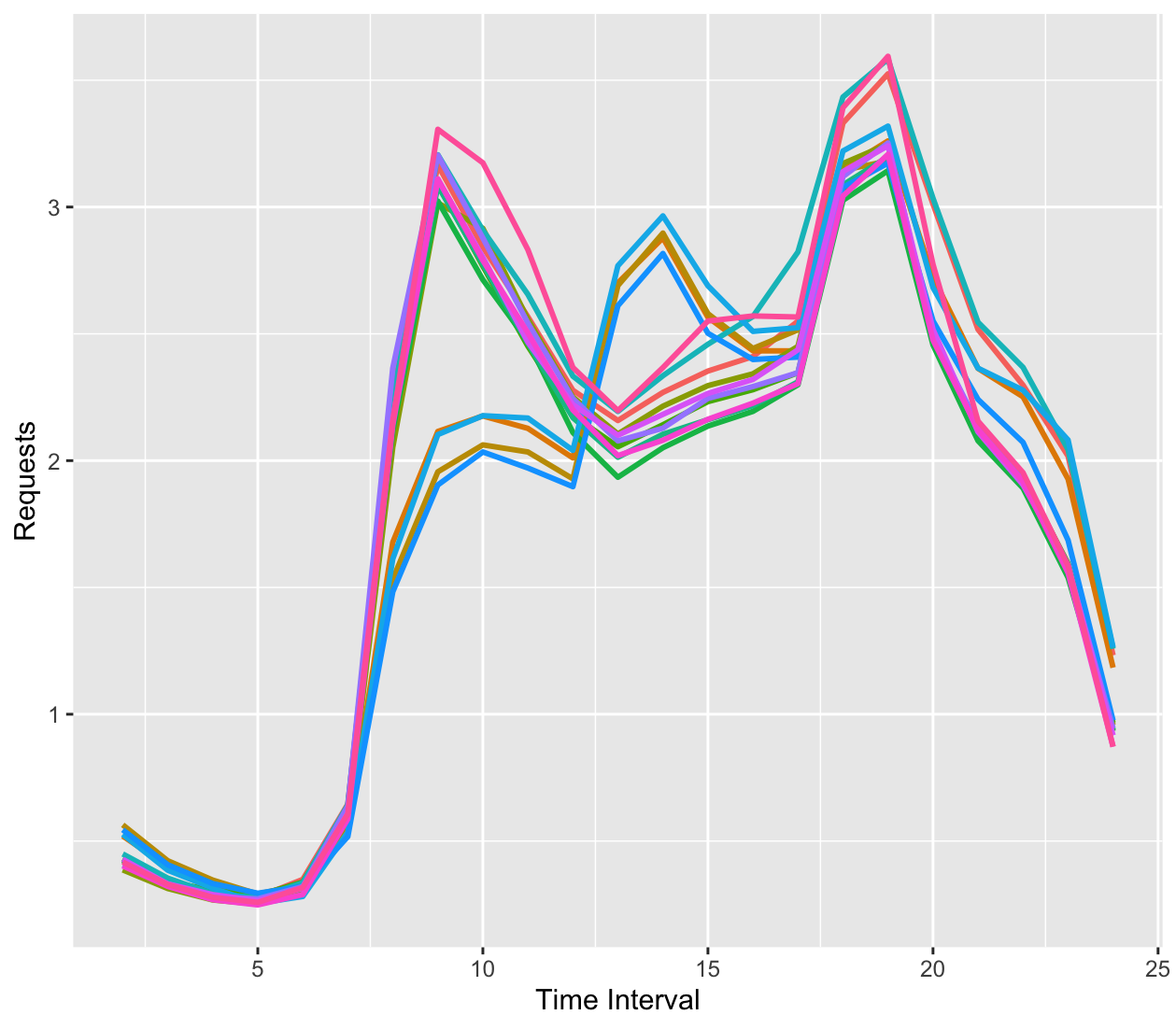} }
\subfigure[]{   \includegraphics[height=3.5cm, width=3.6cm]{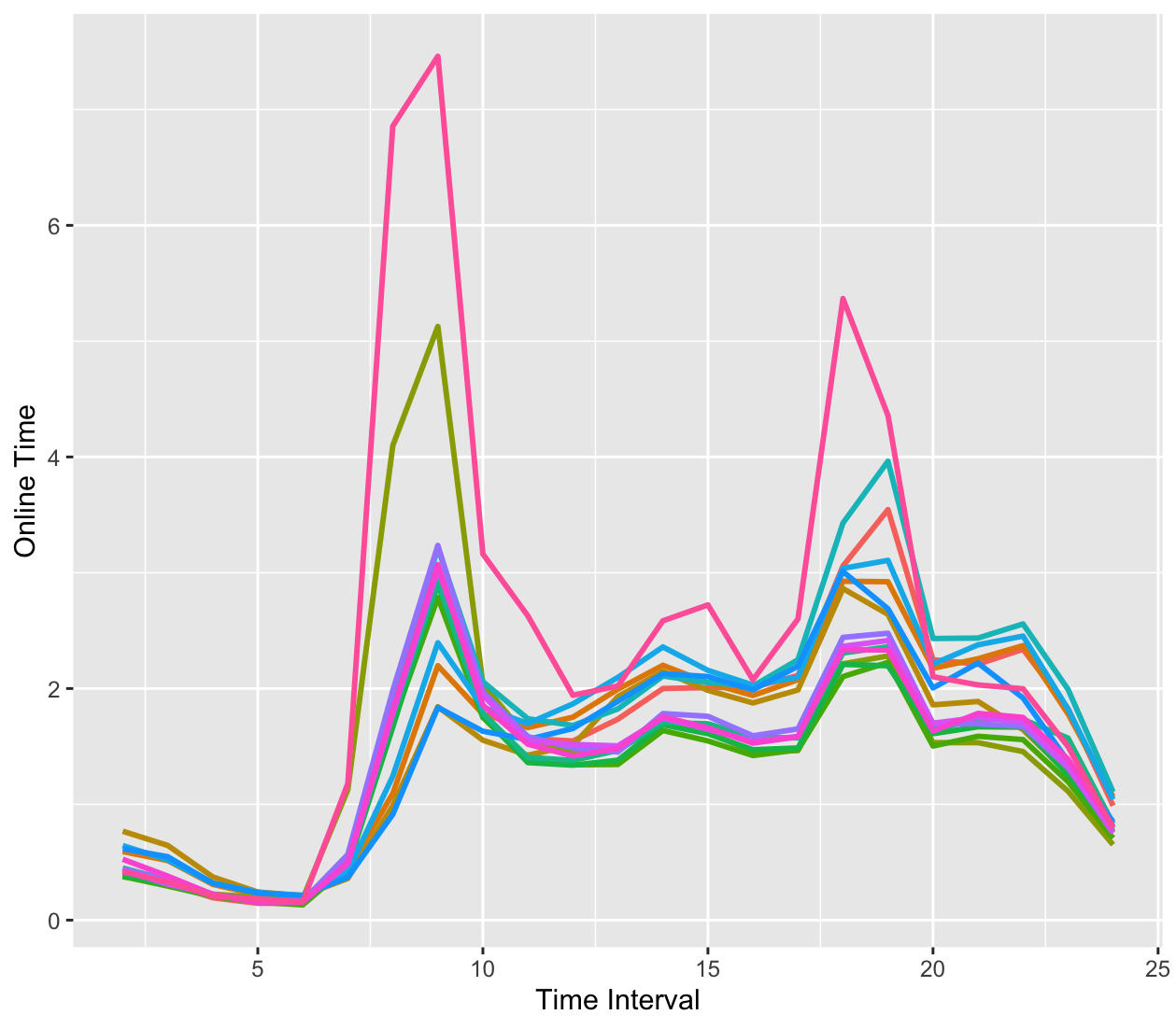} }
\subfigure[]{   \includegraphics[height=3.5cm, width=3.6cm]{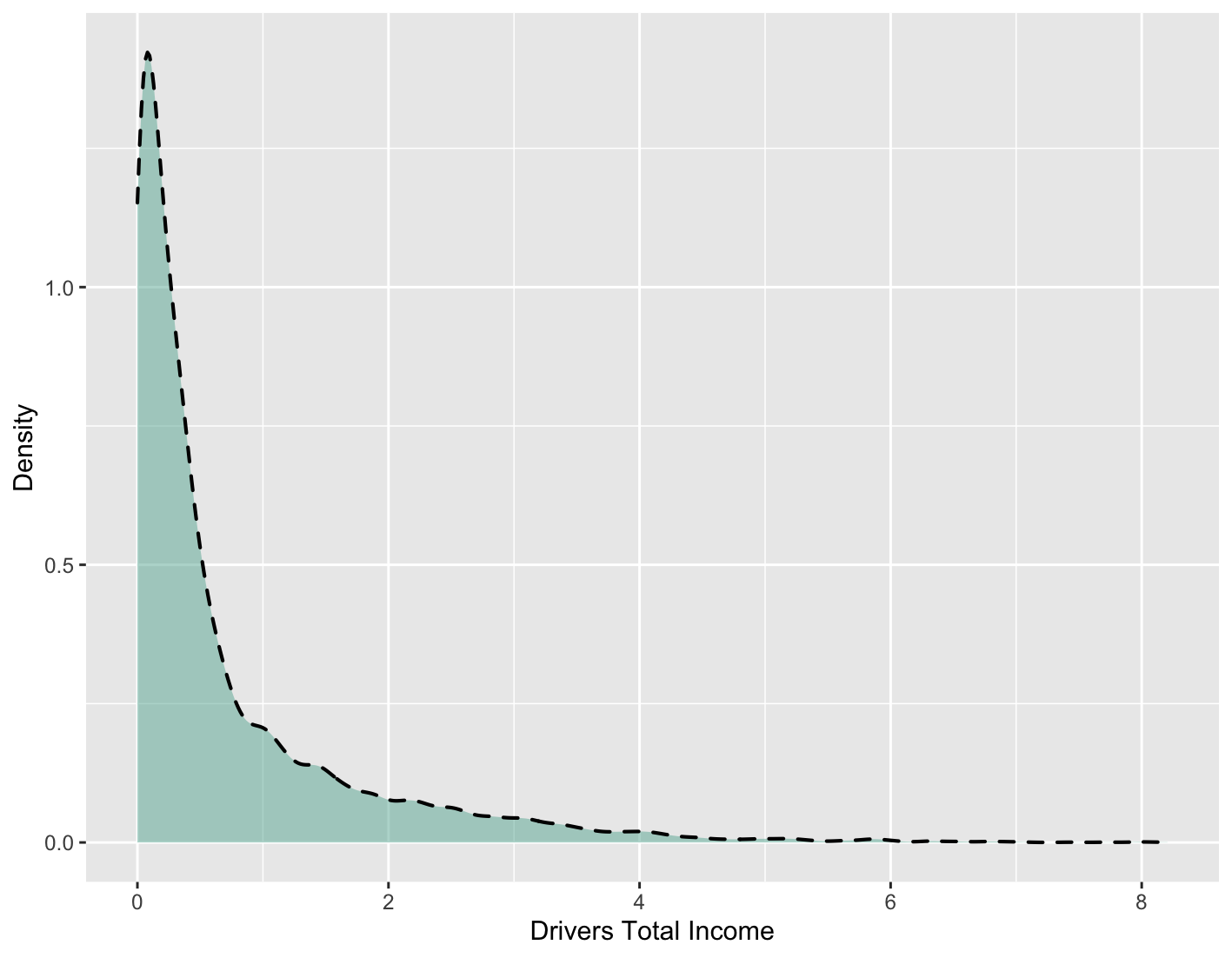} }
	\caption{\small Scaled drivers' total income (a), request (b) and drivers' online time (c)
		in the temporal dependent A/B experiment and the estimated density of drivers' total income (d) in the spatial temporal dependent A/B experiment.}
	\label{fig:AB_temporal}
\end{figure}

\begin{figure}[!t]
	\centering
	\includegraphics[height=4.5cm, width=15cm]{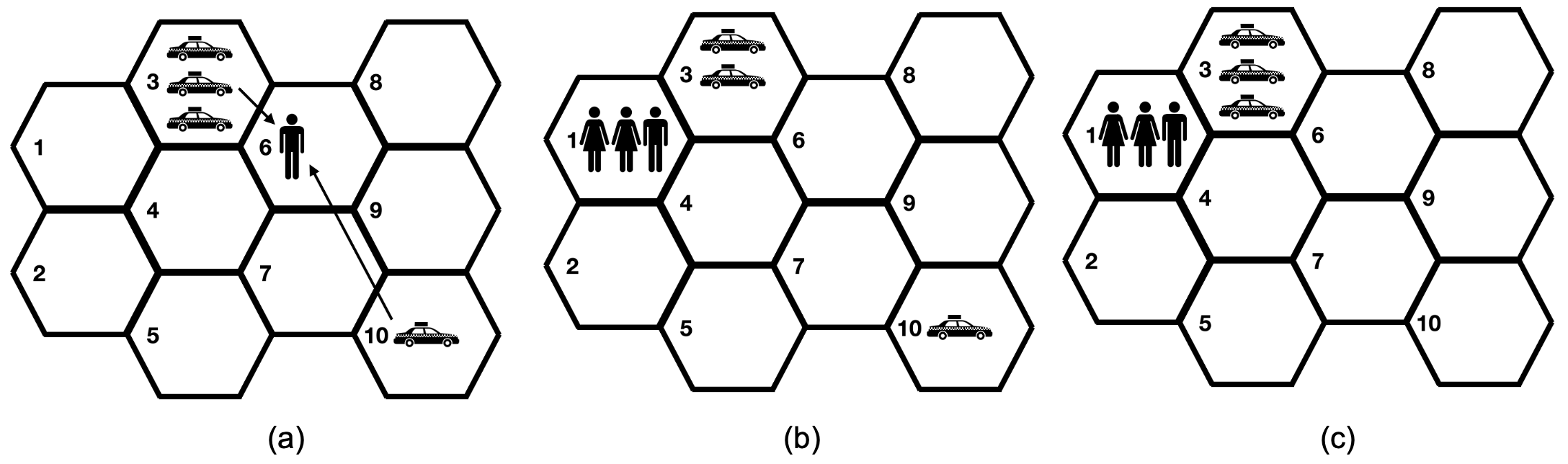}
	\caption{\small  This example illustrates the temporal interference effect in ridesharing, where assigning different drivers to pick up a passenger significantly impacts future ride requests. (a) A city with 10 regions has a passenger in region 6 needing a ride, with three drivers in region 3 and one in region 10. Two actions are possible: assigning a driver from region 3 or region 10. (b) Assigning a driver from region 3 might result in an unmatched future request due to the driver in region 10 being too far from region 1. (c) Assigning the driver in region 10 could lead to all future ride requests being matched, preserving all three drivers in region 3.  }
	
	\label{fig:interference}
\end{figure}

We focus on answering three key questions in these datasets:


(Q1) How can we quantify treatment effects across various quantile levels for the time-dependent A/B experiment data in order to gain a comprehensive understanding of the new policy's effects within the city?


(Q2) How to evaluate the quantile treatment effects  for the above spatiatemporal dependent experiment data?

(Q3) How to determine whether or not to replace the old policy with the new one?

\noindent 
These questions  drive the methodological development outlined in Sections 3 and 4.

\section{Testing CQTE in temporal dependent experiments}
\label{sec:temporal design}
In this section, we explicitly state the test hypotheses for our first research question (Q1) and explore the primary challenge encountered in experiments exhibiting temporal dependence. Subsequently, we detail the key technical assumptions that enable the cumulative quantile treatment effect (CQTE) to be equivalent to the sum of individual CQTEs. Finally, to address the third research question (Q3), we present the proposed estimation and testing strategies for our investigation. 

\subsection{CQTE, test hypotheses and assumptions}\label{sec:CQTE}

We consider the temporal alternation design with a sequence of treatments over time. Specifically, we divide each day into $m$   non-overlapping intervals. 
The platform can implement either one of the two policies at each time interval. For any $t\ge 1$, let $A_{t}$ denote the policy implemented at the $t$th time interval where $A_t = 1$ represents exposure to the new policy
and $A_t = 0$ represents exposure to the old policy. Let $S_t$ and $Y_t$ denote the state (e.g., the supply and demand) and the outcome at time $t$, respectively.

To formulate our problem, we adopt a potential outcome framework \citep{rubin2005causal}. Specifically, we define $\bar{a}_t=(a_1,\ldots,a_{t})^\top\in \{0,1\}^{t}$ as the treatment history up to time $t$. We also define $S_{t}^*(\bar{a}_{t-1})$ and $Y_{t}^*(\bar{a}_{t})$ as the counterfactual state and counterfactual outcome, respectively, that would have occurred had the platform followed the treatment history $\bar{a}_{t}$. 
Our primary interest lies in quantifying the difference between the $\tau$th quantile of the cumulative outcomes under the new policy and that under the old policy, denoted as the quantile treatment effect (QTE):
\begin{eqnarray*}
\textrm{QTE}=Q_\tau \Big(\sum_{t=1}^m Y_t^*(\bm{1}_{t})\Big)  -  Q_\tau\Big(\sum_{t=1}^m Y_t^*( \bm{0}_{t} )\Big), 
\end{eqnarray*}
where $\bm{1}_{t}$ and $\bm{0}_{t}$ are vectors of 1s and 0s of length $t$, respectively, and $Q_{\tau}
(\cdot)$ denotes the quantile function at the $\tau$th level.

However, learning such an unconditional dynamic QTE from our experimental dataset is highly challenging. Remember that in our A/B experiment, the old and new policies are assigned alternately over the $m$ time intervals. Nevertheless, the target policy we aim to evaluate corresponds to the global policy, which allocates the new or old policy globally throughout each day. This leads to an off-policy setting where the target policy differs from the behavior policy that generates the data. 
Existing off-policy quantile evaluation methods based on inverse probability weighting, such as those presented by \cite{wang2018quantile}, are inefficient in our setting with a moderately large $m$. Off-policy evaluation (OPE) methods, including  \citet{shi2020statistical}, \citet{liao2021off}, and \citet{kallus2022efficiently}, are semiparametrically efficient\footnote{\citet{shi2020statistical} and \citet{liao2021off} proposed to use the direct method based on linear sieves or kernels. However, the resulting estimators are semiparametrically efficient as well.} in long-horizon settings. Despite this, these methods primarily focus on the mean return, making it difficult to adapt them for quantile evaluation due to the nonlinear quantile function $Q_{\tau}$.
To illustrate, note that there is no guarantee that $Q_\tau  (\sum_{t=1}^m Y_t^*(\bm{1}_{t}))  -  Q_\tau(\sum_{t=1}^m Y_t^*( \bm{0}_{t} ))=\sum_{t=1}^m\Big[ Q_\tau (Y_t^*(\bm{1}_{t}))  -  Q_\tau(Y_t^*( \bm{0}_{t} ))\Big]$  in general. This observation motivates us to seek an alternative definition for QTE.


Second, let $\mathcal{E}_t$ represent the set of features (e.g., extreme weather events) that have an impact on the outcomes up to time $t$, but are not influenced by the treatment history. This means that for any treatment history $\bar{a}_t$, the potential outcome of these features remains the same, i.e., $\mathcal{E}_t^*(\bar{a}_t)=\mathcal{E}_t$. By definition, $S_1$ is an element of $\mathcal{E}_t$, which ensures that $\mathcal{E}_t$ is non-empty for any $t$. 
We introduce  CQTE as follows:
\begin{eqnarray}\label{eqn:QTE}
\textrm{CQTE}_\tau=  Q_\tau  (  \sum_{t=1}^m Y_t^*( \bm{1}_{t} ) |\mathcal{E}_m )  -  Q_\tau  ( \sum_{t=1}^m Y_t^*( \bm{0}_{t} )  |\mathcal{E}_m  ) .		
\end{eqnarray}
The CQTE is a reasonable measure because the set of conditioning variables remains consistent under both new and old policies. When $m=1$, this definition reduces to the one used in single-stage decision-making, as discussed in previous literature, such as \cite{chernozhukov2006instrumental}.



Conditioning has several benefits. Firstly, it offers a more convenient way to estimate the dynamic Quantile Treatment Effect (QTE) by aggregating individual QTEs over time. This approach simplifies the estimation process and reduces computational complexity. Secondly, conditioning can also help to reduce the variance of the resulting QTE estimator by removing the need to account for variability in the relevant characteristics. By conditioning on certain variables, researchers can effectively control for confounding factors and produce more accurate estimates of treatment effects. 

Third, we introduce the concept of Summed Conditional Quantile Treatment Effects (SCQTE), which represents the sum of individual Conditional Quantile Treatment Effects (CQTE) over time. The SCQTE is defined as follows:
\begin{eqnarray*}
\textrm{SCQTE}_\tau=  \sum_{t=1}^m Q_\tau (   Y_t^*( \bm{1}_{t}  ) |\mathcal{E}_t )  - \sum_{t=1}^m  Q_\tau ( Y_t^*( \bm{0}_{t} ) |\mathcal{E}_t ).		
\end{eqnarray*}
Compared to CQTE, SCQTE is easier to learn from observed data. For example, one can fit a quantile regression model at each stage, estimate individual CQTE values, and then sum these estimators together. 
Although the quantile function is not additive, we demonstrate in the following proposition that SCQTE is equal to CQTE under specific modeling assumptions.

\begin{prop}\label{prop0}
Suppose that for any time point $t$, $Y_t^*(\bar{a}_t)$ follows the structural quantile model $ Y_t^*(\bar{a}_t) = \phi_{t} ( \mathcal{E}_t, \bar{a}_t, U )$ for a specific deterministic function $\phi_t$ and a uniformly distributed random variable $U \stackrel{d}{\sim} \textrm{Unif}(0,1)$, which is independent of $\{\mathcal{E}_t\}_t$. Furthermore, assume that $\phi_{t}(\mathcal{E}_t , \bm{1}_t, \tau)$ and $\phi_{t}(\mathcal{E}_t, \bm{0}_t, \tau)$ are strictly increasing functions of $\tau$ for any $\mathcal{E}_t$. Under these conditions, we find that $\textrm{CQTE}_\tau = \textrm{SCQTE}_\tau$. 
\end{prop}





Proposition \ref{prop0} establishes the equivalence between CQTE and SCQTE  and serves as a fundamental building block for our proposal. It allows us to focus on SCQTE, which is a simplified version of CQTE. This simplification greatly facilitates the estimation and inference procedures that follow, which rely on fitting a quantile regression model at each time point to learn the SCQTE. For more details, see Sections \ref{sec:TQVCDP models} and \ref{subsec:estimation_temporal}.
Moreover, the proposed model in Proposition \ref{prop0} is related to the structural quantile model in the quantile regression literature \citep{chernozhukov2005iv,chernozhukov2006instrumental}. These models assume that, conditional on the covariate $X=x$, the potential outcome $Y^*(a)=q(a, x, U)$ for $a=0, 1$ and $U \sim U(0, 1)$, where $q(d, x, \tau)$ is strictly increasing in $\tau$. The uniformly distributed variable $U$ serves as a rank variable that characterizes the heterogeneity of the outcome across different quantile levels. Under the monotonicity constraint, the $\tau$th conditional quantile of $Y^*(a)|X=x$ can be shown to equal $q(a, x, \tau)$. 
Proposition \ref{prop0} motivates us to focus on testing the following hypotheses for each quantile level $\tau$: 
\begin{eqnarray}
\label{hypo:QTE}
H_0: \textrm{CQTE}_\tau \le 0 \quad \textrm{versus} \quad H_1:\textrm{CQTE}_\tau >0.
\end{eqnarray}
These hypotheses test whether the treatment effect at the $\tau$th quantile is non-negative or positive, respectively.

In this study, we utilize the consistency assumption (CA), sequential ignorability assumption (SRA), and positivity assumption (PA) to identify the causal estimand. These assumptions are frequently used in the dynamic treatment regime literature for learning optimal dynamic treatment policies \citep{gill2001causal}. 
The consistency assumption (CA) states that the potential state and outcome, given the observed data history, should align with the actual observed state and outcome. The sequential ignorability assumption (SRA) demands that the action be conditionally independent of all potential variables, given the past data history. In our application, the SRA is inherently satisfied as the policy is assigned according to the alternating-time-interval design, independent of data history. 
The positivity assumption (PA) necessitates that the probability of $\{A_t=1\}$, given the observed data history, must be strictly between zero and one for any $t\ge 1$. Under the alternating-time-interval design, this probability is equal to 0.5, which satisfies the PA automatically. 
It is essential to note that the combination of CA, SRA, and PA enables the consistent estimation of the potential outcome distribution using the observed data.



\subsection{VCDP models}
\label{sec:TQVCDP models}
Suppose that the  experiment is conducted over $n$ consecutive days. 
Let  $(S_{i,j},A_{i,j},Y_{i,j})$ be the state-treatment-outcome triplet measured at the $j$th time interval of the $i$th day for $i=1, \ldots, n$ and $j=1, \ldots, m$. We assume that these triplets are independent across different days, but may be dependent within each day over time.  



We begin by introducing two varying coefficient decision process models, one for the outcome and the other for the state. The first model characterizes the conditional quantile of the outcome and is given by 
\begin{align}
\label{model:TQVCM DE}
Y_{i,t }= \beta_{0 }(t, U_i) + S_{i,t}^\top\beta(t, U_i) + A_{i,t}\gamma(t, U_i) 
= Z_{i,t}^\top \theta(t, U_i), 
\end{align}
where 
$Z_{i,t} = (1, S_{i,t}^\top, A_{i,t})^\top\in \mathbb{R}^{d+2}$, 
  $\theta(t, U_i) = (\beta_{0}(t, U_i), \beta(t, U_i)^\top,  \gamma(t, U_i))^\top\in \mathbb{R}^{d+2}$ is a vector of time-varying coefficients, and $U_i\sim U(0,1)$ is the rank variable. 
Model \eqref{model:TQVCM DE} extends the idea of using rank variables to represent unobserved heterogeneity across different quantiles in a single-stage study to sequential decision making. 




The second model characterizes the conditional mean of the observed state variables and is given by:
\begin{align}
\label{model:TQVCM IE}
S_{i,t+1} =\phi_{0 }(t)+\Phi (t) S_{i,t}+A_{i,t} \Gamma(t) + E_{i}(t+1)
=  \Theta(t)Z_{i,t} + E_{i} (t +  1),
\end{align}
where $\phi_{0}(t)$ and $\Gamma(t)$ are $d$-dimensional vectors, $\Phi(t)$ is a $d\times d$ matrix of autoregressive coefficients, and $\Theta(t) = [\phi_{0}(t) ~~ \Phi(t) ~~ \Gamma(t)]$ is a $d\times (d+2)$ coefficient matrix. The term $E_i(t+1)$ is a random error term whose conditional mean given $Z_{i,t}$ equals zero. In addition, $\{E_i(t)\}_{t}$ are independent over time.
Therefore, the conditional expectation of $S_{i,t+1}$ given $Z_{i,t}$ is:
$
\mathbb{E}( S_{i,t+1 }  | Z_{i,t} ) = \Theta(t)Z_{i,t}.
$

It is worth noting that models \eqref{model:TQVCM DE} and \eqref{model:TQVCM IE} belong to the class of varying-coefficient regression models. The existing literature on this topic mainly focuses on estimating the relationships between scalar predictors and scalar responses \citep{sherwood2016partially}, or between scalar predictors and functional responses \citep{zhang2021high}, or between longitudinal predictors and responses \citep{wang2009quantile}. However, to the best of our knowledge, none of these works have utilized varying-coefficient regression models for policy evaluation in sequential decision making.  

Furthermore, the temporal independence between $E_{i}(t+1)$'s implies that the state vector satisfies the Markov property, i.e., $S_{i,t+1}$ is independent of the past data history given $(S_{i,t}, A_{i,t})$ for any $i$ and $t$. However, the immediate outcomes are conditionally dependent due to the existence of the rank variable $U$. Consequently, the resulting data generating process does not fall under the classical Markov decision processes \citep[MDPs, see e.g.,][]{puterman2014markov}.
Moreover, models \eqref{model:TQVCM DE} and \eqref{model:TQVCM IE} are valid when the potential outcomes satisfy similar assumptions in the quantile varying coefficient models. Please refer to   \eqref{model:TQVCM_outcome_DE} and \eqref{model:TQVCM_outcome_IE} in the supplementary material for more details. 


If the residual $E_i(t+1)$ is independent of the treatment history and $U_i$, we can define $\mathcal{E}_t$ as $\{S_1,E(1),\cdots, E(t)\}$. Under this condition, the assumptions in Proposition \ref{prop0} are satisfied. Hence, under the proposed VCDP models, CQTE is equivalent to SCQTE.
Next, we introduce the following function: 
\begin{eqnarray*}
\phi_{t} ( \mathcal{E}_{t}, \bar a_{t}, U  )
&=&
\beta_{0 }(t, U)  + a_{t} \gamma(t, U) + \beta(t, U) ^\top 
\Big( \sum_{k=1}^{t-1} \Big\{ \prod_{l=k+1}^{t-1} \Phi(l) [ \phi_{0}(k) + a_{k} \Gamma(k)  ]  \Big\}  \Big)  \\
&&+
\prod_{l=1}^{t-1}  \Phi(l)  S_{1} + \sum_{k=2}^t [ \prod_{l=k}^{t-1} \Phi(l) E(k) ],
\end{eqnarray*}
The subsequent proposition offers a closed-form formula for CQTE.

\begin{prop}
\label{prop:QTE_form}
Assuming that CA and Equations \eqref{model:TQVCM_outcome_DE} and \eqref{model:TQVCM_outcome_IE} in the supplementary material hold,  
$U$ is independent of $\{\mathcal{E}_t\}_t$,  and  
$\phi_{t} ( \mathcal{E}_{t}, \bm{1}_t, \tau  )$ and $\phi_{t} ( \mathcal{E}_{t}, \bm{0}_t, \tau  )$ are strictly increasing in $\tau$ for any $\mathcal{E}_{t}$,  
then we have 
\begin{eqnarray}\label{eqn:QTEtau}
	\text{CQTE}_\tau  = \text{SCQTE}_\tau = 
	\sum_{t=1}^m  \gamma(t, \tau)  + \sum_{t=2}^m \beta(t, \tau)^\top  \left\{ \sum_{k=1}^{t-1} \left[\prod_{l=k+1}^{t-1} \Phi(l)\right]  \Gamma(k)  \right\},
\end{eqnarray}
where 
the product $\prod_{l=k+1}^{t-1} \Phi(l) =1$ when $t -1  < k+1$.
\end{prop}

Proposition \ref{prop:QTE_form} enables us to estimate CQTE through SCQTE under certain assumptions. Among these, the monotonicity assumption can be satisfied under various conditions. For example, it holds when $\beta_{0}(t, \tau)$, $\gamma(t, \tau)$, and all elements in $\beta(t, \tau)$ are strictly increasing in $\tau$, and $\phi_0(t)$, all elements in $\Phi(t)$, and $\Gamma(t)$ are positive. Additionally, when $\Gamma(t)=0$ and $\phi_0(t)=0$ for any $t$, it suffices to require $\gamma(t, \tau)$ and $\beta_{0}(t, \tau)$ to be strictly increasing in $\tau$.

To evaluate policy value, we need to estimate the model parameters $\beta$, $\gamma$, $\Phi$, and $\Gamma$. Notice that under the conditions of Proposition \ref{prop:QTE_form}, we have that: 
\begin{align}
\label{model:TQVCM DE tau}
Y_{i,t }= \beta_{0 }(t, \tau) + S_{i,t}^\top\beta(t, \tau) + A_{i,t}\gamma(t, \tau) + e_{i}(t, \tau)
= Z_{i,t}^\top \theta(t, \tau) + e_i(t, \tau),	
\end{align}
where $e_i(t, \tau)=Z_{i,t}^\top [\theta(t, U)-\theta(t, \tau)]$, and its conditional $\tau$-th quantile given $Z_{i,t}$ equals zero. Therefore, we can employ ordinary quantile regression to learn $\beta$ and $\gamma$. Meanwhile, since the residuals $E_i(t)$s are independent over time, ordinary least-squares regression is applicable to the state regression model to estimate $\Phi$ and $\Gamma$. We detail our estimating procedure in the next section. 

\subsection{Estimation and inference procedures}
\label{subsec:estimation_temporal}
In this subsection, we outline the procedures for estimating and testing CQTE based on the results in Proposition \ref{prop:QTE_form}. We first estimate the regression coefficients in models \eqref{model:TQVCM DE tau} and \eqref{model:TQVCM IE}. We then plug these estimates into \eqref{eqn:QTEtau} to estimate CQTE. Finally, we develop a bootstrap-assisted procedure to test CQTE.


Let $S_{i,t+1}^{(\nu)}$, $\phi^{(\nu)}_{0}(t)$, and $\Gamma^{(\nu)}(t)$ denote the $\nu$-th entries of $S_{i,t+1}$, $\phi_{0}(t)$, and $\Gamma(t)$, respectively. Let $\Phi^{(\nu)}(t)$ and $\Theta^{(\nu)}(t)$ denote the $\nu$-th rows of $\Phi(t)$ and $\Theta(t)$, respectively. It follows from \eqref{model:TQVCM IE} that:  
$$
S^{(\nu)}_{i,t+1} =\phi^{(\nu)}_{0}(t)+ S_{i,t}^\top \Phi^{(\nu)}(t) 
+A_{i,t} \Gamma^{(\nu)}(t) + E^{(\nu)}_{i}(t+1)
=  Z_{i,t}^\top \Theta^{(\nu)}(t) + E^{(\nu)}_{i}(t+1).
$$

We propose a two-step procedure to estimate $\theta(t, \tau)$ and $\Theta(t)$. In the first step, we minimize the following functions: 
\begin{eqnarray}
\label{eq:theta_estimator}
\widehat{\theta}( t, \tau)&=& \argmin \sum_i \rho_\tau ( Y_{it} -  Z_{i,t}^\top \theta(t, \tau) ), \quad  \text{for~} t=1, \dots, m,  \\
\label{eq:Theta_estimator}
{\widehat \Theta}^{(\nu)}(t)
&=& \argmin \sum_i  [ S_{i, t+1}^{(\nu)} -  Z_{i,t}^\top \Theta^{(\nu)}(t) ]^2,
\quad \text{for~} \nu =1, \dots, d, \quad t=1, \dots, m-1. 
\end{eqnarray}	
These one-step estimates can be computed easily but suffer from large variances as they rely solely on observations at time $t$. In the second step, we employ kernel smoothing to reduce the variances of these initial estimators and identify weak signals \citep{zhu2014spatially}. Specifically, for a given kernel function $K(\cdot)$, the second-step estimators $\widetilde{\theta}(t, \tau)$ and ${\widetilde \Theta}_\tau^{(\nu)}(t)$ are defined as:  
\begin{eqnarray}
\label{eq:theta_estimator_refined}
\widetilde{\theta}(t, \tau) &=&\sum_{j =1}^m \omega_{j,h}(t) \widehat{\theta} (j, \tau), \quad \text{for~} t=1, \dots, m \\
\label{eq:Theta_estimator_refined}
\widetilde{\Theta}^{(\nu)}(t)
&=& \sum_{j =1}^m \omega_{j,h}(t) {\widehat \Theta}^{(\nu)}( j),
\quad \text{for~} \nu =1, \dots, d, \quad t=1, \dots, m-1,
\end{eqnarray}	
where $\omega_{j,h}(t)=K((j -t )/(mh))/\sum_{k=1}^m K((k -t)/(mh))$ is 
the weight function and $h$ denotes the kernel bandwidth.
 The use of kernel smoothing allows us to estimate the varying coefficients $\theta_\tau(t)$ and $\Theta_\tau(t)$ for any real-valued $t$. Given $\widetilde{\theta}(t, \tau) $ and ${\widetilde \Theta}_\tau^{(\nu)}(t)$, we can compute the following CQTE estimator: 
\begin{eqnarray}
\label{eq:QTE_estimate}
\widehat{\textrm{CQTE} }_\tau=\sum_{t=1}^m  \widetilde \gamma(t, \tau)
+  \sum_{t=2}^m \widetilde \beta(t,\tau)^\top  \left\{ \sum_{k=1}^{t-1} \left[\prod_{l=k+1}^{t-1} \widetilde \Phi(l)\right]  \widetilde\Gamma(k)  \right\}.	
\end{eqnarray}	


To test \eqref{hypo:QTE}, we use the test statistic $T_{\tau}$, which is set to $\widehat{\textrm{CQTE} }_\tau$. Under the null hypothesis, $T_{\tau}$ is expected to be negative or close to zero. Therefore, we reject the null hypothesis for a large value of $\widehat{\textrm{CQTE}}_\tau$. However, deriving the limiting distribution of $T_{\tau}$ for large $m$ is complicated due to the complex dependence of $\widehat{\textrm{CQTE}}_\tau$ on the estimated model parameters. To address this issue, we use the bootstrap method to simulate the distribution of $\widehat{ \textrm{CQTE}}_\tau$ under the null hypothesis. Specifically, we modify the bootstrap method proposed by \cite{horowitz2018bootstrap} and adapt it to our setting as follows. \cite{horowitz2018bootstrap} proposed to resample the estimated residuals to infer the conditional quantile function in a nonparametric quantile regression model. In our case, to handle the dependence over time, we resample the entire error process (see Step 3 below for details).

The bootstrap method for $\widehat{\textrm{CQTE}}_\tau$ is implemented as follows: 
\begin{itemize}
\item 
Step 1. Compute the estimators $\widetilde{\theta}(t, \tau)$ and $\widetilde{\Theta}(t)$ in \eqref{eq:theta_estimator_refined} and \eqref{eq:Theta_estimator_refined}.
\item 
Step 2. Estimate the residuals by $\hat e_{i}(t, \tau) = Y_{i,t} - Z_{i, t}^\top \widetilde{\theta}(t, \tau)$ for $t=1, \dots, m$ and
$\widehat E_{i} (t+1) =S_{i, t+1} - \widetilde{\Theta}(t) Z_{i,t}$ for $t=1, \dots, m-1$.
\item

Step 3. For each $b$, generate i.i.d. random variables $\{ \tilde e_{i} (1, \tau), \dots, \tilde e_{i} (m, \tau) \}$ 
by randomly sampling $\{ \hat e_{i} (1, \tau), \dots, \hat e_{i} (m, \tau) \}$ with replacement. Similarly, generate $\{ \tilde E_{i} (2), \dots, \tilde E_{i} (m) \}$ 
by randomly sampling $\{ \hat E_{i} (2), \dots, \hat E_{i} (m) \}$ with replacement. Next, generate pseudo outcomes $\{\widehat{S}^b_{i,t}\}_{ i,t}$ and $\{\widehat{Y}^b_{ i,t}\}_{i,t}$ as follows,
\begin{eqnarray}
	\label{eqn:pseudo_outcomes}
	\widehat{S}^b_{ i,t+1}=\widetilde{\Theta}(t) \widehat{Z}^b_{ i,t}+ \tilde E_{i} (t+1)  \,\,\hbox{and}\,\,
	\widehat{Y}^b_{ i,t}=\widehat{Z}_{ i,t}^{b \top} \widetilde{\theta}(t, \tau) +  \tilde{e}_{i}(t, \tau).
\end{eqnarray}

\item 
Step 4. For each $b$, compute the bootstrap estimates $\widetilde{\theta}^b(t, \tau)$ and $\widetilde{\Theta}^b(t)$ according to 
equations \eqref{eq:theta_estimator}-\eqref{eq:Theta_estimator_refined} using the pseudo outcomes $\{(\widehat{S}^b_{i,t}, \widehat{Y}^b_{ i,t}):i,t\} $.

\item Step 5. For each $b$, compute the bootstrapped statistic $T_{\tau}^b=\widehat{\textrm{CQTE} }^b_\tau$.  

\item Step 6. 
Repeat Steps 3-5 $B$ times. Given a significance level $\alpha$, reject $H_0$ (see \ref{hypo:QTE})
if the statistic
$T_{\tau}$ exceeds the upper $\alpha$th empirical quantile of
$\{ T_{\tau}^b-T_{\tau} \}_{b=1}^B$. 
\end{itemize}


In the supplementary material, we present Theorem \ref{thm:bootstrap_consistency_QTE}, which rigorously establishes the consistency of the aforementioned bootstrap method. It's worth noting that the bootstrap consistency theory elaborated in \cite{horowitz2018bootstrap} isn't readily applicable to our context, where $m$ can increase along with the sample size. 


\section{Extension to spatiotemporal dependent experiments}
\label{sec:ST method}
In this section, we aim to address (Q2) and
expand upon the method proposed in Section \ref{sec:temporal design} to analyze data from spatiotemporal dependent experiments involving multiple non-overlapping regions receiving distinct treatments in a sequential manner over time. Let $r$ represent the number of these non-overlapping regions. As previously discussed, these experiments are not only subject to temporal interference effects but also exhibit spatial interference, whereby the policy implemented in one location may influence the outcomes in other locations. We begin by defining the test hypotheses, followed by an introduction to our models and the suggested procedures.

 

\subsection{Test hypotheses}\label{sec:STPOF}

 For the $\iota$-th region, we use $\bar{a}_{t,\iota}=(\bar{a}_{1,\iota},\ldots,\bar{a}_{t,\iota})^\top$ to denote its treatment history up to time $t$.
Let $\bar{a}_{t,[1:r]}=(\bar{a}_{t,1},\ldots,\bar{a}_{t,r})^\top$ represent the treatment history across all regions. Similarly, define $S_{t,\iota}^*(\bar{a}_{t-1,[1:r]})$ and $Y_{t,\iota}^*(\bar{a}_{t,[1:r]})$ as the potential observation and outcome for the $\iota$-th region, respectively. The set of potential observations at time $t$ is denoted as $S_{t,[1:r]}^*(\bar{a}_{t-1,[1:r]})$.

In the spatiotemporal context, our focus is on the cumulative quantile treatment effects, aggregated over all regions. Specifically, we define CQTE and SCQTE at the $\tau$-th quantile level as:
\begin{eqnarray*}
\textrm{CQTE}_{\tau st}= Q_\tau \Big( \sum_{\iota=1}^r\sum_{t =1}^{m} Y_{t,\iota}^*(\bm{1}_{t,[1:r]}) | \mathcal{E}_{m, [1:r]} \Big)
- Q_\tau\Big( \sum_{\iota=1}^r\sum_{t =1}^{m} Y_{t,\iota}^*(\bm{0}_{t,[1:r]}) | \mathcal{E}_{m, [1:r]} \Big ),  \\
\textrm{SCQTE}_{\tau st}=  \sum_{\iota=1}^r\sum_{t =1}^{m}   Q_\tau \Big(Y_{t,\iota}^*(\bm{1}_{t,[1:r]})  | \mathcal{E}_{t, [1:r]} \Big)
-  \sum_{\iota=1}^r\sum_{t =1}^{m}  Q_\tau  \Big(  Y_{t,\iota}^*(\bm{0}_{t,[1:r]}) | \mathcal{E}_{t, [1:r]} \Big),
\end{eqnarray*}
respectively,  
where $\mathcal{E}_{t, [1:r]}$ denotes the set of characteristics independent of the treatment history up to time $t$ across all regions.  
For a given quantile level $\tau$, our goal is to test whether a new policy outperforms   the old one as follows: 
\begin{eqnarray}
\label{hypo:QTE_st}
H_0: \textrm{CQTE}_{\tau st} \le 0 \quad \textrm{versus} \quad H_1: \textrm{CQTE}_{\tau st}>0.
\end{eqnarray}
Compared to the testing problem in \eqref{hypo:QTE}, \eqref{hypo:QTE_st} focuses on global treatment effects aggregated over time and regions. We assume the consistency assumption holds. Similar to Section \ref{sec:temporal design}, under the spatial alternating-time-interval design, one can 
show that the sequential ignorability assumption and the positivity assumption are automatically satisfied, ensuring that CQTE is identifiable from the observed data.

\subsection{Spatiotemporal VCDP models}  
\label{sec:STQVCDP}

Suppose that the experiment  
last for $n$ days, and   each day is divided into $m$ time intervals.
For $i=1, \ldots, n$, $t=1, \ldots, m$, and $\iota=1, \ldots, r$, let
$(S_{i,t,\iota}, A_{i,t,\iota}, Y_{i,t,\iota})$ represent the state-treatment-outcome triplet measured from  the $\iota$th region at the $t$-th time interval of  the $i$-th day.  
For each $\iota$,  $\mathcal{N}_\iota$ denotes the neighbouring regions of $\iota$.  
To model the quantiles of $  Y_{t,\iota}$ and $S_{t,\iota}$,  we extend the two VCDP models in Section \ref{sec:temporal design} to two spatialtemporal VCDP (STVCDP) models in this section.

The first STVCDP model describes the quantile structure of the outcome. It assumes the following form: 
\begin{align}
 \label{model:STQVCM DE}
Y_{i,t,\iota} &= \beta_{0 } (t,\iota, U_i) + S_{i,t,\iota}^\top\beta (t,\iota, U_i) + A_{i,t,\iota}\gamma_{1}(t,\iota, U_i)+ \bar{A}_{i,t,\mathcal{N}_\iota}\gamma_{2 }(t,\iota, U_i)   \\
	&= Z_{i,t,\iota}^\top \theta (t,\iota, U_i),   \nonumber 
\end{align}
where    $\bar{A}_{i,t,\mathcal{N}_\iota}$ denotes the average of $\{A_{i,t,k}\}_{k\in \mathcal{N}_\iota}$,    $Z_{i,t,\iota} = (1, S_{i,t,\iota}^\top, A_{i,t,\iota}, \bar{A}_{i,t,\mathcal{N}_\iota})^\top$,   and 
$\theta (t,\iota, U_i) = (\beta_{0 } (t,\iota, U_i), \beta(t,\iota, U_i)^\top, \gamma_{1}(t,\iota, U_i), \gamma_{2 }(t,\iota, U_i))^\top$. 
Model \eqref{model:STQVCM DE} is based on two key assumptions. Firstly, it is assumed that the effect of treatments in other regions on the conditional quantile of $Y_{i,t,\iota}$ is limited to those of its neighboring regions, as long as each experimental region is large enough. This is because drivers can only travel between neighboring regions in one time unit, meaning that treatments in non-neighboring regions are not expected to impact $Y_{i,t,\iota}$.
Secondly, it is assumed that the influence of treatments in neighboring regions on the conditional quantile of $Y_{i,t,\iota}$ is only through the mean of the treatments. This is a common mean-field assumption used to model spillover effects \citep[e.g.,][]{hudgens_toward_2008,shi2022multi} and can be tested using observed data \citep{shi2022multi}.


The second STVCDP model models the conditional distribution of the next state given the current state-action pair as follows: 
\begin{align}\label{model:STQVCDP IE}
\begin{split}
	S_{i,t+1,\iota} &= \phi_{0}(t,\iota) + \Phi(t,\iota)S_{i,t,\iota} + A_{i,t,\iota}\Gamma_{1}(t,\iota) + \bar{A}_{i,t,\mathcal{N}_\iota}\Gamma_{2 }(t,\iota) + E_{i} (  t+1,\iota )\\
	&= \Theta(t,\iota)Z_{i,t,\iota} + E_{i} (  t+1,\iota ),
\end{split}    
\end{align}
where 
$\Theta(t,\iota)=[\phi_{0 }(t,\iota),\Phi (t,\iota),\Gamma_{1}(t,\iota), \Gamma_{2}(t,\iota)]\in\R^{d\times (d+3)}$
and
$\Phi(t,\iota)$ is a $d\times d$ matrix of autoregressive coefficients. 
The conditional mean of each entry in the error process $E_{i} (t, \iota)$ given $Z_{i,t,\iota}$ is zero. 
The error process is required to be independent over time, although it may be dependent across different locations. Additionally, the varying coefficients are required to be smooth over the entire spatial domain, which will help to reduce the variances of the model estimators and improve the accuracy of the CQTE estimator. The models \eqref{model:STQVCM DE} and \eqref{model:STQVCDP IE} hold under the assumption that the potential outcomes satisfy the quantile varying coefficient models, as described in the supplementary material (models \eqref{model:STQVCM_outcome_DE} and \eqref{model:STQVCM_outcome_IE}).

The following proposition provides a closed-form expression for $\textrm{CQTE}_{\tau st}$ and proves that $\textrm{CQTE}_{\tau st}=\textrm{SCQTE}_{\tau st}$.
Let  
\begin{eqnarray*}
\phi_{t,\iota} ( \mathcal{E}_{t,\iota}, \bar a_{t, [1:r]}, U  )
&=&
\beta_{0 }(t,\iota, U)  + a_{t,\iota} \gamma_1(t, \iota, U) +
\bar a_{t,N_\iota} \gamma_2(t, \iota, U)  \\
&+&  \beta(t, \iota, U) ^\top 
\Big( \sum_{k=1}^{t-1} \Big\{ \prod_{l=k+1}^{t-1} \Phi(l,\iota) [ \phi_{0}(k,\iota) + a_{k, \iota} \Gamma_1(k,\iota) + a_{k, N_\iota} \Gamma_2(k,\iota)  ]  \Big\}  \Big)  \\
&+&
\prod_{l=1}^{t-1}  \Phi(l,\iota)  S_{1, \iota} + \sum_{k=2}^t [ \prod_{l=k}^{t-1} \Phi(l,\iota) E(k, \iota) ],  
\end{eqnarray*}
where $\mathcal{E}_{t,\iota} =  \{S_{1,\iota}, E (2, \iota), \dots, E (t, \iota) \} $ and the product $\prod_{l=k+1}^{t-1} \Phi(j,\iota) =1$ when $t -1  < k+1$.  

\begin{prop}
\label{prop:QTE_form_spatial}
Suppose that CA and the conditions in equations \eqref{model:STQVCM_outcome_DE} and \eqref{model:STQVCM_outcome_IE} of the supplementary material hold, and that $U$ is independent of the collection of error processes $\{\mathcal{E}_{t,\iota}\}_{t,\iota}$. Furthermore, assume that the functions $\phi_{t,\iota} ( e_{t,\iota}, \bm{1}_{t, [1:r]}, \tau)$ and $\phi_{t,\iota} ( e_{t,\iota}, \bm{0}_{t, [1:r]}, \tau)$ are strictly increasing in $e_{t,\iota}$. Then, we have 
\begin{eqnarray*}
	\textrm{CQTE}_{\tau st}= \textrm{SCQTE}_{\tau st} 
	&=&
	\sum_{\iota=1}^r\sum_{t=1}^m\{\gamma_{1 }(t,\iota,\tau)+\gamma_{2 }(t,\iota, \tau)\}  \\
	&+&
	\sum_{\iota=1}^r\sum_{t=2}^m \beta(t,\iota, \tau)^\top \left\{ \sum_{k=1}^{t-1} \left[\prod_{j=k+1}^{t-1} \Phi(j,\iota)\right] [\Gamma_{1 }(k,\iota)+\Gamma_{2  }(k,\iota)] \right\}.
\end{eqnarray*}

\end{prop}

Proposition \ref{prop:QTE_form_spatial} provides a foundation for constructing a plug-in estimator for $\textrm{CQTE}_{\tau st}$. This forms the basis of the proposed inference procedure, which we discuss in more detail in the next section. Additionally, from models \eqref{model:STQVCM DE} and \eqref{model:STQVCDP IE}, we can obtain the expression for $Y_{i,t,\iota}$ as: 
$$
Y_{i,t,\iota} = \beta_{0 } (t,\iota, \tau) + S_{i,t,\iota}^\top\beta (t,\iota, \tau) + A_{i,t,\iota}\gamma_{1}(t,\iota, \tau)+ \bar{A}_{i,t,\mathcal{N}_\iota}\gamma_{2 }(t,\iota, \tau) +  e_{i} (t, \iota, \tau),
$$
where  $e_{i} (t, \iota, \tau)$ is the residual term, defined as $Z_{i, t, \iota}^\top [\theta(t,\iota,U)-\theta(i,\iota,\tau)]$, and its conditional $\tau$-th quantile given $Z_{i, t, \iota}$ is equal to zero.
It is worth mentioning that these models can be further extended to incorporate the effects of states from neighboring regions on the immediate outcome by including another mean-field term $\Phi_{2}(t,\iota) \bar{S}_{i,t,\mathcal{N}_\iota}$ where $\bar{S}_{i,t,\mathcal{N}_\iota} = \sum_{\iota^\prime \in \mathcal{N}_\iota} {S}_{i,t,\iota^\prime}/\mathcal{N}_\iota $. In this case, the closed-form expression for $\textrm{CQTE}_{\tau st}$ can be similarly derived.

\subsection{Estimation and inference procedures} 
\label{sec:inference_st}

In this subsection, we outline the estimation and testing procedures for $\textrm{CQTE}_{\tau st} $.

Firstly, we calculate  raw estimators of the unknown
coefficients in the two STVCDP models. For each region $\iota$, we employ standard quantile regression and linear regression as shown in \eqref{eq:theta_estimator} and \eqref{eq:Theta_estimator}
to the data subsets $\{(Z_{i,t,\iota}, Y_{i,t,\iota})\}_{i,t}$ and $\{(Z_{i,t,\iota}, S_{i,t+1,\iota})\}_{i,t}$ to obtain the initial estimators
$ \widehat{\theta} (t, \iota, \tau) $ and $ \widehat{\Theta} (t, \iota) $ for $\theta (t, \iota, \tau) $ and $\Theta(t, \iota)$, respectively. Next, we apply kernel smoothing techniques as illustrated in \eqref{eq:theta_estimator_refined} and \eqref{eq:Theta_estimator_refined} to refine these initial estimators over time. We denote the resulting estimators as $ \widetilde{\theta}^0 (t, \iota, \tau) $ and $ \widetilde{\Theta}^0 (t, \iota) $.

Secondly, we further refine these raw estimators by employing kernel smoothing to borrow information across space. 
Specifically, we define 
\begin{eqnarray*} 
\widetilde\theta(t,\iota, \tau)              =   \sum_{\ell=1}^r\kappa_{\ell,h_{st}}(\iota)\widetilde\theta^0(t,\ell, \tau)  \quad 
\mbox{and} \quad 
\widetilde\Theta^{(\nu )}(t,\iota)  =   \sum_{\ell=1}^r \kappa_{\ell,h_{st}}(\iota)\widetilde\Theta^{0(\nu)}(t,\ell),
\end{eqnarray*} 
where
$\widetilde{\Theta}^{0(\nu)} (t, \iota)$ is the $\nu$th column of $\widetilde{\Theta}^0 (t, \iota)$
and $\kappa_{\ell,h_{st}}(\iota)$ is a normalized kernel function with bandwidth parameter $h_{st}$. The kernel function
$\kappa_{\ell,h_{st}}(\iota)$ is given by
$$
\kappa_{\ell,h_{st}}(\iota) = \frac{ K ( ( u_\iota - u_\ell ) /h_{st} ) K (( v_\iota - v_\ell ) /h_{st} )   }{  \sum_{j=1}^r K ( ( u_\iota - u_j) /h_{st}  ) K ( ( v_\iota - v_j) /h_{st})  },
$$
where $ (u_\iota, v_\iota) $   represents the longitude and latitude of region $\iota$. Consequently, regions with smaller spatial distances contribute more significantly.

Thirdly, we estimate $\textrm{CQTE}_{\tau st}$ by substituting the refined estimators $\widetilde\theta_{\tau st}(t,\iota)$ and $\widetilde\Theta_{st}(t,\iota)$ and use the resulting estimator $\widehat{ \textrm{CQTE} }_{\tau st}$ as the test statistic $T_{\tau st}$.

Finally, we introduce a bootstrap method to test \eqref{hypo:QTE_st}. During each iteration, we resample the estimated error processes to obtain the bootstrap estimates $\widetilde\theta_{\tau st}^b(t,\iota)$, $\widetilde\Theta_{ st}^b(t,\iota)$, and the bootstrapped statistic $T_{\tau st}^b=\widehat{\textrm{CQTE} }^b_{\tau st}$. We reject $H_0$ in \eqref{hypo:QTE_st} if $T_{\tau st}$ exceeds the upper $\alpha$th empirical quantile of $\{ T^b_{\tau st}-T_{\tau st} \}_{b=1}^B$. As this approach is highly similar to the one presented in Section \ref{subsec:estimation_temporal}, we omit further details for brevity.



\section{Direct and indirect effects}\label{sec:DEIE}
Recall that Proposition \ref{prop:QTE_form} provides the closed-form expression of $\textrm{CQTE}_\tau$, 
which is 
\begin{eqnarray*}
\sum_{t=1}^m  \gamma(t, \tau)+\sum_{t=2}^m \beta(t, \tau)^\top  \left\{ \sum_{k=1}^{t-1} \left[\prod_{l=k+1}^{t-1} \Phi(l)\right]  \Gamma(k)  \right\}.
\end{eqnarray*}
Consequently, we can divide the quantile treatment effect into two components. Specifically, the first term $\sum_{t=1}^m \gamma(t, \tau) $ of $\textrm{CQTE}_\tau$ represents the direct effect of the treatment on the immediate outcome, expressed as
\begin{eqnarray*}
\textrm{CQDE}_{\tau}=Q_{\tau}\Big(  \sum_{t=1}^m Y_t^*( \bm{1}_{t} ) | \mathcal{E}_m \Big)  -  Q_\tau \Big( \sum_{t=1}^m Y_t^*( 0, \bm{1}_{t-1} ) | \mathcal{E}_m\Big).		
\end{eqnarray*}
Observe that for each $t$, the two potential outcomes $Y_t^*( \bm{1}_{t} )$ and $Y_t^*( 0, \bm{1}_{t-1} )$ differ in the treatment received at time $t$, but they share the same treatment history. 
The second term $\sum_{t=2}^m \beta(t, \tau)^\top  \left\{ \sum_{k=1}^{t-1} \left[\prod_{l=k+1}^{t-1} \Phi(l)\right]  \Gamma(k)  \right\}$ quantifies the carryover effects of past treatments on the current outcome, defined as  
\begin{eqnarray*}
\textrm{CQIE}_{\tau}=Q_{\tau}\Big(  \sum_{t=1}^m Y_t^*( 0, \bm{1}_{t-1} ) | \mathcal{E}_m \Big)  -  Q_\tau \Big( \sum_{t=1}^m Y_t^*( \bm{0}_{t} ) | \mathcal{E}_m \Big).		
\end{eqnarray*}
Similar decompositions have been considered in \cite{li2022network} and \cite{shi2022multi}.

The corresponding testing hypotheses are given by
\begin{eqnarray}
\label{hypo:QDE}
H^{D}_0: \textrm{CQDE}_\tau \le 0 \quad \textrm{versus} \quad H_1^D:\textrm{CQDE}_\tau >0,  \\
\label{hypo:QIE}
H^{I}_0: \textrm{CQIE}_\tau \le 0 \quad \textrm{versus} \quad H_1^I:\textrm{CQIE}_\tau >0.
\end{eqnarray}
Testing these hypotheses not only enables us to determine whether the new policy is significantly better than the old one or not, but also helps us understand how the new (or the old) policy outperforms the other. 

To test \eqref{hypo:QDE} and \eqref{hypo:QIE}, we use the two-step estimators in \eqref{eq:theta_estimator_refined} and \eqref{eq:Theta_estimator_refined} to construct the plug-in estimators $\widehat{\textrm{CQDE} }_\tau$ and $\widehat{\textrm{CQIE} }_\tau$ for CQDE and CQIE, respectively. Next, we employ the bootstrap method in Section \ref{subsec:estimation_temporal} to approximate the limiting distributions of $\widehat{\textrm{CQDE} }_\tau$ and $\widehat{\textrm{CQIE} }_\tau$ under the null hypotheses. We note that although $\widehat{\textrm{CQDE} }_\tau$ has a tractable limiting distribution and is asymptotically normal, estimating its asymptotic variance without using bootstrap remains challenging.

Finally, we can similarly define the direct effect and indirect effect in the spatiotemporal design as follows, 
\begin{eqnarray*}
\textrm{CQDE}_{\tau st}&=&  \sum_{\iota=1}^r\sum_{t=1}^m\{\gamma_1(t,\iota,\tau)+\gamma_2(t,\iota, \tau)\},\\
\textrm{CQIE}_{\tau st}& =& 
\sum_{\iota=1}^r\sum_{t=2}^m \beta(t,\iota, \tau)^\top \left[ \sum_{k=1}^{t-1} \left(\prod_{j=k+1}^{t-1} \Phi(j,\iota)\right) \{\Gamma_1(k,\iota)+\Gamma_2(k,\iota)\} \right].	
\end{eqnarray*} 
The estimation and inference procedures can be derived similarly.

\section{Asymptotic Properties}
\label{sec:asymptotic results}
In this section, we  
investigate the theoretical properties of  
the proposed test statistics for CQTE.   
Firstly, we provide an upper error bound as a function of $h,m$ and $n$ to measure the approximation error of the proposed bootstrap method in temporal dependent experiments.

\begin{thm}
\label{thm:bootstrap_consistency_QTE}
Suppose that Assumptions \ref{assump:kernel}--\ref{asmp:st1} of the supplement material hold, if
$h=o(n^{-1/4})$, $m\asymp n^{c_2}$ for some $1/2< c_2<3/2$ and $mh\to\infty$, 
as $n \rightarrow \infty$,
\begin{eqnarray*}
	\sup_{ \tau \in [\varepsilon, 1- \varepsilon] }\sup_{z}|\prob(T_{\tau}-\textrm{CQTE}_{\tau}  \leq z)-\prob( T^b_{\tau}-T_{\tau}  \leq z |\textrm{Data}) \vert 
	\leq & C(\sqrt{n}h^2 + \sqrt{n} m^{-1} 
	+n^{-1/8}) 
\end{eqnarray*}
with probability approaching $1$, 
for some $\varepsilon \in (0, 1)$ 
and some positive constant $C$.
\end{thm}

 Under the conditions specified in Theorem \ref{thm:bootstrap_consistency_QTE}, the error bound delineated in Theorem \ref{thm:bootstrap_consistency_QTE} tends towards zero. When the null hypothesis holds true, we derive that
$
\prob(T_{\tau} \leq z)\le \prob( T^b_{\tau}-T_{\tau} \leq z |\textrm{Data})+ o_p(1),
$
where the little-$o_p$ term uniformly applies to $z$ and $\tau$. When the alternative hypothesis is true and the QTE signal satisfies $m^{-1}$QTE$_{\tau}\gg n^{-1/2}
\log(n m) $, the power of the proposed test method approaches $1$ (refer to the proof of Theorem \ref{thm:bootstrap_consistency_QTE} for more details). Consequently, the consistency of the proposed test is demonstrated. 

There are two significant challenges in establishing
Theorem \ref{thm:bootstrap_consistency_QTE}: (i) the non-differentiable nature of the checkloss function, and (ii) the allowance for $m$ to diverge with $n$. In order to address the first challenge, we utilize classical M-estimation theory from the literature to obtain a Bahadur representation of the proposed estimator \citep[see, for instance,][]{koenker1987estimation}. 
To tackle the second challenge, we draw from arguments similar to those from the high-dimensional multiplier bootstrap theorem \citep{chernozhukov_gaussian_2013} to generate a nonasymptotic error bound that explicitly characterizes the dependence of the bootstrap approximation error on $m$.

Secondly, we  establish  the consistency of the proposed test procedure in spatiotemporal dependent experiments.

\begin{thm}
\label{thm:bootstrap_consistency_QTE_st}
Suppose that Assumption \ref{assump:kernel} and
Assumptions \ref{assump:F_e_spatial}--\ref{asmp:st1_spatial} of the supplement material hold,
$h, h_{st}=o(n^{-1/4})$, $m, r \asymp n^{c_2}$ for some $1/2 < c_2<3/2$ and $rh_{st}, mh\to\infty$, then
\begin{eqnarray*}
	\sup_{ \tau \in [\varepsilon, 1- \varepsilon] } \sup_{z}|\prob(T_{\tau st}-\textrm{CQTE}_{\tau st}  \leq z)-\prob( T^b_{\tau st}-T_{\tau st} \leq z |\textrm{Data}) \vert \\
	\leq  C[\sqrt{n}(h^2 +h_{st}^2) + \sqrt{n} ( m^{-1} + r^{-1})  
	+n^{-1/8}]. 
\end{eqnarray*}
with probability approaching 1, for some $\varepsilon \in (0,1)$ and some positive constant $C$.
\end{thm}

Based on Theorem \ref{thm:bootstrap_consistency_QTE_st}, we can demonstrate that the proposed test effectively controls the type-I error and its power tends towards $1$ when $m^{-1} r^{-1}$QTE$_{\tau st}\gg n^{-1/2} \log(n m r) $. More details can be found in the proof of Theorem \ref{thm:bootstrap_consistency_QTE_st}. 
Furthermore, we have provided additional theoretical properties of the proposed test statistics for CQDE and CQIE in Section \ref{sec: decomposition theoretical results} of the supplementary material.

\section{Real Data Analysis}
\label{sec:real data}

To address (Q1)-(Q3), we apply the proposed test procedures to the three real datasets obtained from Didi Chuxing introduced in Section \ref{sec:data}. 

Firstly, we examine the dataset from a temporally dependent A/B experiment conducted from Dec 10, 2021 to Dec 23, 2021. As detailed in Section \ref{sec:data}, two order dispatch policies are tested in alternating one-hour time intervals. The new policy, in comparison to the old one, is designed to fulfill more call orders and elevate drivers' total income. We set drivers' total income as the outcome, and the observation variables include the number of call orders and drivers' total online time. 
To address question (Q1), we apply model \eqref{model:TQVCM DE} to elucidate the correlation structure between supply and demand and model \eqref{model:TQVCM IE} to elucidate the temporal interference effects. For question (Q3), we utilize the testing procedure described in Section \ref{subsec:estimation_temporal} for these temporally dependent experiments. 
As a means to validate the proposed test, we also apply our procedure to the A/A dataset outlined in Section \ref{sec:data}, where a single order dispatch strategy is employed. We anticipate that our test will not reject the null hypothesis when applied to this dataset.

In Figure \ref{fig:residuals_temporal}, we display the estimated residuals of the outcome over time for $\tau \in {0.1, 0.5, 0.9}$. As can be seen from Figure \ref{fig:residuals_temporal}, some residuals are significantly larger than others, suggesting that the outcome likely originates from heavy-tailed distributions. This reinforces the use of quantile treatment effects for policy evaluation. 
Table \ref{table:temporal_results} presents the $p$-values of the proposed test for CQTE$_{\tau}$, CQDE$_{\tau}$, and CQIE$_{\tau}$, respectively. Furthermore, Figure \ref{fig:p-values_temporal} illustrates the estimated treatment effects and the p-values across various quantiles. As expected, the proposed test does not reject the null hypothesis at any quantile level when applied to the A/A experiment. 
However, when applied to the A/B experiment, the new policy demonstrates significant quantile direct effects on the business outcome at most quantile levels. In contrast, the indirect effects are not significant.

\begin{figure}[!t]
\centering
\includegraphics[height=3.5cm, width=4.5cm]{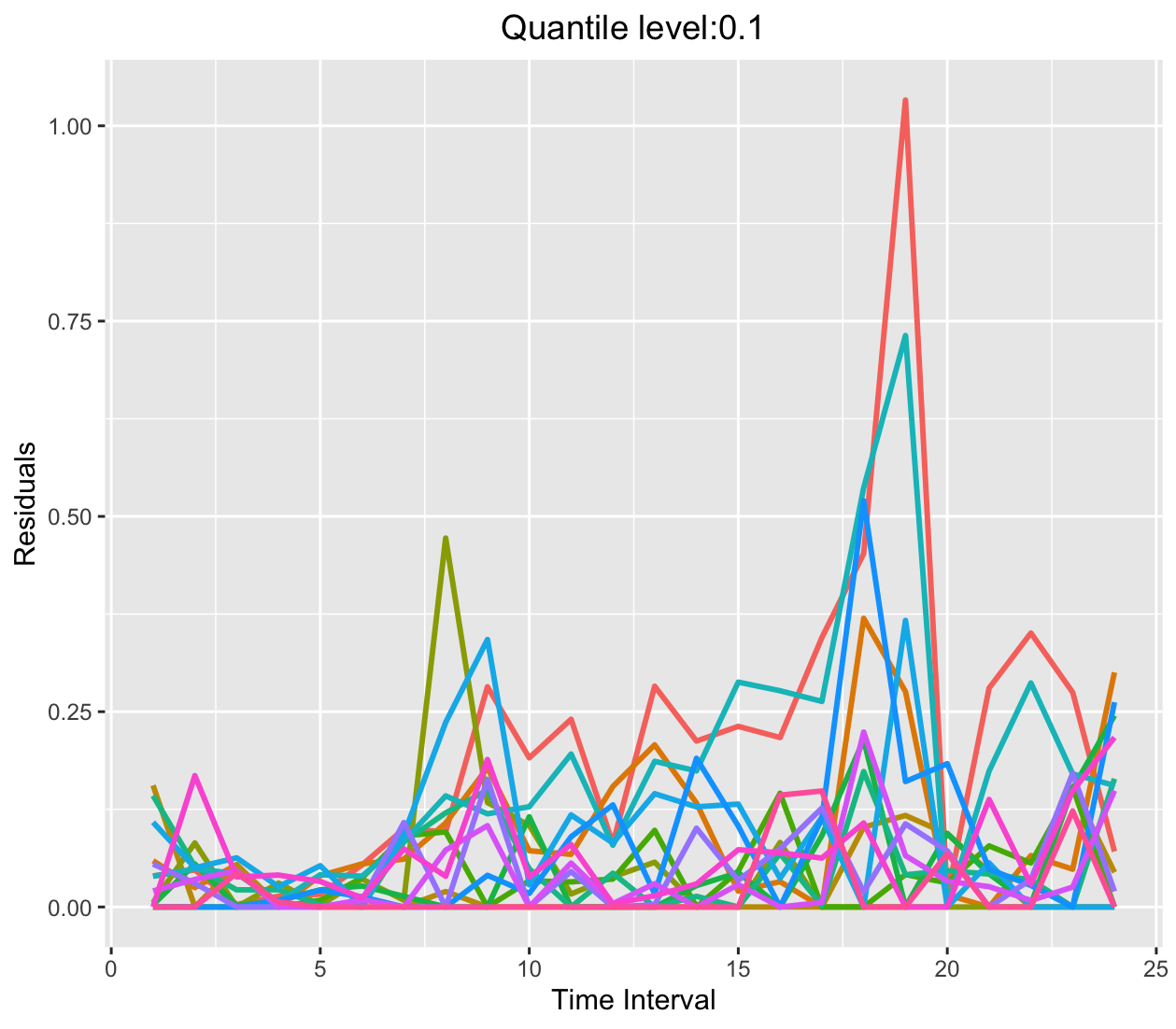}
\includegraphics[height=3.5cm, width=4.5cm]{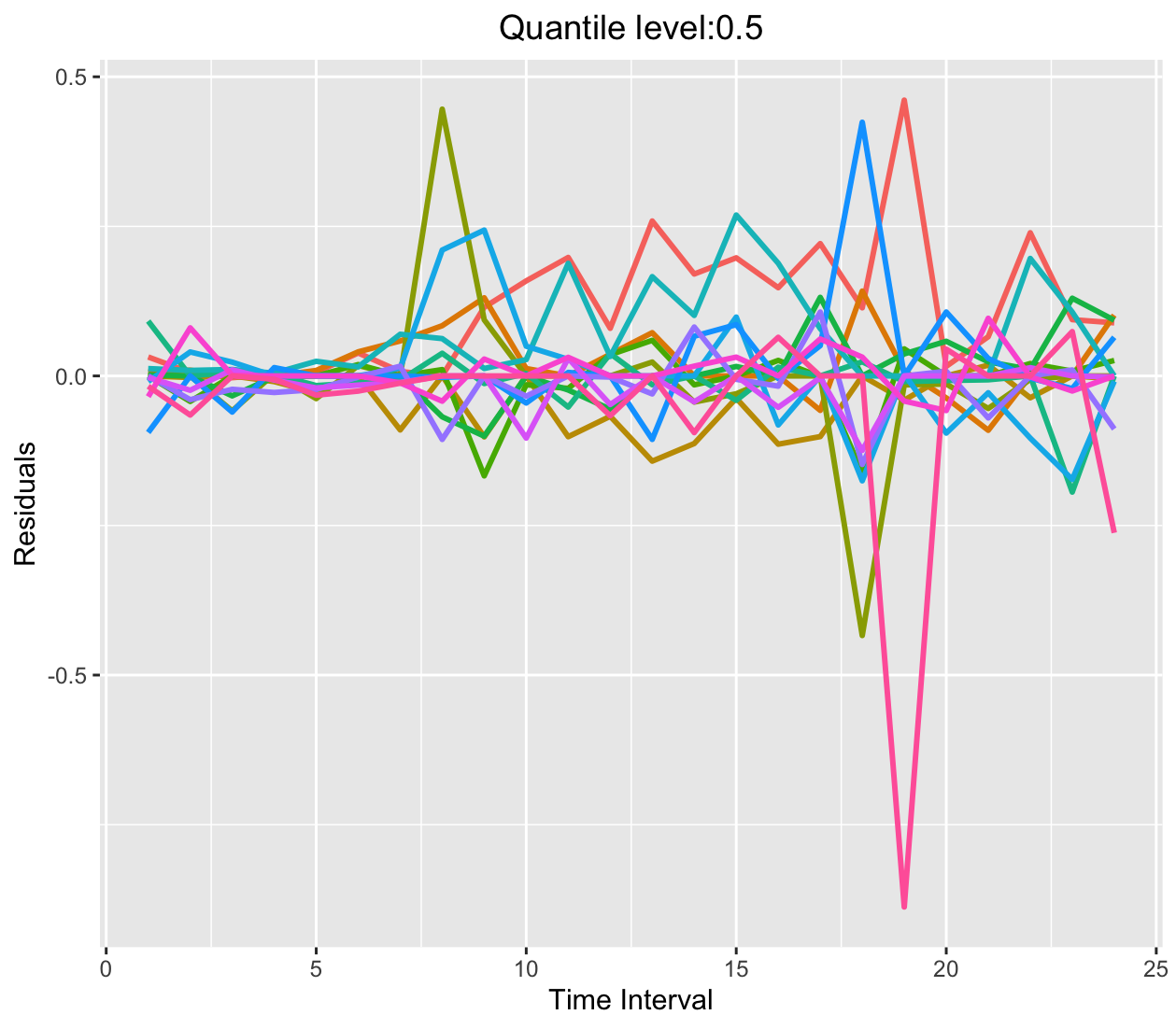}
\includegraphics[height=3.5cm, width=4.5cm]{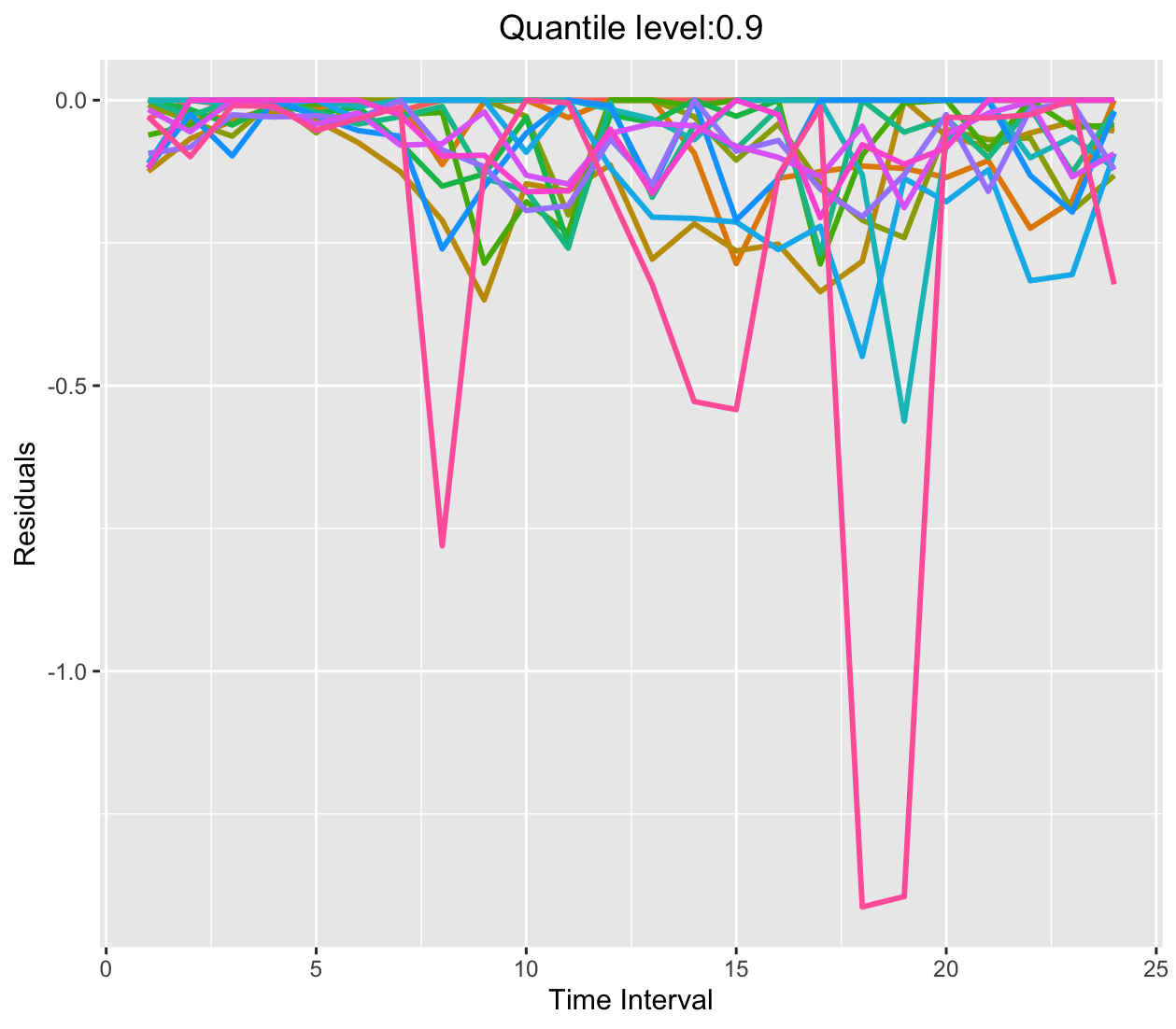} 
\caption{\small Estimated residuals of drivers' total income at quantile levels 0.1, 0.5, and 0.9 in the temporal dependent experiment.}
\label{fig:residuals_temporal}
\end{figure}

\begin{table}
\begin{center}
\small 
\caption{
	$p$-values of the proposed test for CQDE$_{\tau}$ and CQIE$_{\tau }$ for both datasets from the A/A experiment and A/B experiment, utilizing the time-alternation design.  
	}
\label{table:temporal_results}

\tabcolsep 8pt
\begin{tabular}{cccccccccccc}
	\hline
	\hline
	&  \multicolumn{2}{c}{pvalues for AA} & &	\multicolumn{2}{c}{pvalues for AB}  \\
	\cline{2-4} \cline{5-6}
	$\tau$  & $\text{CQDE}_{\tau }$   &  $\text{CQIE}_{\tau }$  & &   $\text{CQDE}_{\tau }$  & CQIE$_{\tau }$   \\
	\hline
	0.1    &  0.286  &  0.084   &  &  0.208	& 0.076 \\
	0.2  &   0.522   &  0.096   &   & 0.080 & 0.060 \\
	0.3	 &   0.53	& 0.098  &  & 0.002 &	0.068 \\
	0.4	 &   0.568	&  0.122 &  &  0.010 &	0.086 \\
	0.5	 &  0.536	&  0.116  &  &  2e-4 &	0.072 \\
	0.6	 &  0.464	&  0.100  &  &  0.002  &	0.068 \\
	0.7	&  0.548	&  0.102 &  &  7e-4  &	0.092 \\
	0.8	&  0.606	&  0.108  &  &  2e-4  &	0.068 \\
	0.9	&  0.322	&  0.102  &  &  7e-5 &	0.100 \\         
	\hline \hline     
\end{tabular}
\end{center}
\end{table}

\begin{figure}
\centering
\includegraphics[height=3.5cm, width=4cm]{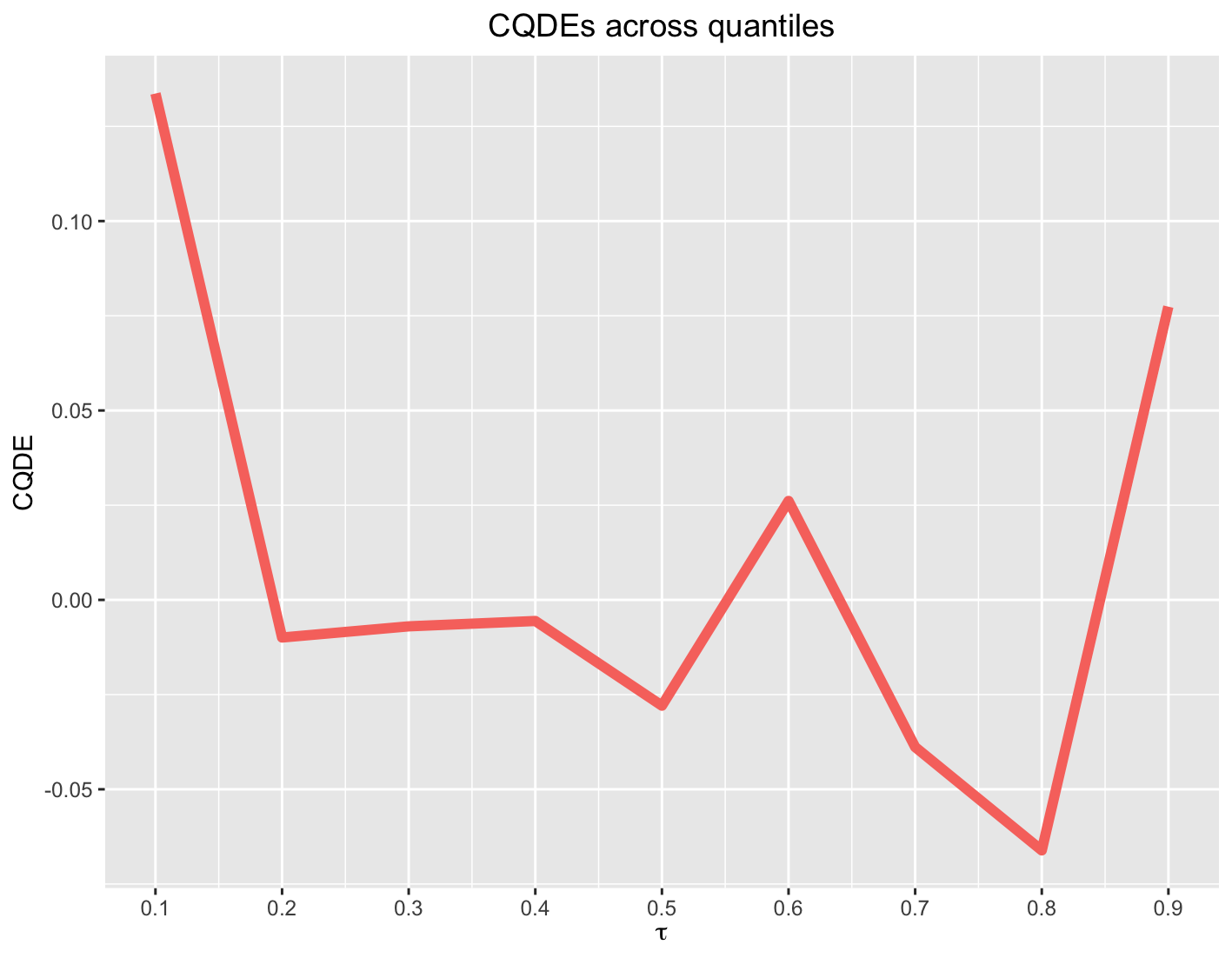}
\includegraphics[height=3.5cm, width=4cm]{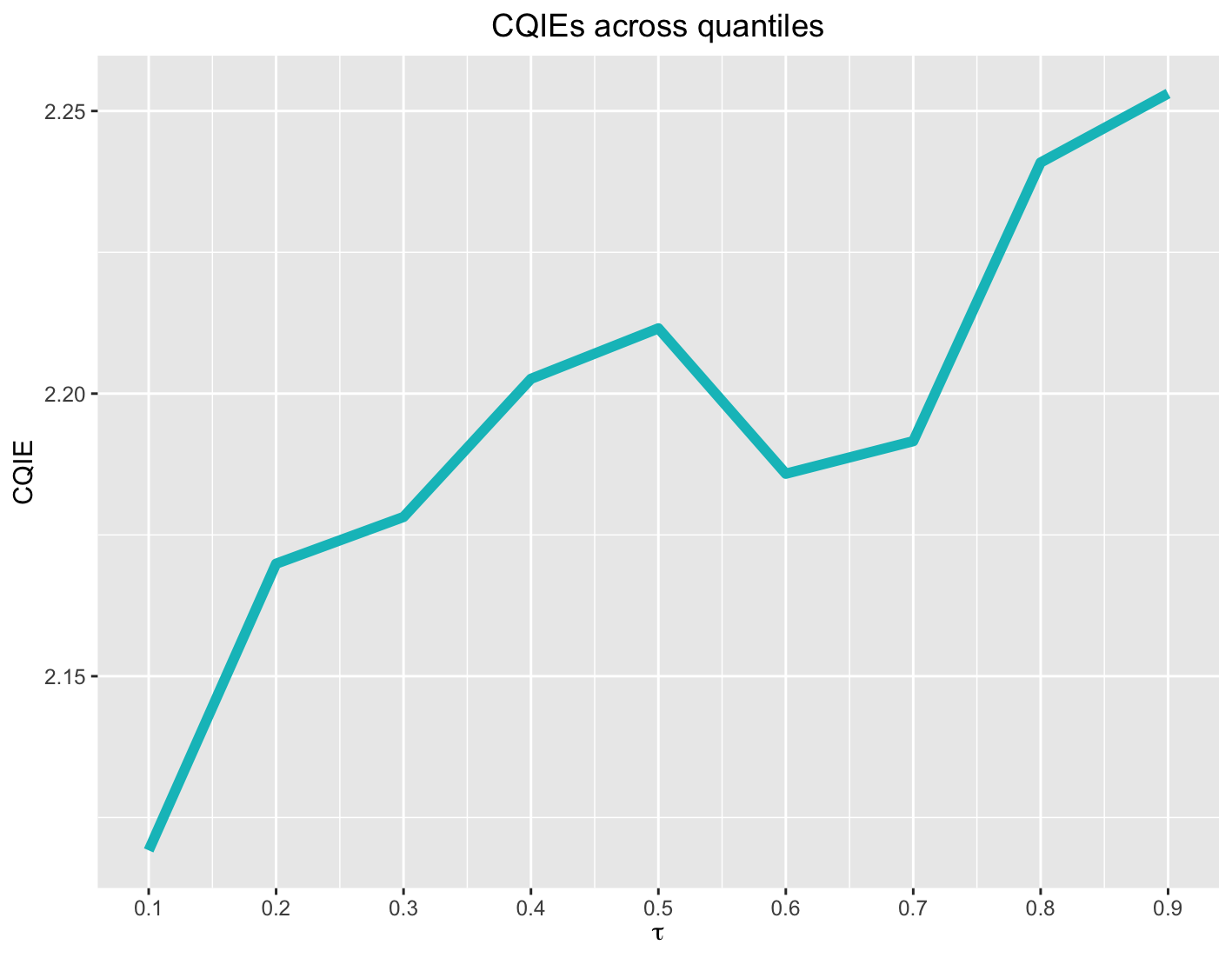}
\includegraphics[height=3.5cm, width=5cm]{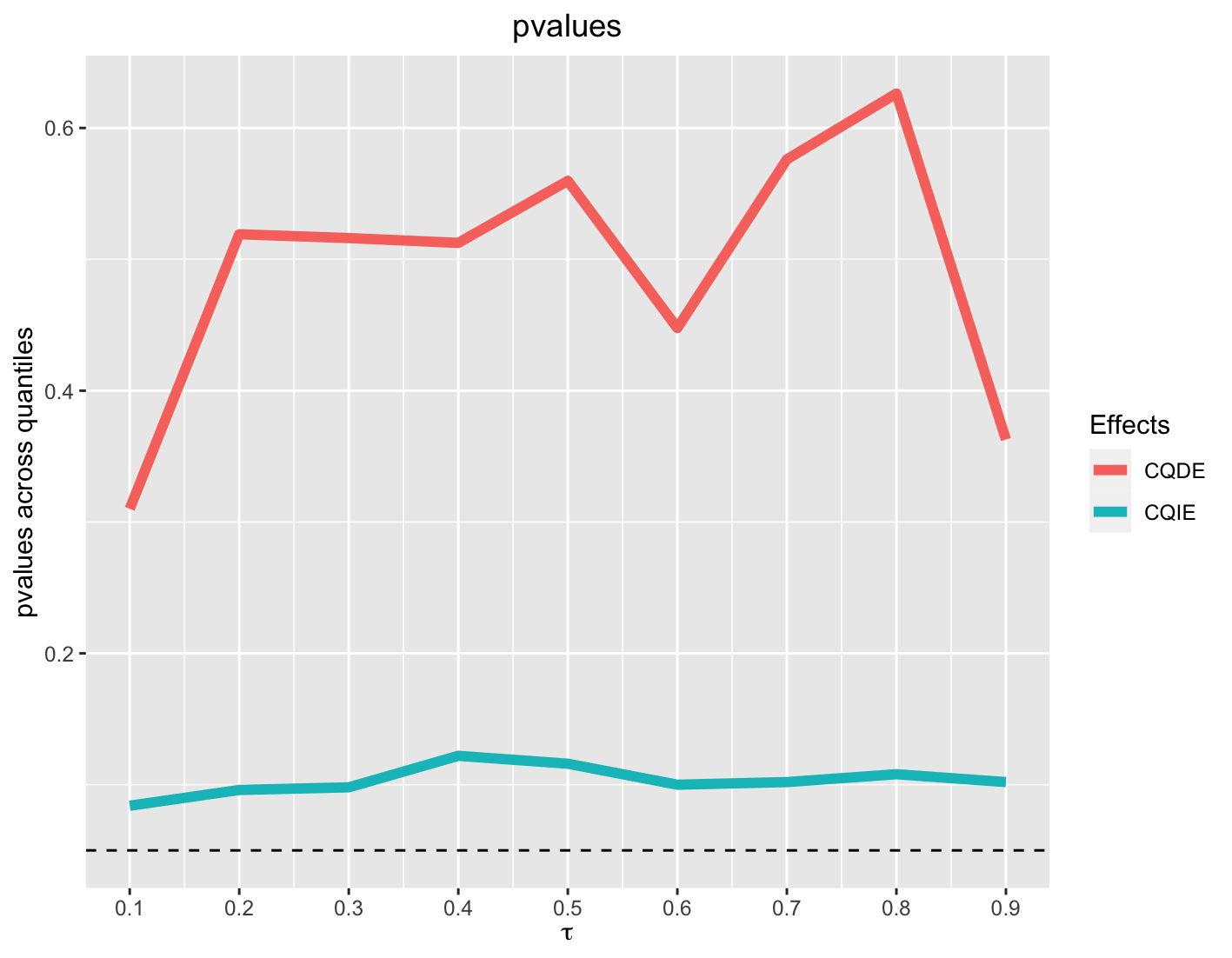} \\
\includegraphics[height=3.5cm, width=4cm]{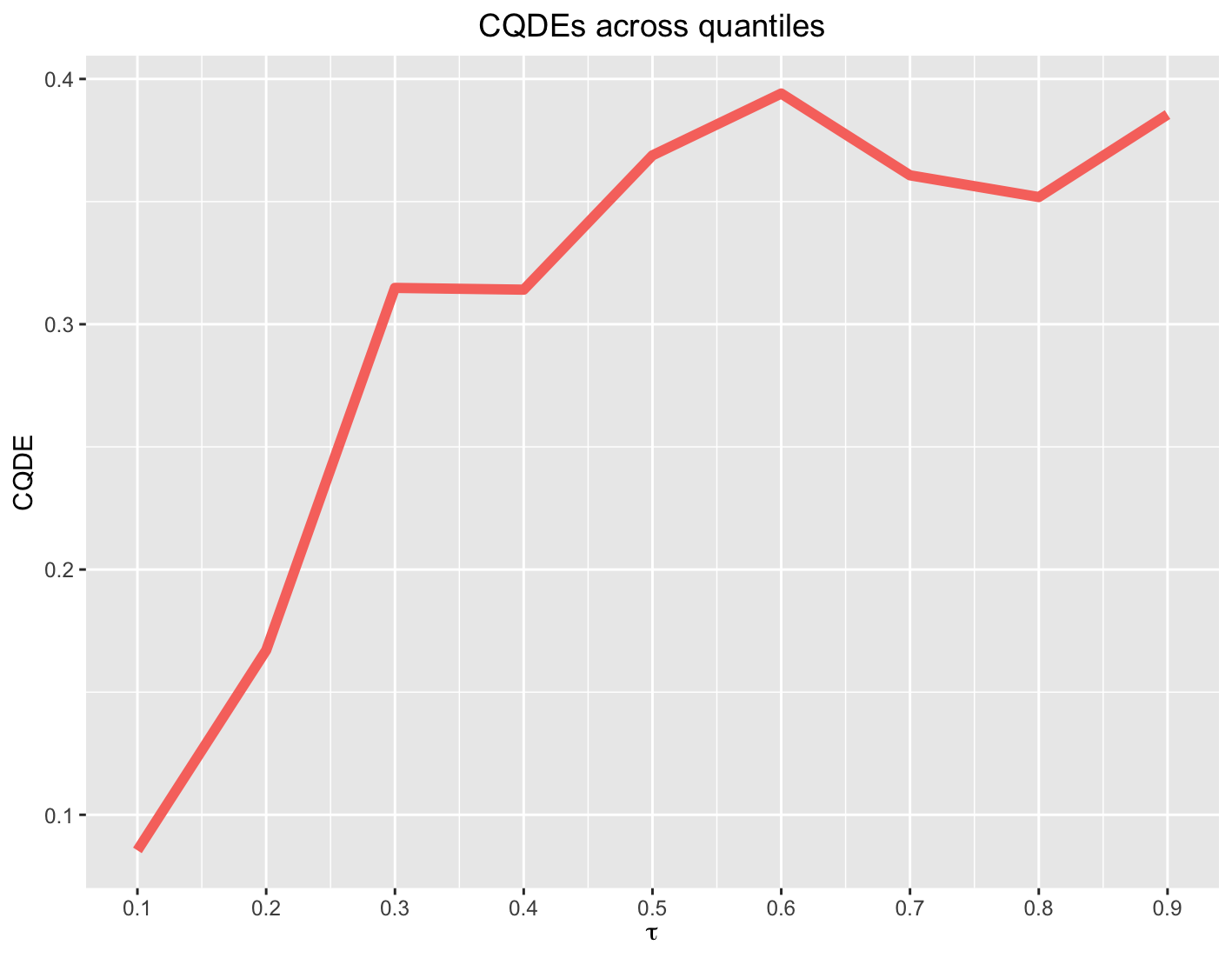}
\includegraphics[height=3.5cm, width=4cm]{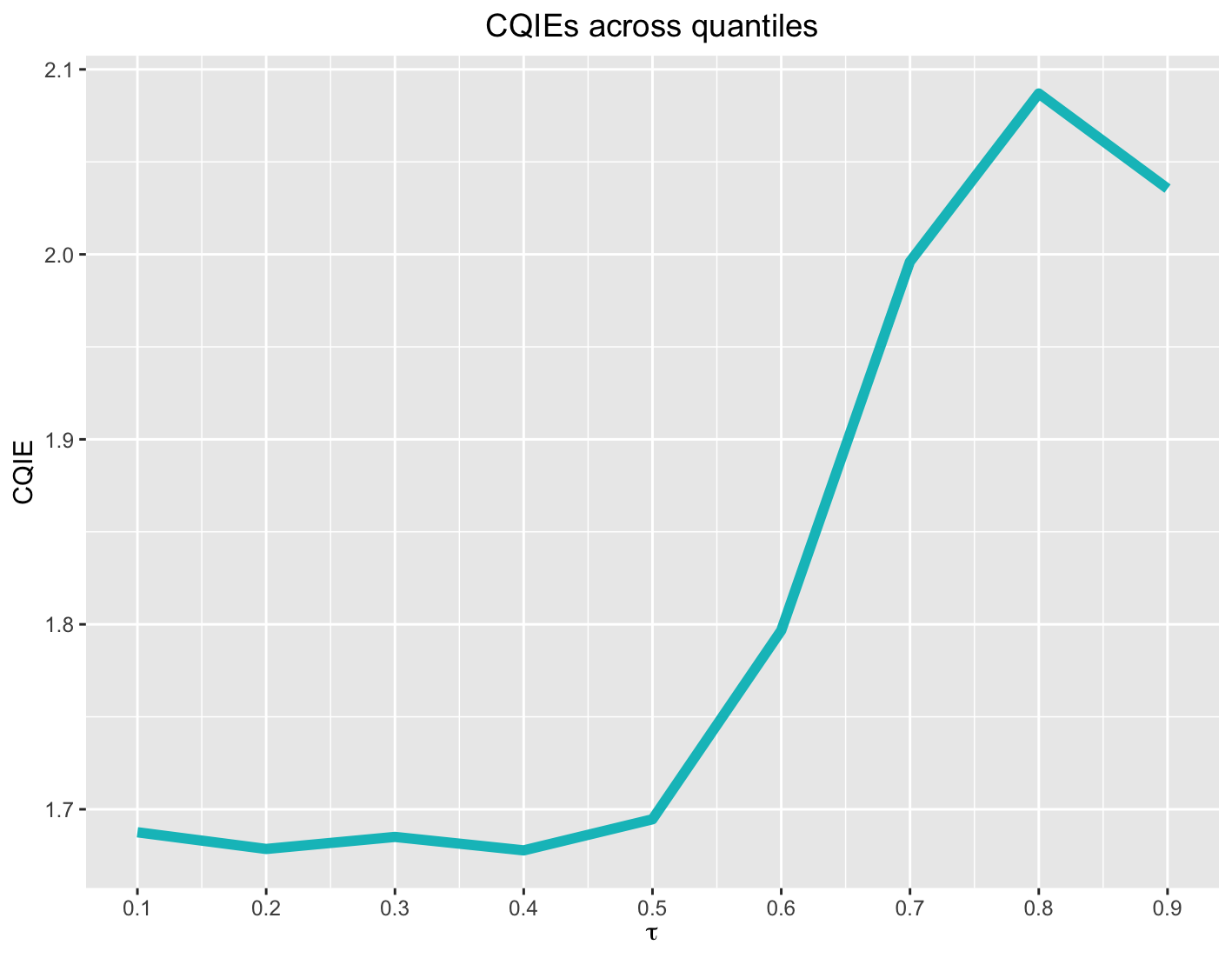}
\includegraphics[height=3.5cm, width=5cm]{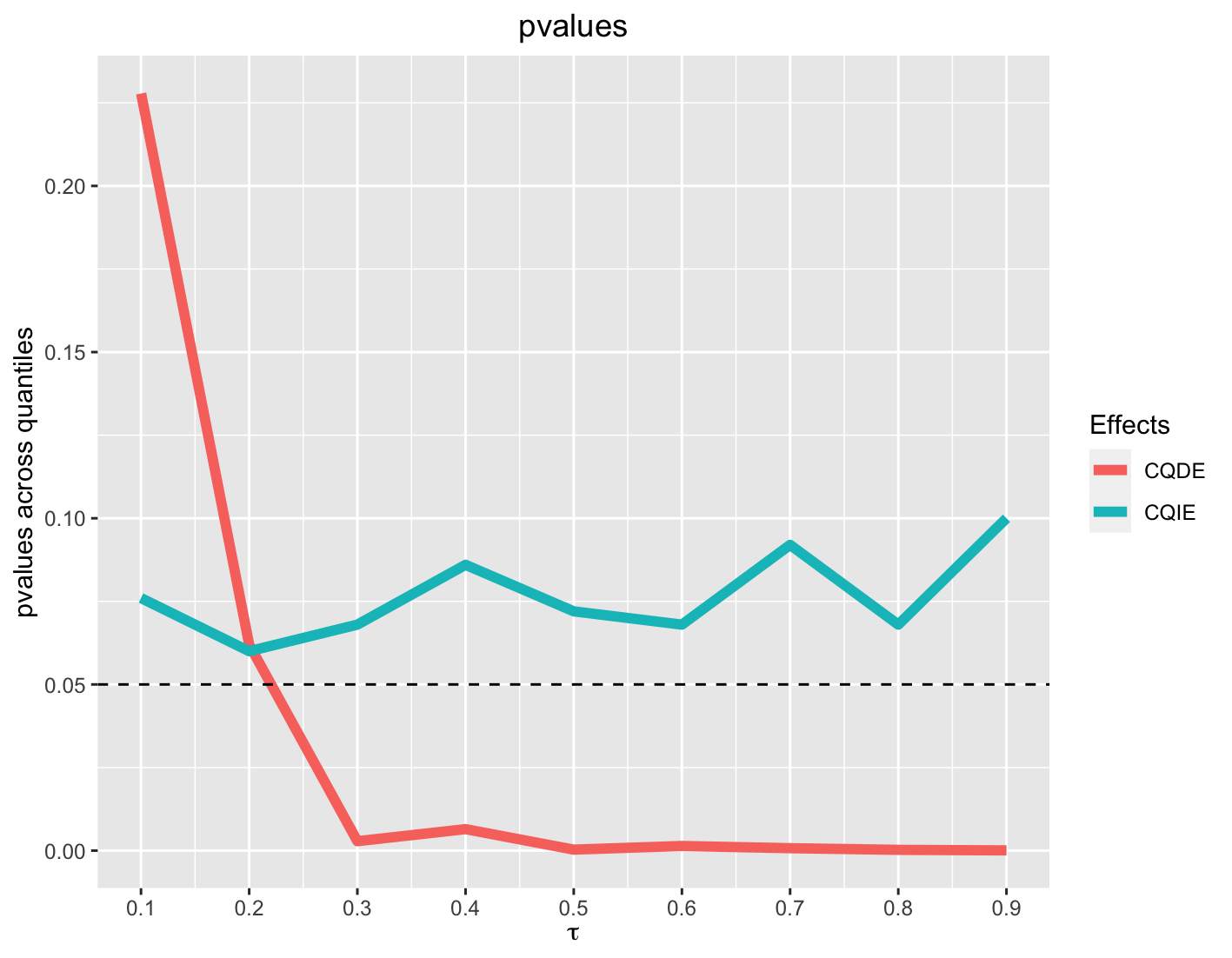} 
\caption{\small 
Estimates of CQDE$_{\tau}$ and CQIE$_{\tau}$ and their $p$-values across quantile levels for the A/A experiment (top panels) and A/B experiment (bottom panels) under the temporal design. }
\label{fig:p-values_temporal}
\vspace{-0.1in}
\end{figure}

Secondly, we analyze the dataset from the spatiotemporal dependent experiment as described in Section \ref{sec:data}. Recall that in this experiment, the city is divided into 12 regions. Policies are implemented based on alternating 30-minute time intervals within each region. We concentrate on a data subset collected from 7 am to midnight each day, as there are relatively few order requests from midnight to 7 am. The drivers' total income and the number of call orders are designated as the outcome and state variable, respectively. 
We fit the spatiotemporal VCDP models \eqref{model:STQVCM DE} and \eqref{model:STQVCDP IE} to address (Q2), and apply the testing procedure from Section \ref{sec:inference_st} to address (Q3) for this spatiotemporal dependent experiment. 
Our aim is to determine whether the new policy has significant treatment effects on drivers' total income across various quantile levels.

For each quantile level, we implement the proposed estimation and testing procedures on the data. The $p-$values are generated through the bootstrap procedure outlined in Section \ref{sec:ST method}, utilizing 500 bootstrap samples. The estimation and testing results for CQDE$_{\tau st}$ and CQIE$_{\tau st}$ are summarized in Table \ref{table:spatial_results} and Figure \ref{fig:p-values}. The treatment effects are significant at most quantile levels, and both the estimated direct and indirect effects are positive across all quantiles. Generally, these effects escalate with the quantile level. However, the new policy doesn't seem to boost the lower quantile of the outcome (e.g., $\tau=0.1$). These results underline the heterogeneous effects of the new policy across different quantile levels.

Finally, we display the scaled outcomes, and residuals for  the representative region 5 over time, with $\tau \in \{0.1, 0.5, 0.9\}$, in Figure \ref{fig:residuals}. It is evident that there may be several outliers in the data. This observation further supports the use of quantiles as the evaluation metric. Similar patterns are observed for other regions as well.

\begin{table}[!t]
\begin{center}
\caption{\small $p$-values and estimators of CQDE$_{\tau st}$ and CQIE$_{\tau st}$ for the spatiotemporal data.}
\label{table:spatial_results}

\tabcolsep 8pt
\begin{tabular}{cccccccccccc}
	\hline
	\hline
	$\tau$ & pvalue$_{ \text{CQDE}_{\tau st}}$   & pvalue$_{ \text{CQIE}_{\tau st}}$  &  $\widehat{\textrm{ CQDE} }_{\tau st}$ &  $\widehat{\textrm{ CQIE} }_{\tau st}$   \\
	\hline
	0.1	&  0.290  &  0.024   &  1.566  &  14.153 		\\
	0.2	&  0.072  &  0.036   &  3.403  &  15.002 		\\
	0.3	&  0.026  &  0.020   &  4.022  &  16.032		\\
	0.4	&  0.032  &  0.016   &  3.678  &  16.939		\\
	0.5	&  0.010  &  0.022   &  5.482  &  17.725	 	\\
	0.6	&  0.004  &  0.020   &  5.902  &  18.559	\\
	0.7	&  0.004  &  0.022   &  7.139  &  19.535	 	\\
	0.8	&  0.006  &  0.014   &  5.746  &  20.473	 	\\
	0.9	&  7e-4   &  0.008   &  8.414  &  21.320		\\    
	\hline \hline     
\end{tabular}
\end{center}
\end{table}

\begin{figure}[t]
\centering
\includegraphics[height=3.5cm, width=4cm]{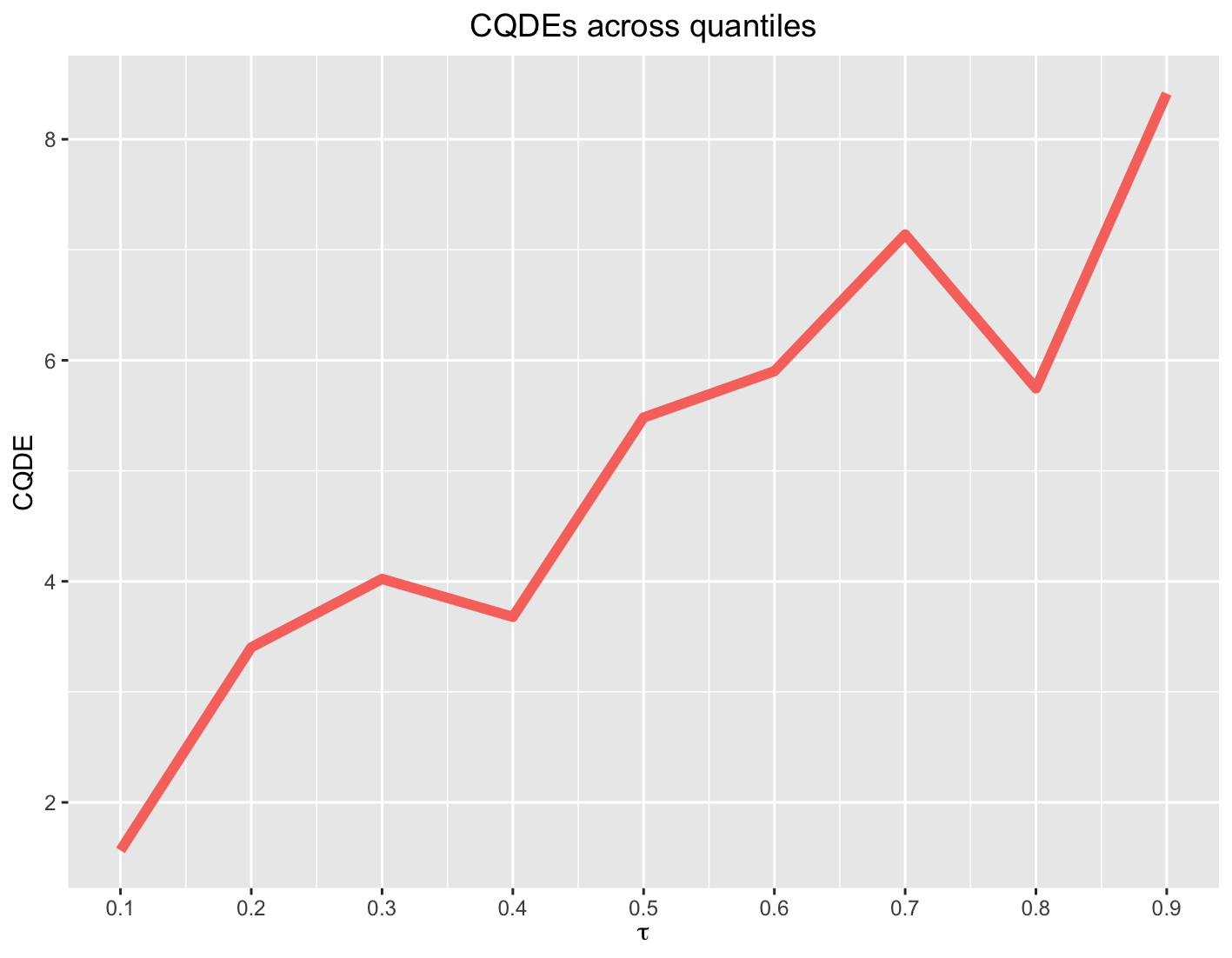}
\includegraphics[height=3.5cm, width=4cm]{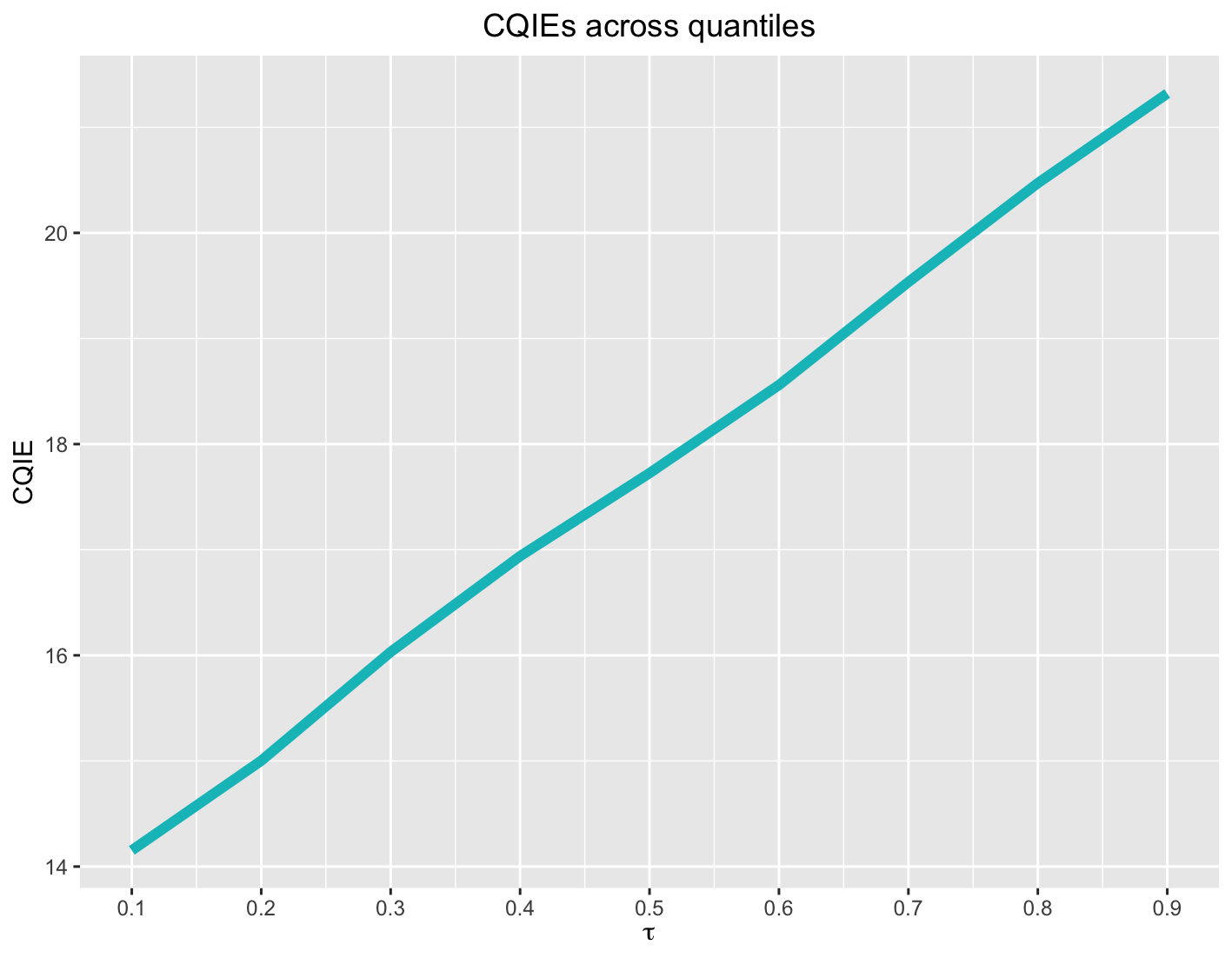}
\includegraphics[height=3.5cm, width=5cm]{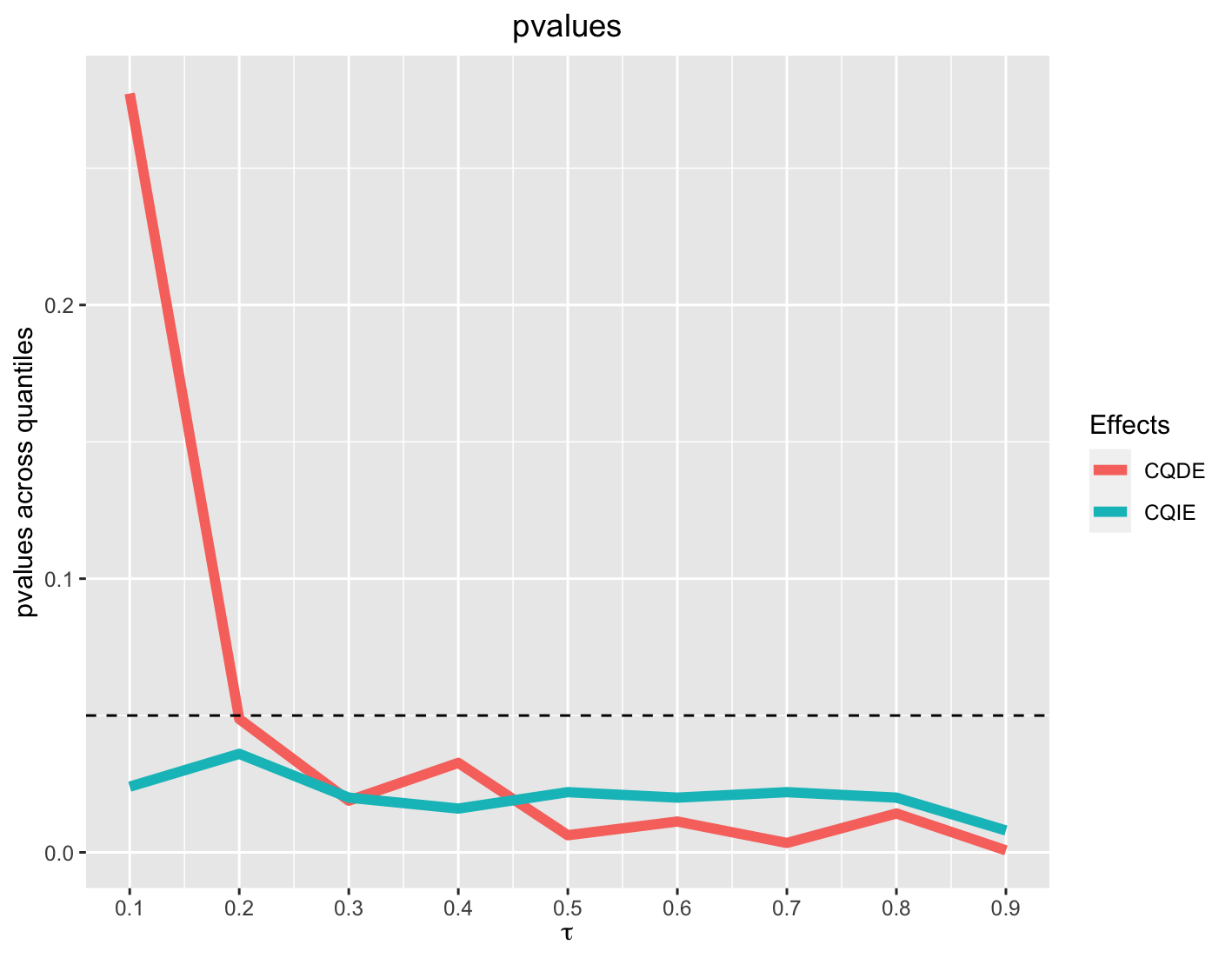}
\caption{\small Estimates of CQDE$_{\tau st}$ and CQIE$_{\tau st}$ and pvalues for the spatiotemporal data across quantiles.}
\label{fig:p-values}
\end{figure}

\begin{figure}[!t]
\centering
\includegraphics[height=3.5cm, width=3.8cm]{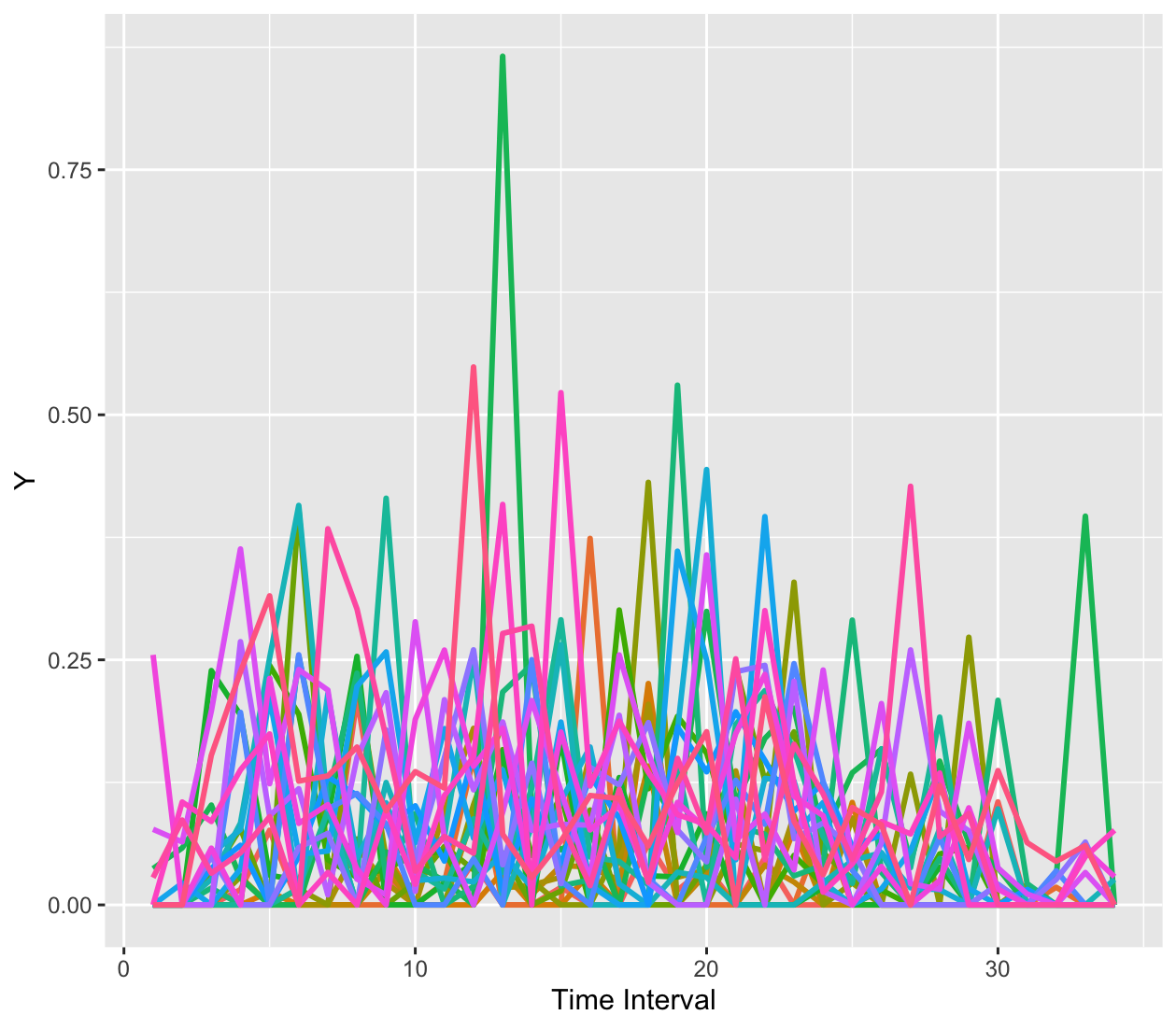}
\includegraphics[height=3.5cm, width=3.8cm]{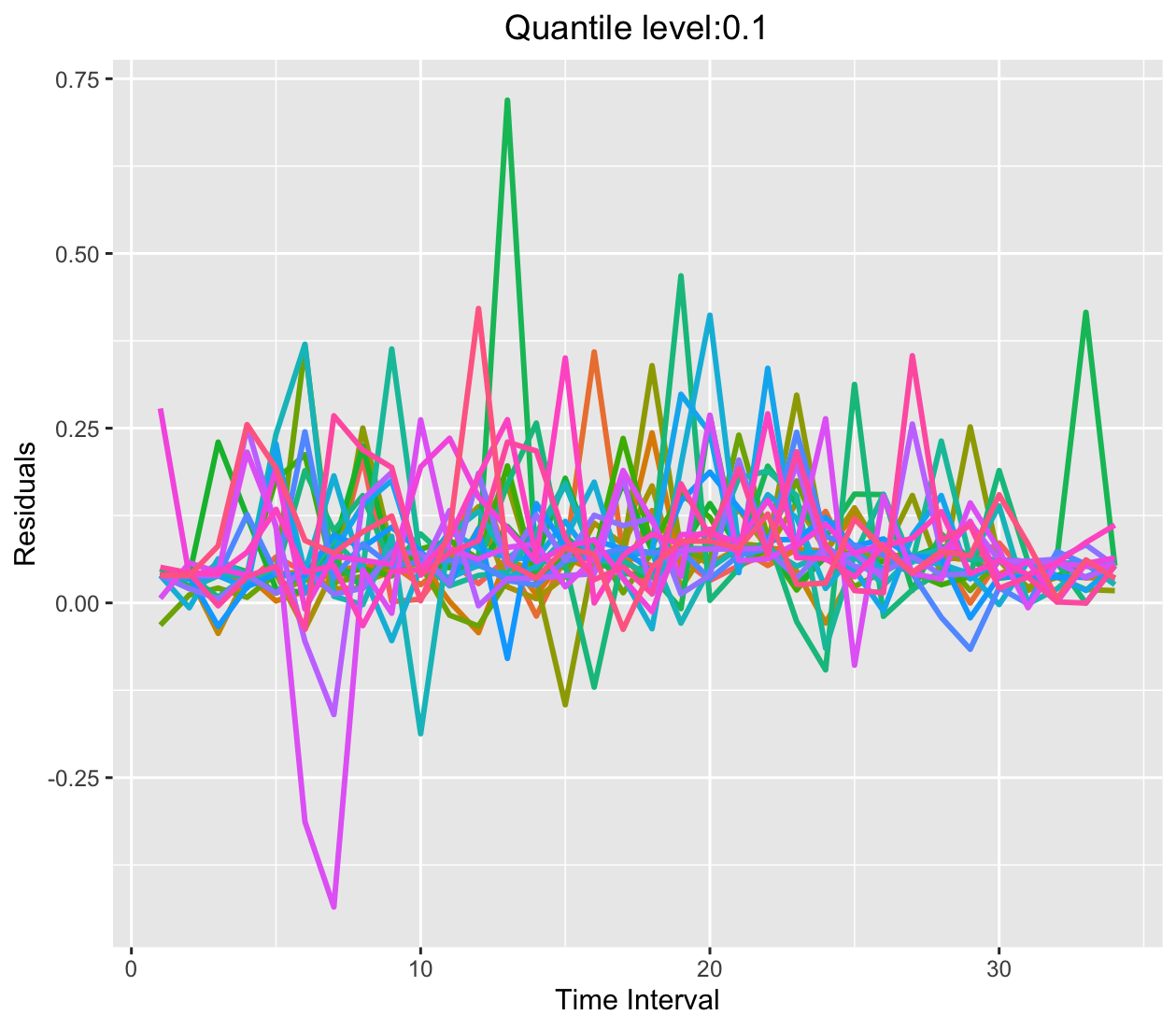}
\includegraphics[height=3.5cm, width=3.8cm]{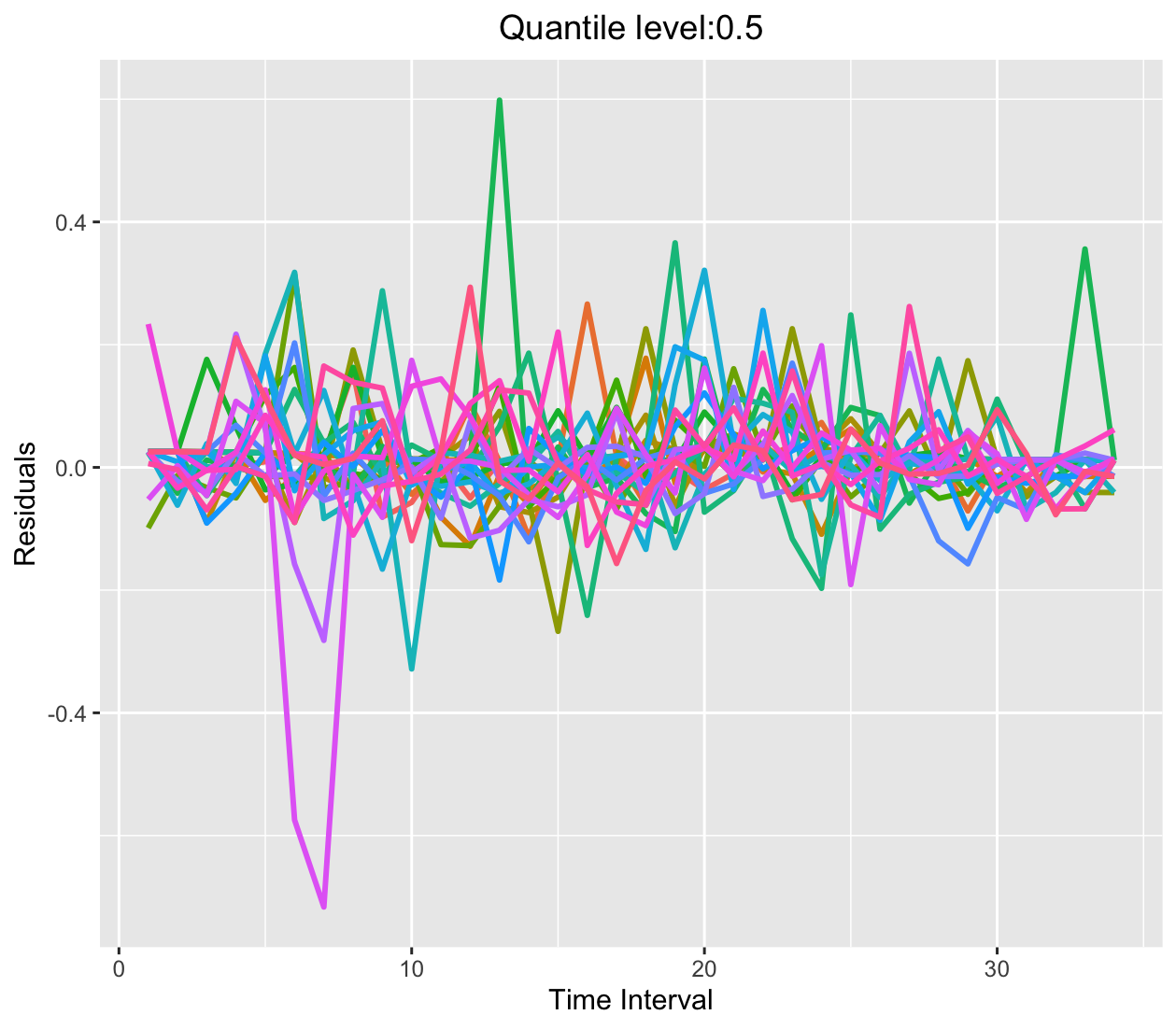}
\includegraphics[height=3.5cm, width=3.8cm]{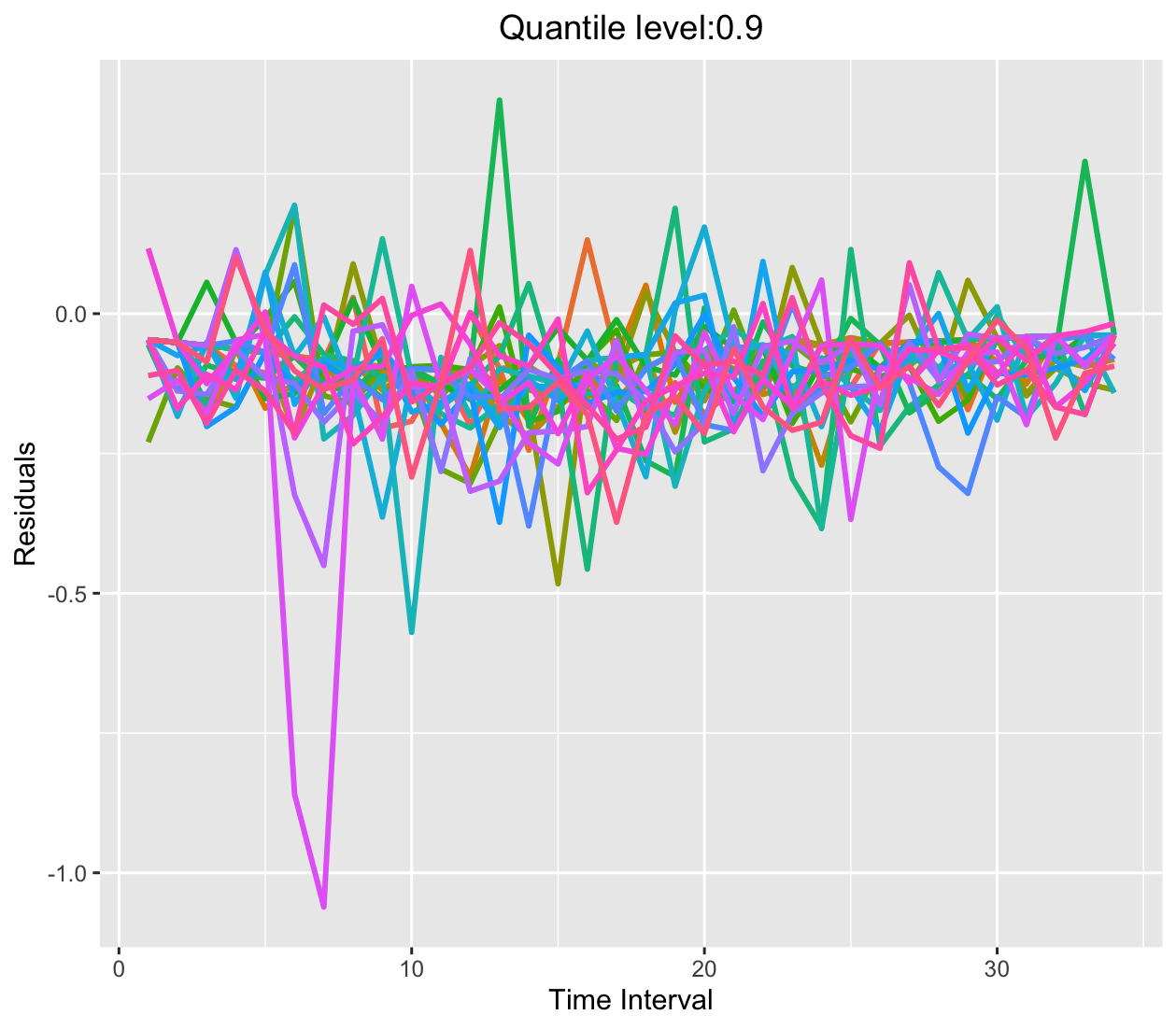} 
\caption{\small Scaled values of drivers' total income and their estimated residuals at quantile levels 0.1, 0.5, 0.9 in Regions 5. } 
\label{fig:residuals}
\end{figure}

\section{Real data based simulations}
\label{sec:simulation}

In this section, we evaluate the finite sample performance of the proposed estimation and testing procedures through simulations. Simulation experiments are conducted based on the real dataset collected from the A/A experiment described in Section \ref{sec:data}. Recall that one hour is defined as a time unit, and drivers' total income within each time unit is set as the outcome of interest. The observation variables correspond to the number of call orders and drivers' total online time. These variables characterize the demand and supply of the ridesharing platform and have a substantial impact on the outcome.

Next, we outline the simulation environment. For a given quantile level $\tau$, we fit the proposed VCDP models \eqref{model:TQVCM DE} and \eqref{model:TQVCM IE} to the data by setting $\gamma(t, \tau) = \Gamma(t)$ = 0, since the two policies being compared are essentially the same. This enables us to obtain the estimated model parameters $ \widetilde \beta_{0 \tau} (t) $, $ \widetilde \beta_{ \tau} (t) $, $ \widetilde \phi_{ 0} (t) $, and $ \widetilde \Phi (t) $ and the estimated error processes $\widetilde e_{i,\tau}(t)$ and  $\widetilde E_{i}(t)$ for $1\le t\le 24$ and $1\le i\le 68$. To simulate data, we set 
$
\widetilde \gamma ( t, \tau )  = \delta  Q_\tau ( Y_{t} )$ and $ \widetilde \Gamma ( t )  = \delta  E( S_{t} )$ 
for some constant $\delta\ge 0$, 
where  $Q_\tau ( Y_{t} )$  and $ E( S_{t} )$ denote the (elementwise) empirical $\tau$-th quantile of $\{ Y_{it} \}_i$ and  empirical mean of $\{ S_{it} \}_i$, respectively. 
Under this formulation, the coefficients  $ \widetilde \gamma ( t, \tau )$ 
are allowed to vary across different quantiles, and the constant $\delta$ controls the strength of CQTE.   Specifically, no CQTE exists if $\delta = 0$, and the new policy is better if $\delta > 0$.

Thirdly, we employ the bootstrap method for data generation. Specifically, in each simulation run, we randomly sample $n$ initial observations and $n$ error processes with replacement. Then, we generate $n$ days of data according to the proposed VCDP models:
\begin{eqnarray*}
\widetilde Y_{i,t } &=& \widetilde \beta_{0 } (t, \tau) + \widetilde S_{i,t}^\top \widetilde \beta(t, \tau) + A_{i,t}\widetilde \gamma(t, \tau) + \widetilde e_{i,\tau}(t), \\
\widetilde S_{i,t+1} &=& \widetilde \phi_{0 }(t)+ \widetilde \Phi(t) \widetilde S_{i,t}+A_{i,t} \widetilde \Gamma (t) + \widetilde E_{i}(t+1),
\end{eqnarray*}
based on these samples and the estimated model parameters. The treatments $A_{i,t}$ are generated according to the temporal alternation design. Specifically, we first implement one policy for TI time units, then switch to the other policy for another TI time units, and alternate between the two policies.
We consider a wide range of simulation settings by setting $\tau \in {0.2, 0.5, 0.8}$, $n \in { 20, 40 }$, TI $\in {1, 3 }$, and $\delta \in {0, 0.001, 0.025, 0.05, 0.075, 0.1}$.
For each scenario, we generate 500 simulation runs to compute the empirical type-I error rate and power.
The significance level is fixed at 5\% throughout the simulation.

Finally, we discuss the simulation results. Figure \ref{fig:QTE_simu} presents the empirical rejection rates of the proposed test for CQTE (refer also to Table \ref{table:simu_QTE} in the supplementary material). The type-I error is approximately at the nominal level in all cases. The empirical power generally increases with the sample size and approaches $1$ as the signal
strength $\delta$ increases to $0.1$. Furthermore, the empirical power increases with the quantile level $\tau$, which is expected since $\widetilde \gamma_\tau ( t )$ are set to be proportional to $Q_{\tau}(Y_t)$, whose values increase with the quantile level. These results validate our theoretical assertions.
We also report the empirical rejection rates of the proposed test for CQDE and CQIE in Figures \ref{fig:QDE_simu} and \ref{fig:QIE_simu} of the supplementary material, respectively. The results are very similar to those of CQTE. It is worth noting that the power for CQDE is generally larger than that of CQTE, whereas the power for CQIE is generally smaller than that for CQTE. This is because the test statistics of CQIE have much larger variances than those of CQDE.

\begin{figure}[!t]
\centering
\includegraphics[height=4cm, width=5cm]{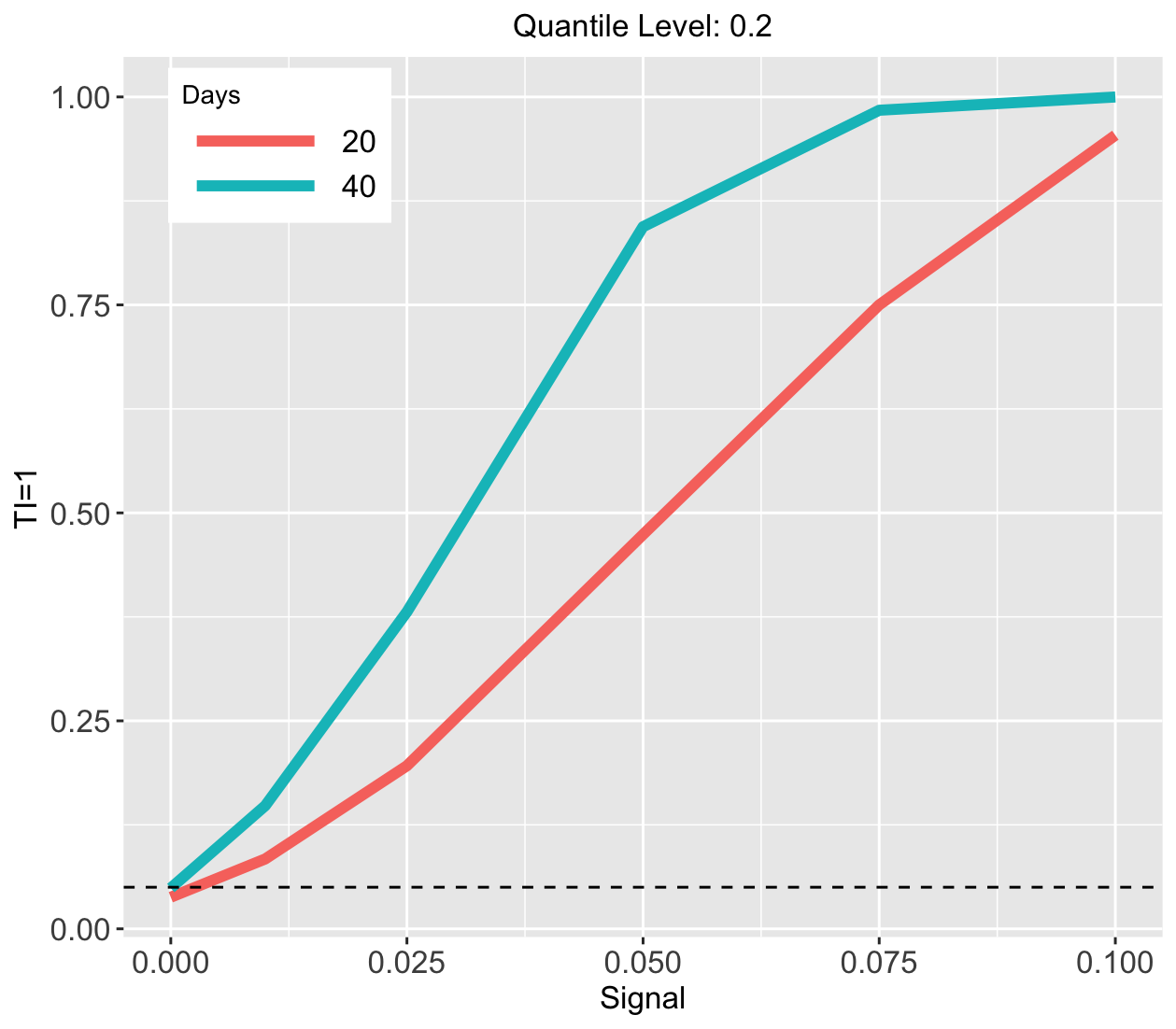}
\includegraphics[height=4cm, width=5cm]{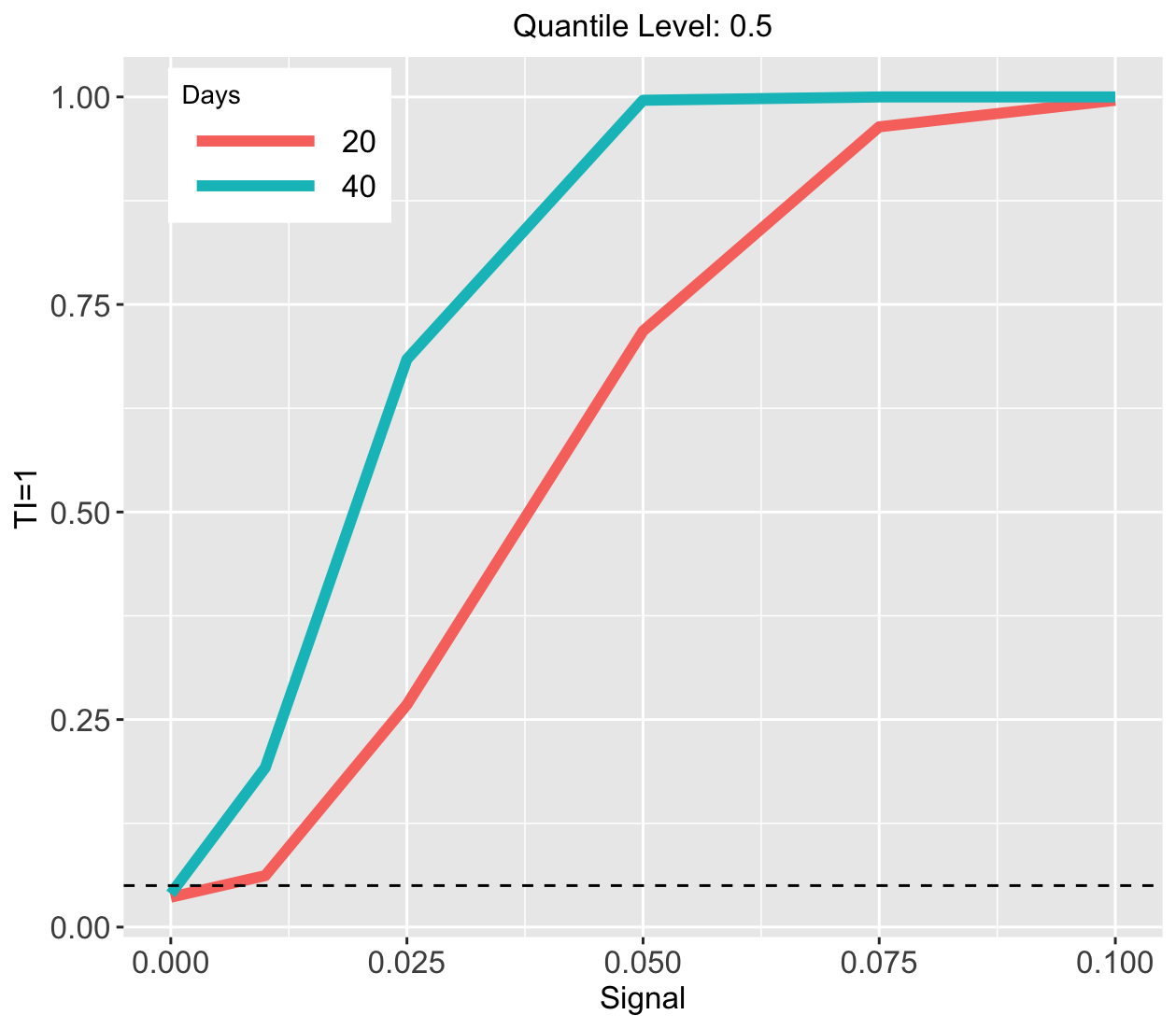}
\includegraphics[height=4cm, width=5cm]{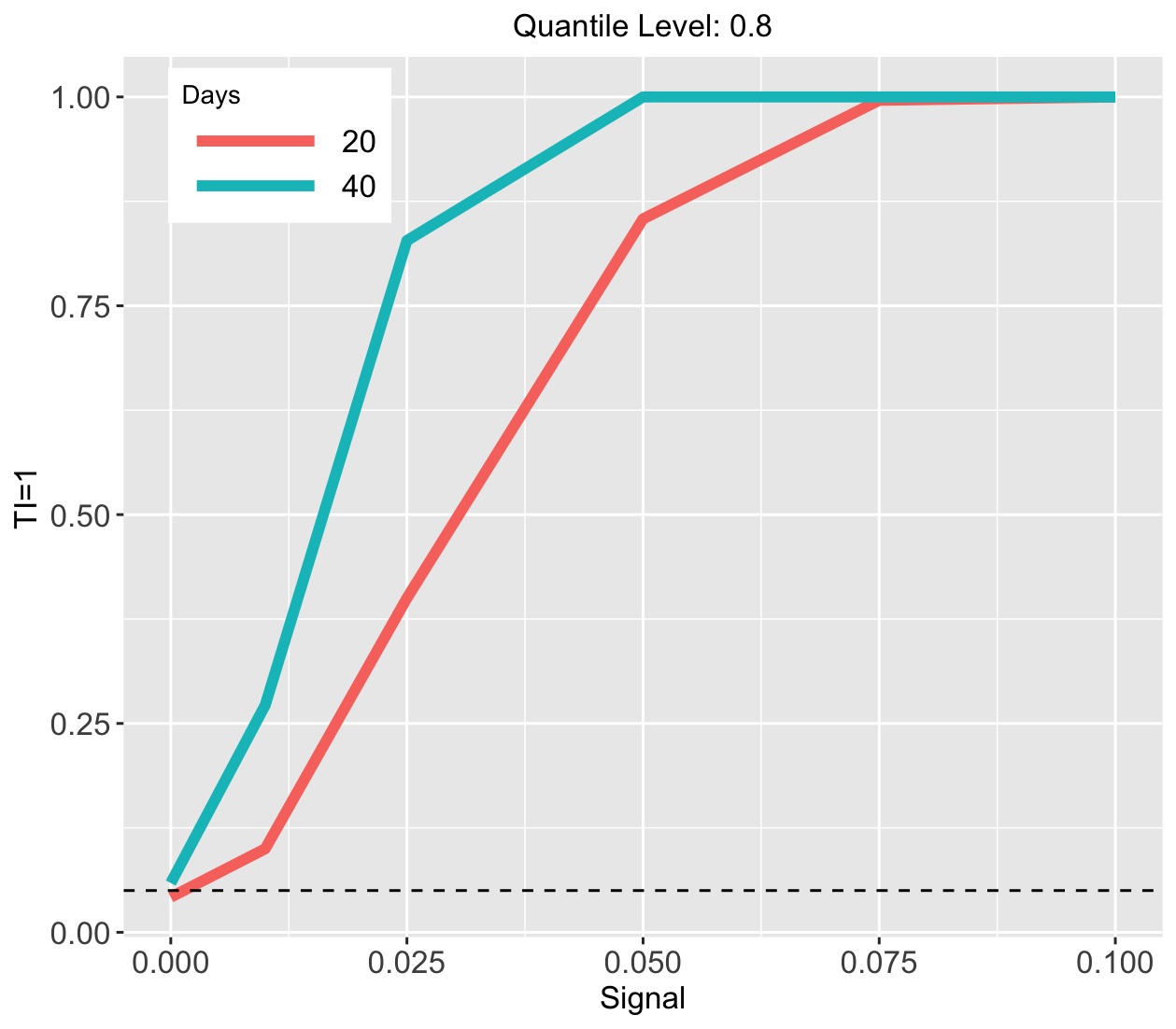} \\
\includegraphics[height=4cm, width=5cm]{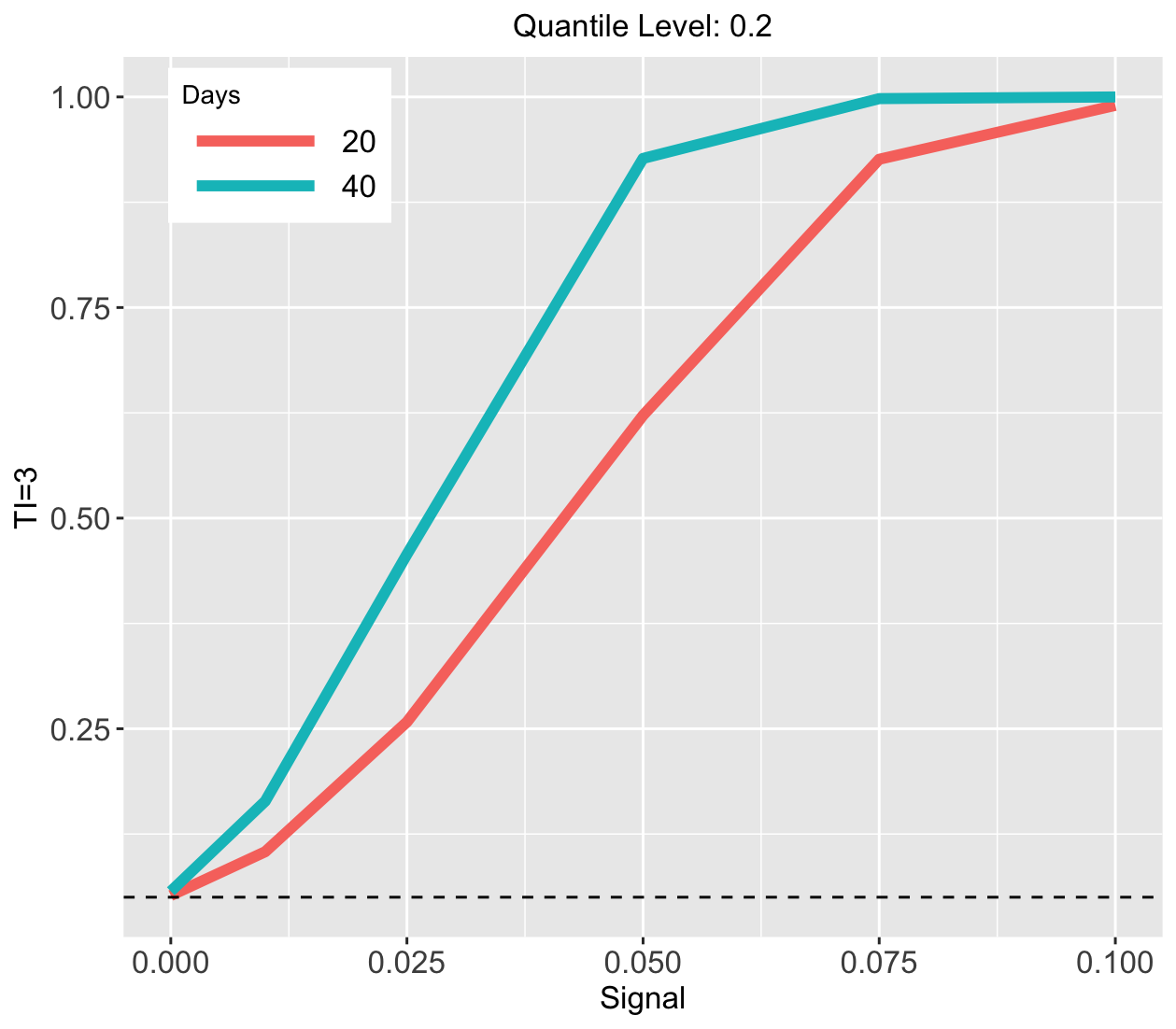}
\includegraphics[height=4cm, width=5cm]{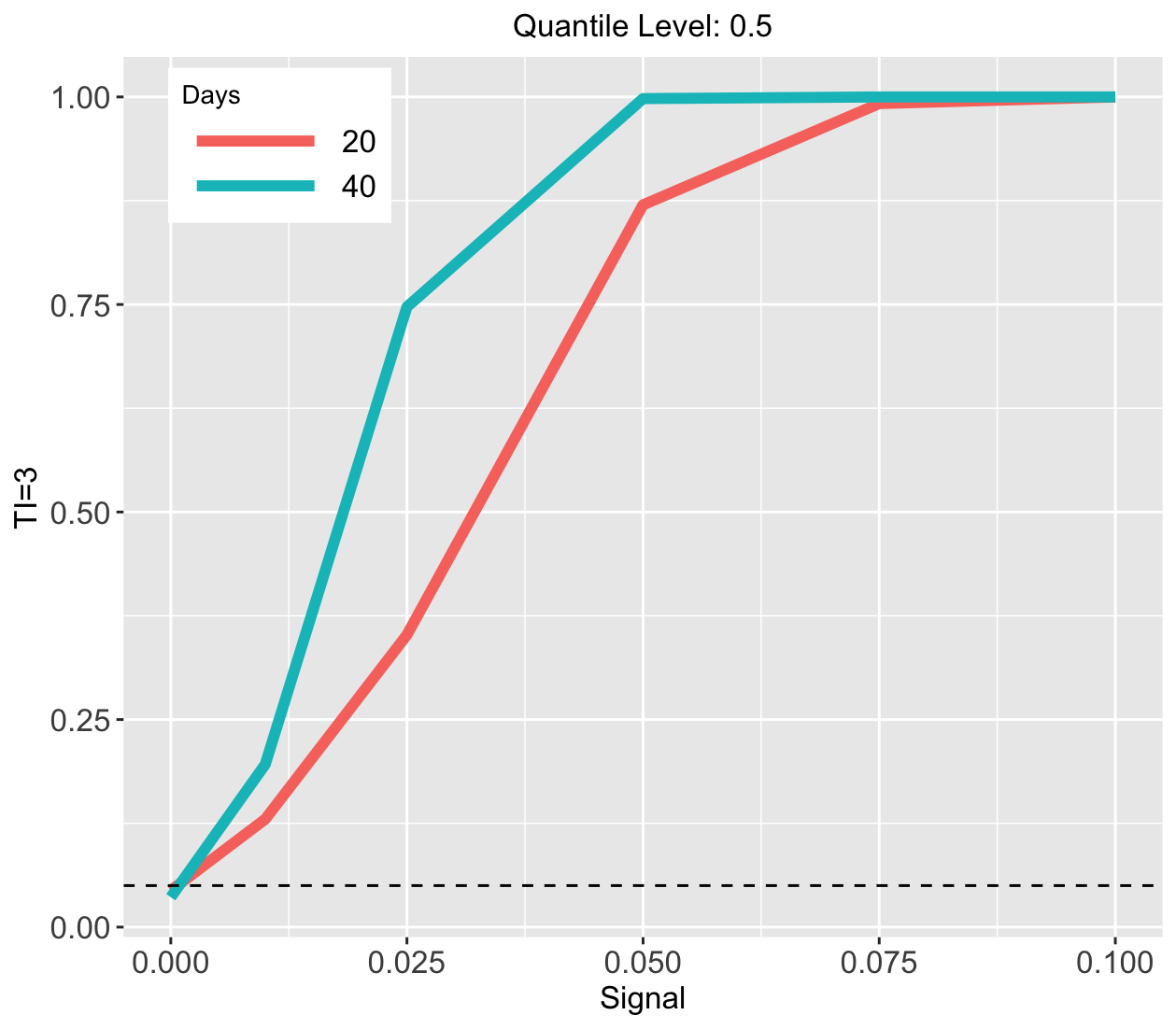}
\includegraphics[height=4cm, width=5cm]{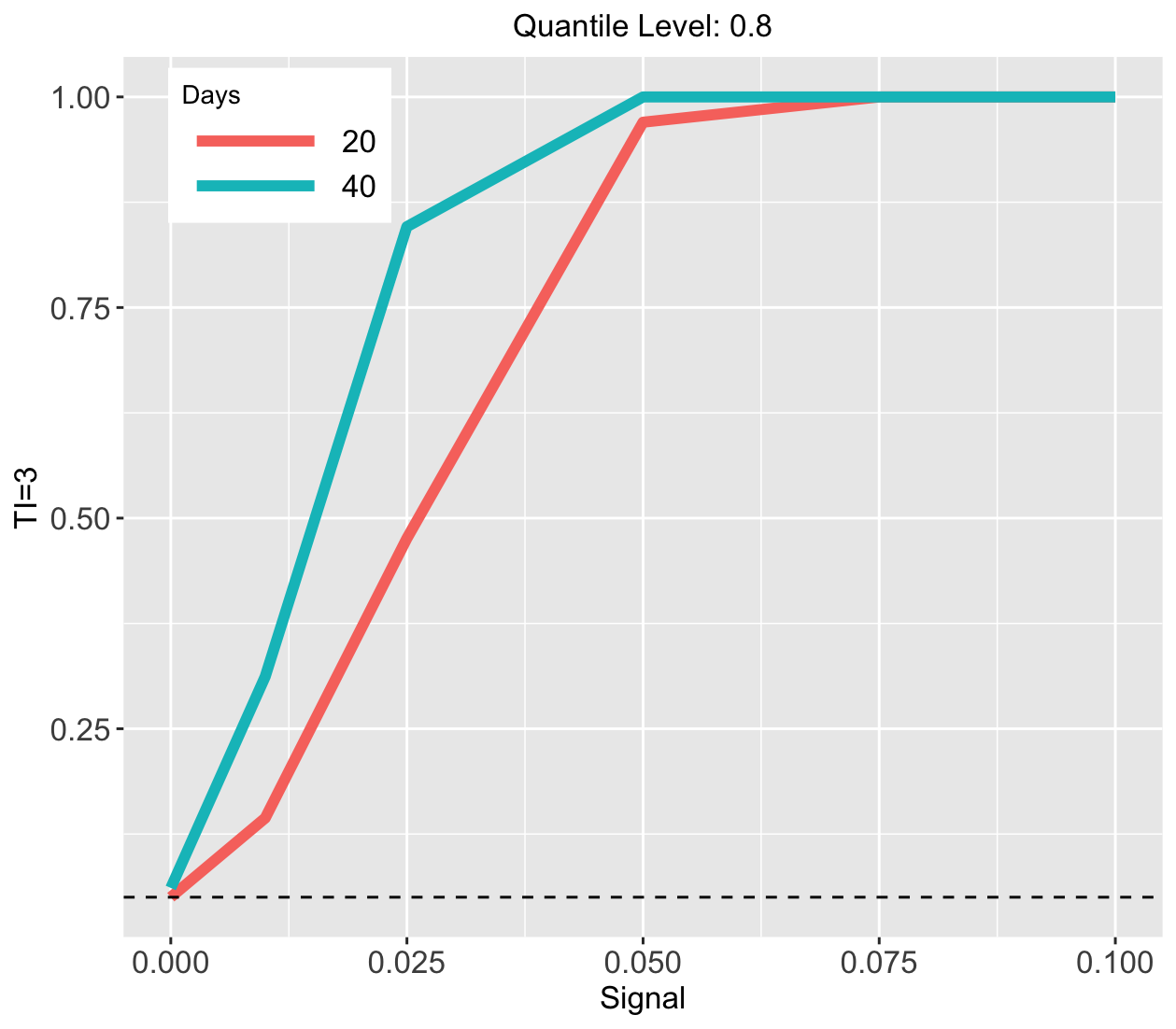}
\caption{\small Empirical rejection rates of the proposed test for CQTE$_\tau$. TI equals $1$ for the top panels and $3$ for the bottom panels. The quantile level $\tau=0.2$, $0.5$ and $0.8$, from left to right plots.} 
\label{fig:QTE_simu}
\end{figure}

\section{Discussion}
Motivated by the policy evaluation of A/B testing, we proposed a framework for inferring dynamic CQTE on market features, demonstrating that under certain conditions, the CQTE equals the sum of individual CQTE. This significantly simplifies the estimation and inference procedures by focusing on the CQTE at each spatiotemporal unit. To address the interference effect, we proposed two VCDP models for parameter estimation and inference. We developed a two-step method and a bootstrap-based testing procedure for the inference of CQTE. 
Further, we extended the proposed procedure to accommodate spatiotemporal data, decomposing the CQTE into the sum of CQDE and CQIE. We established consistency results for parameter estimation and the test procedure. Through the analysis of real datasets obtained from DiDi Chuxing, we demonstrated that our proposed method is a valuable statistical tool for assessing the dynamic QTE of new policies.

\bibliographystyle{chicago}
\bibliography{references}

\appendix

\section{Tables for simulation results}
\label{sec:tables}
We report Tables \ref{table:simu_QTE} to \ref{table:simu_QIE} in this section. These tables contain empirical rejection rates of the proposed tests for CQTE, CQDE and CQIE in simulation studies, respectively, based on 500 simulation replications. 
Figures \ref{fig:QDE_simu} and \ref{fig:QIE_simu} depict the empirical rejection rates for
CQDE and CQIE.

\begin{table}[htbp]
	\begin{center}
		\caption{\small Empirical rejection rates of the proposed test for CQTE$_\tau$ with standard errors in parentheses 
			based on 500 replicates.}
		\label{table:simu_QTE}
		
		\vspace{4ex}
		\scriptsize
		\begin{tabular}{ccccccccc}
			\hline
			\hline
			$\tau$	&	TI  &      $n$    &   0   &   0.001    &  0.025  &  0.050  &   0.075   &   0.100 \\
			\hline
			0.2 &	1 &	20 &	0.038(0.009) & 	0.084(0.012) & 	0.196(0.018) & 	0.474(0.022) & 	0.750(0.019) & 	0.954(0.009) \\
			&	  &	40 &	0.049(0.010) & 	0.148(0.016) & 	0.381(0.022) & 	0.844(0.016) & 	0.984(0.006) & 	1.000(0.000) \\
			\cline{2-9}
			&	3 &	20 &	0.052(0.010) & 	0.104(0.014) & 	0.258(0.020) & 	0.622(0.022) & 	0.926(0.012) & 	0.990(0.004) \\
			&	  &	40 &	0.057(0.010) & 	0.164(0.017) & 	0.457(0.022) & 	0.927(0.012) & 	0.998(0.002) & 	1.000(0.000) \\ 
			\hline
			0.5 &	1 &	20 &	0.036(0.008) & 	0.062(0.011) & 	0.268(0.020) & 	0.718(0.020) & 	0.964(0.008) & 	0.996(0.003) \\ 
			&	  &	40 &	0.040(0.009) & 	0.192(0.018) & 	0.684(0.021) & 	0.996(0.003) & 	1.000(0.000) & 	1.000(0.000) \\ 
			\cline{2-9}
			&	3 &	20 &	0.044(0.009) & 	0.130(0.015) & 	0.352(0.021) & 	0.870(0.015) & 	0.992(0.004) & 	1.000(0.000) \\ 
			&	  &	40 &	0.036(0.008) & 	0.196(0.018) & 	0.747(0.019) & 	0.998(0.002) & 	1.000(0.000) & 	1.000(0.000) \\ 
			\hline
			0.8 &	1 &	20 &	0.042(0.009) & 	0.100(0.013) & 	0.400(0.022) & 	0.854(0.016) & 	0.996(0.003) & 	1.000(0.000) \\ 
			&	  &	40 &	0.059(0.011) & 	0.272(0.020) & 	0.828(0.017) & 	1.000(0.000) & 	1.000(0.000) & 	1.000(0.000) \\ 
			\cline{2-9}
			&	3 &	20 &	0.050(0.010) & 	0.144(0.016) & 	0.476(0.022) & 	0.970(0.008) & 	1.000(0.000) & 	1.000(0.000) \\ 
			&	  &	40 &	0.061(0.011) & 	0.312(0.021) & 	0.846(0.016) & 	1.000(0.000) & 	1.000(0.000) & 	1.000(0.000) \\ 
			\hline \hline
		\end{tabular}
	\end{center}
\end{table}

\begin{table}[htbp]
	\begin{center}
		\caption{\small Empirical rejection rates of the proposed test for CQDE$_\tau$ with standard errors in parentheses 
			based on 500 replicates.}
		\label{table:simu_QDE}
		
		\vspace{4ex}
		\scriptsize
		\begin{tabular}{ccccccccc}
			\hline
			\hline
			$\tau$	&	TI  &      $n$    &   0   &   0.001    &  0.025  &  0.050  &   0.075   &   0.100 \\
			\hline
			0.2 &	1 &	20 &	0.072(0.012) & 	0.158(0.016) & 	0.380(0.022) & 	0.706(0.020) & 	0.924(0.012) & 	0.998(0.002) \\
			&	  &	40 &	0.075(0.012) & 	0.236(0.019) & 	0.524(0.022) & 	0.929(0.011) & 	0.996(0.003) & 	1.000(0.000) \\
			\cline{2-9}
			&	3 &	20 &	0.074(0.012) & 	0.164(0.017) & 	0.346(0.021) & 	0.790(0.018) & 	0.972(0.007) & 	1.000(0.000) \\
			&	  &	40 &	0.075(0.012) & 	0.216(0.018) & 	0.524(0.022) & 	0.957(0.009) & 	0.998(0.002) & 	1.000(0.000) \\
			\hline
			0.5 &	1 &	20 &	0.064(0.011) & 	0.152(0.016) & 	0.524(0.022) & 	0.936(0.011) & 	0.998(0.002) & 	1.000(0.000) \\
			&	  &	40 &	0.056(0.010) & 	0.256(0.020) & 	0.828(0.017) & 	0.998(0.002) & 	1.000(0.000) & 	1.000(0.000) \\
			\cline{2-9}
			&	3 &	20 &	0.064(0.011) & 	0.206(0.018) & 	0.484(0.022) & 	0.958(0.009) & 	0.998(0.002) & 	1.000(0.000) \\
			&	  &	40 &	0.050(0.010) & 	0.278(0.020) & 	0.856(0.016) & 	1.000(0.000) & 	1.000(0.000) & 	1.000(0.000) \\
			\hline
			0.8 &	1 &	20 &	0.054(0.010) & 	0.198(0.018) & 	0.582(0.022) & 	0.940(0.011) & 	1.000(0.000) & 	1.000(0.000) \\
			&	  &	40 &	0.059(0.011) & 	0.312(0.021) & 	0.828(0.017) & 	1.000(0.000) & 	1.000(0.000) & 	1.000(0.000) \\
			\cline{2-9}
			&	3 &	20 &	0.078(0.012) & 	0.210(0.018) & 	0.536(0.022) & 	0.944(0.010) & 	0.994(0.003) & 	1.000(0.000) \\
			&	  &	40 &	0.059(0.011) & 	0.284(0.020) & 	0.818(0.017) & 	0.998(0.002) & 	1.000(0.000) & 	1.000(0.000) \\
			\hline \hline
		\end{tabular}
	\end{center}
\end{table}

\begin{table}[htbp]
	\begin{center}
		\caption{\small Empirical rejection rates of the proposed test for CQIE$_\tau$ with standard errors in parentheses 
			based on 500 replicates.}
		\label{table:simu_QIE}
		
		\vspace{4ex}
		\scriptsize
		\begin{tabular}{ccccccccc}
			\hline
			\hline
			$\tau$	&	TI  &      $n$    &   0   &   0.001    &  0.025  &  0.050  &   0.075   &   0.100 \\
			\hline
			0.2 &	1 &	20 &	0.044(0.009) & 	0.040(0.009) & 	0.030(0.008) & 	0.034(0.008) & 	0.062(0.011) & 	0.090(0.013) \\ 
			&	  &	40 &	0.048(0.010) & 	0.046(0.009) & 	0.034(0.008) & 	0.103(0.014) & 	0.170(0.017) & 	0.302(0.021) \\ 
			\cline{2-9}
			&	3 &	20 &	0.036(0.008) & 	0.038(0.009) & 	0.044(0.009) & 	0.074(0.012) & 	0.110(0.014) & 	0.138(0.015) \\
			&	  &	40 &	0.040(0.009) & 	0.046(0.009) & 	0.071(0.011) & 	0.113(0.014) & 	0.200(0.018) & 	0.346(0.021) \\
			\hline
			0.5 &	1 &	20 &	0.030(0.008) & 	0.032(0.008) & 	0.038(0.009) & 	0.078(0.012) & 	0.130(0.015) & 	0.216(0.018) \\
			&	  &	40 &	0.030(0.008) & 	0.054(0.010) & 	0.105(0.014) & 	0.233(0.019) & 	0.393(0.022) & 	0.583(0.022) \\
			\cline{2-9}
			&	3 &	20 &	0.038(0.009) & 	0.046(0.009) & 	0.048(0.010) & 	0.102(0.014) & 	0.178(0.017) & 	0.310(0.021) \\
			&	  &	40 &	0.034(0.008) & 	0.048(0.010) & 	0.117(0.014) & 	0.269(0.020) & 	0.470(0.022) & 	0.690(0.021) \\
			\hline
			0.8 &	1 &	20 &	0.038(0.009) & 	0.032(0.008) & 	0.078(0.012) & 	0.186(0.017) & 	0.400(0.022) & 	0.564(0.022) \\
			&	  &	40 &	0.034(0.008) & 	0.072(0.012) & 	0.209(0.018) & 	0.625(0.022) & 	0.879(0.015) & 	0.982(0.006) \\
			\cline{2-9}
			&	3 &	20 &	0.044(0.009) & 	0.056(0.010) & 	0.120(0.015) & 	0.298(0.020) & 	0.538(0.022) & 	0.752(0.019) \\
			&	  &	40 &	0.034(0.008) & 	0.092(0.013) & 	0.249(0.019) & 	0.700(0.020) & 	0.957(0.009) & 	0.996(0.003) \\
			\hline \hline
		\end{tabular}
	\end{center}
\end{table}

\begin{figure}[!h]
	\centering
	\includegraphics[height=4cm, width=5cm]{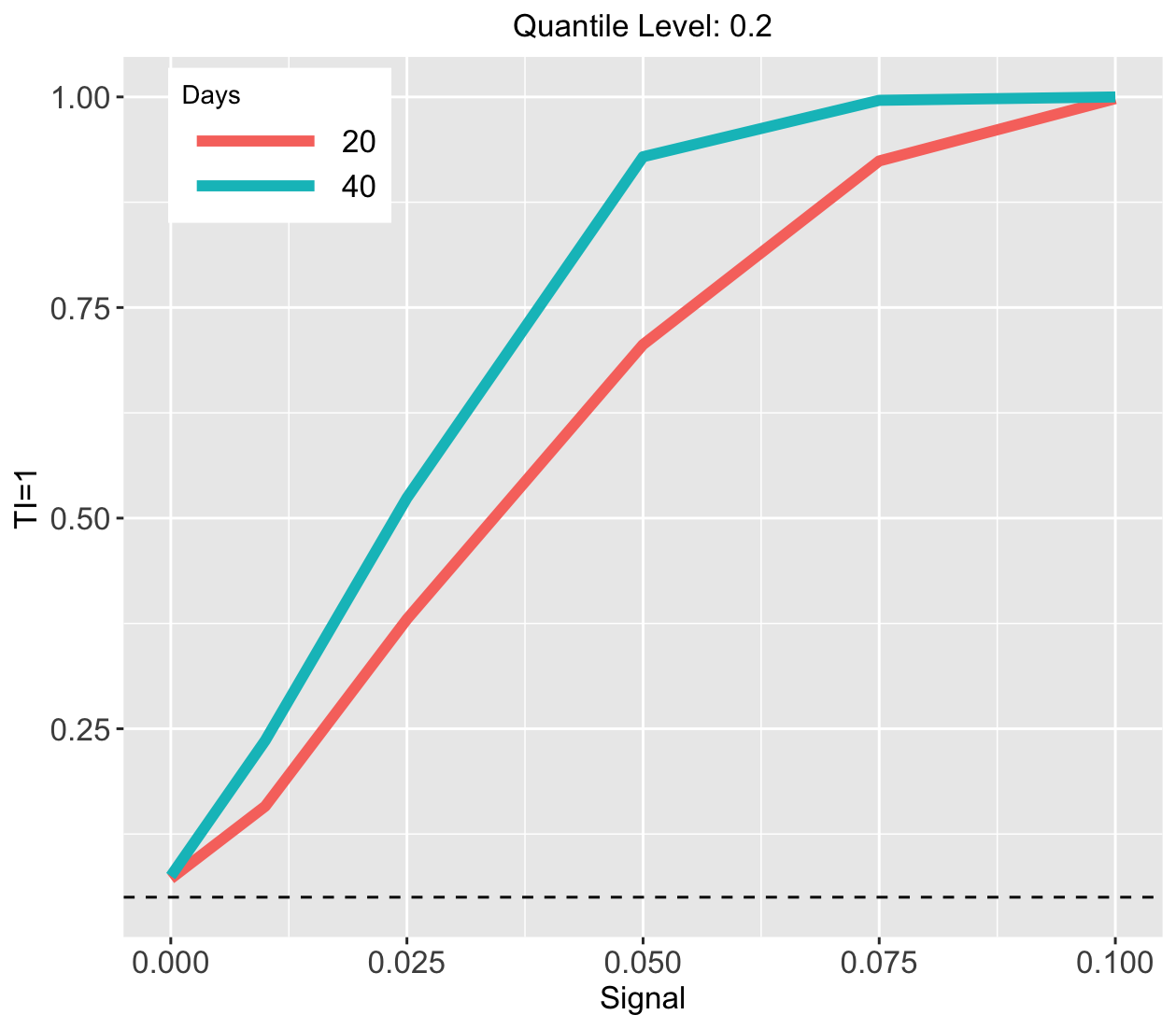}
	\includegraphics[height=4cm, width=5cm]{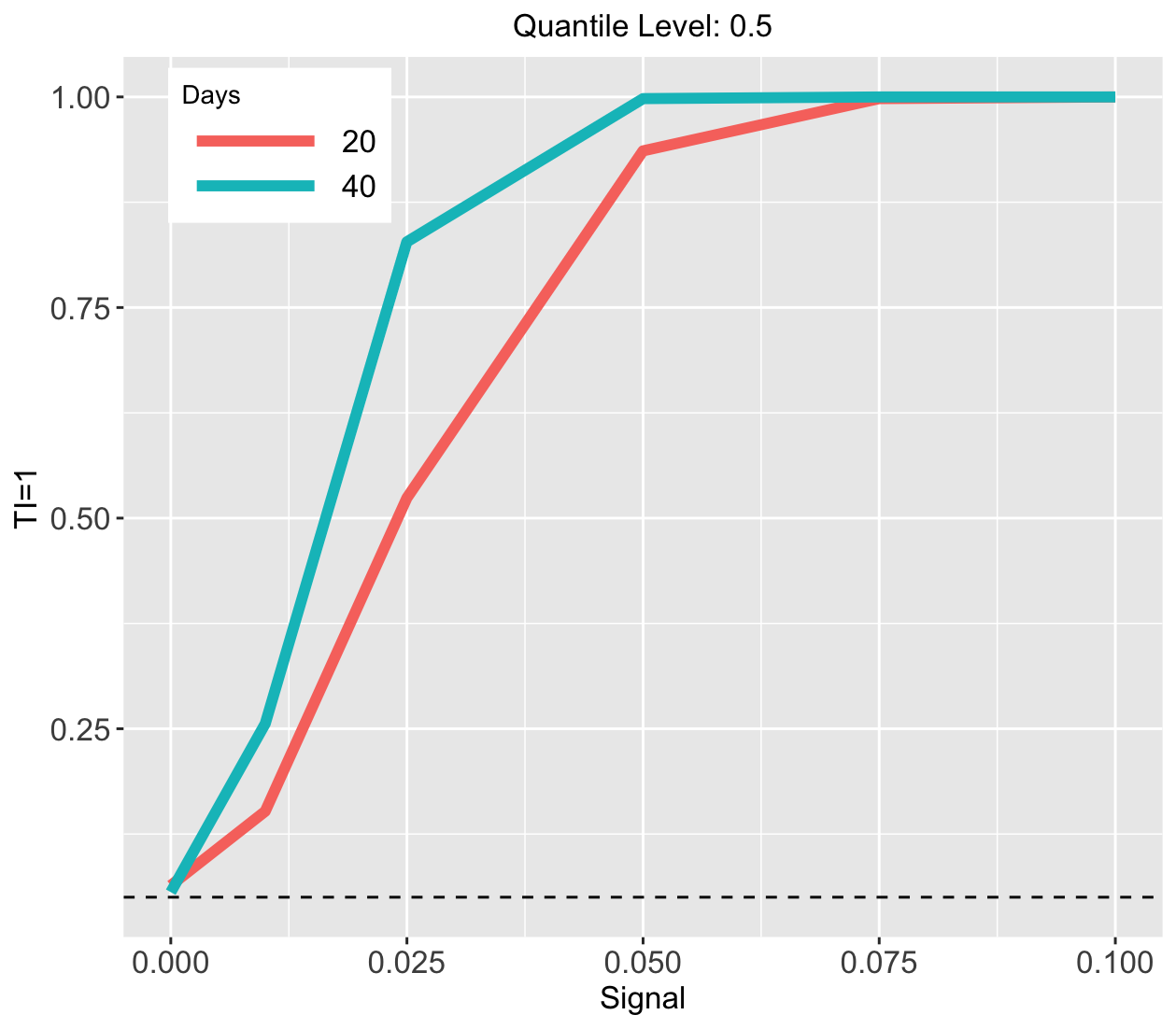}
	\includegraphics[height=4cm, width=5cm]{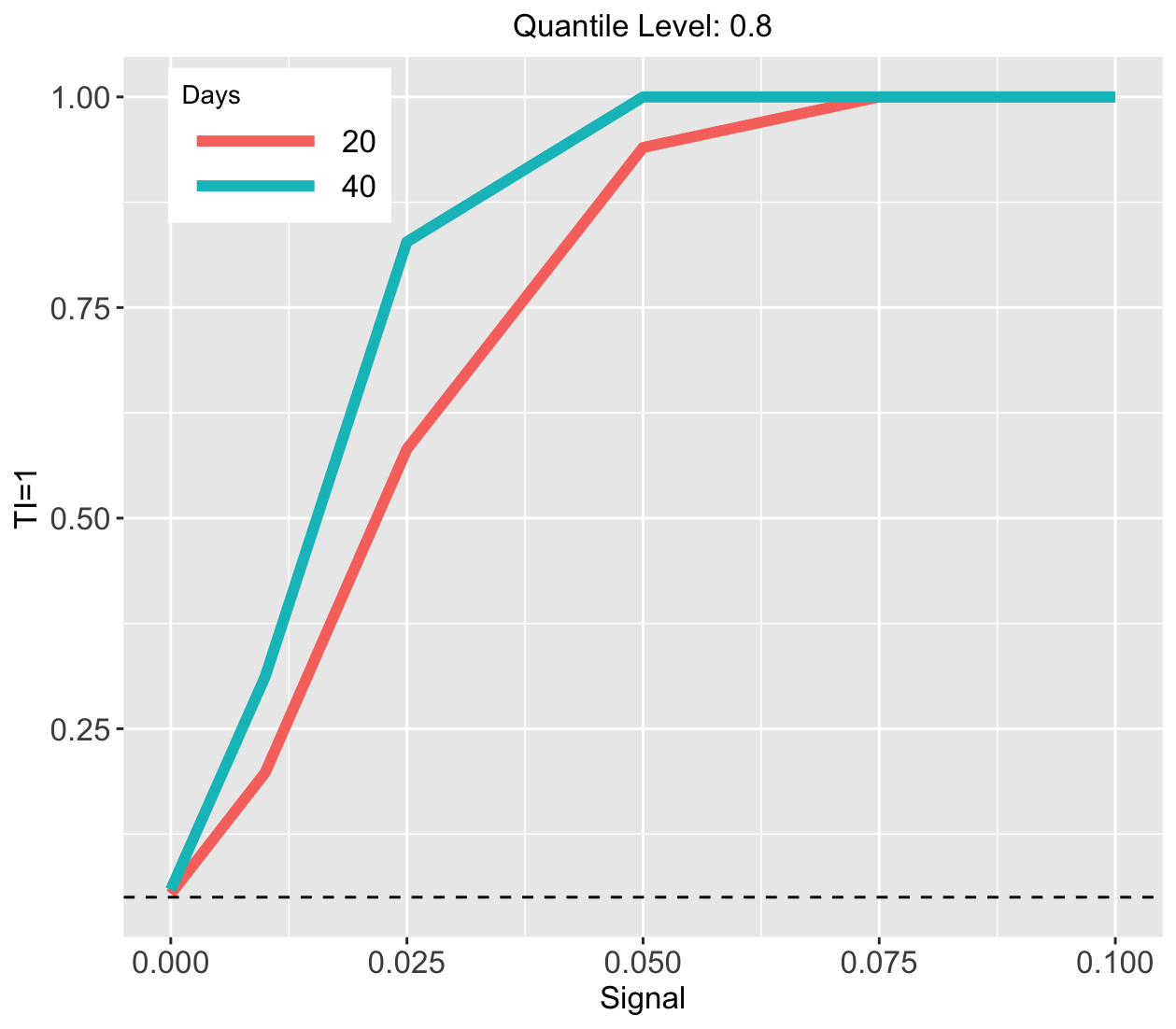} \\
	\includegraphics[height=4cm, width=5cm]{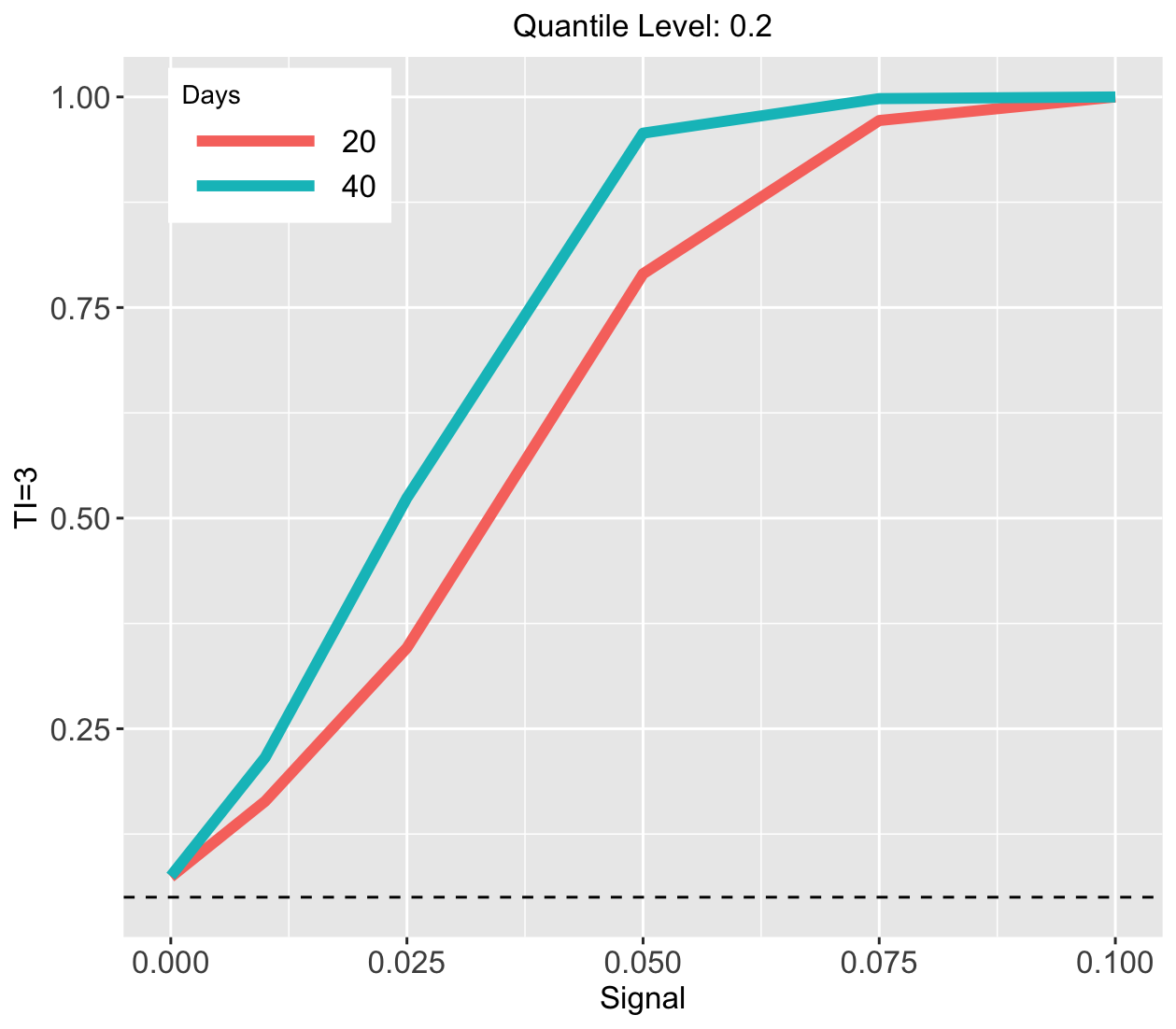}
	\includegraphics[height=4cm, width=5cm]{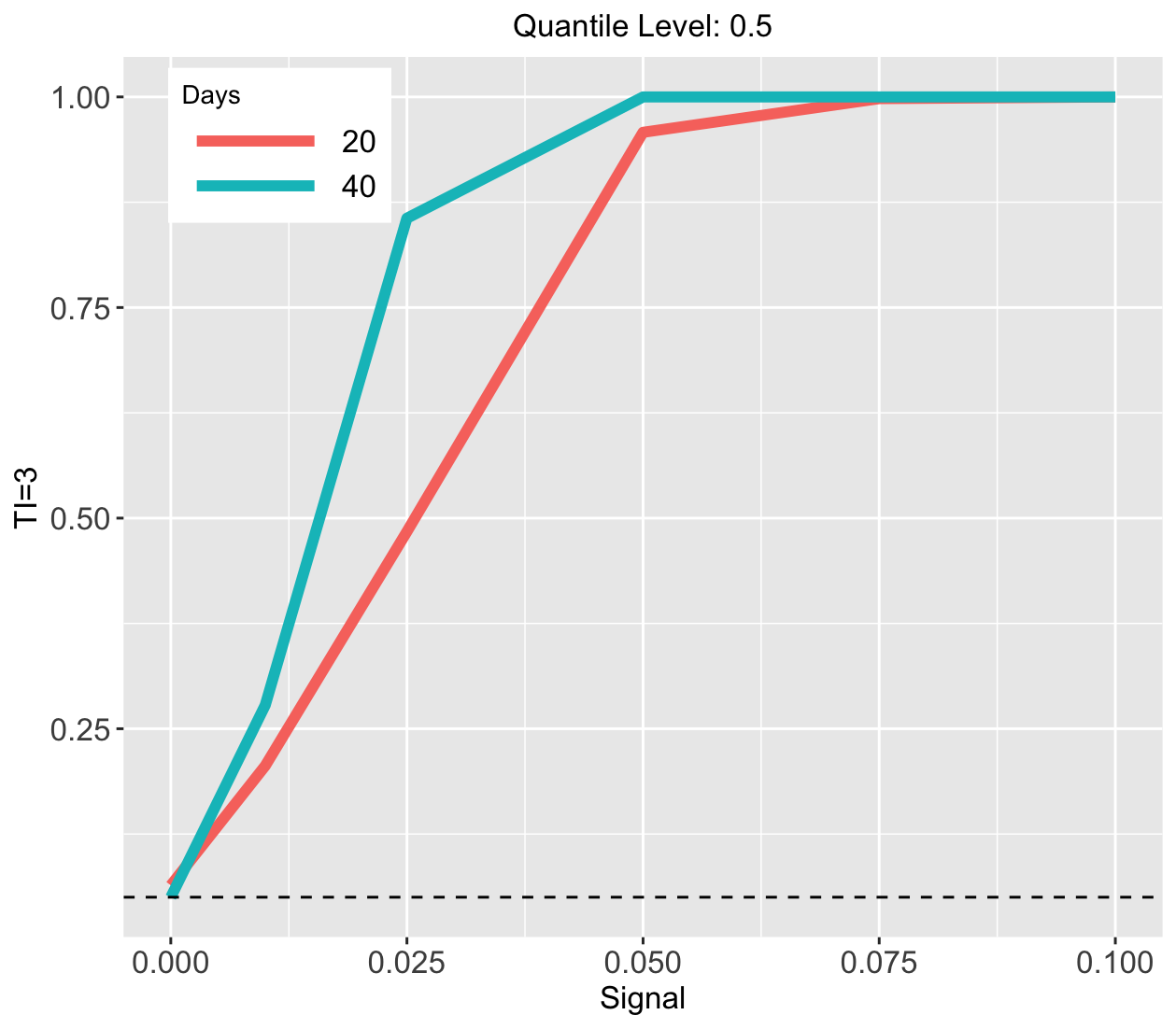}
	\includegraphics[height=4cm, width=5cm]{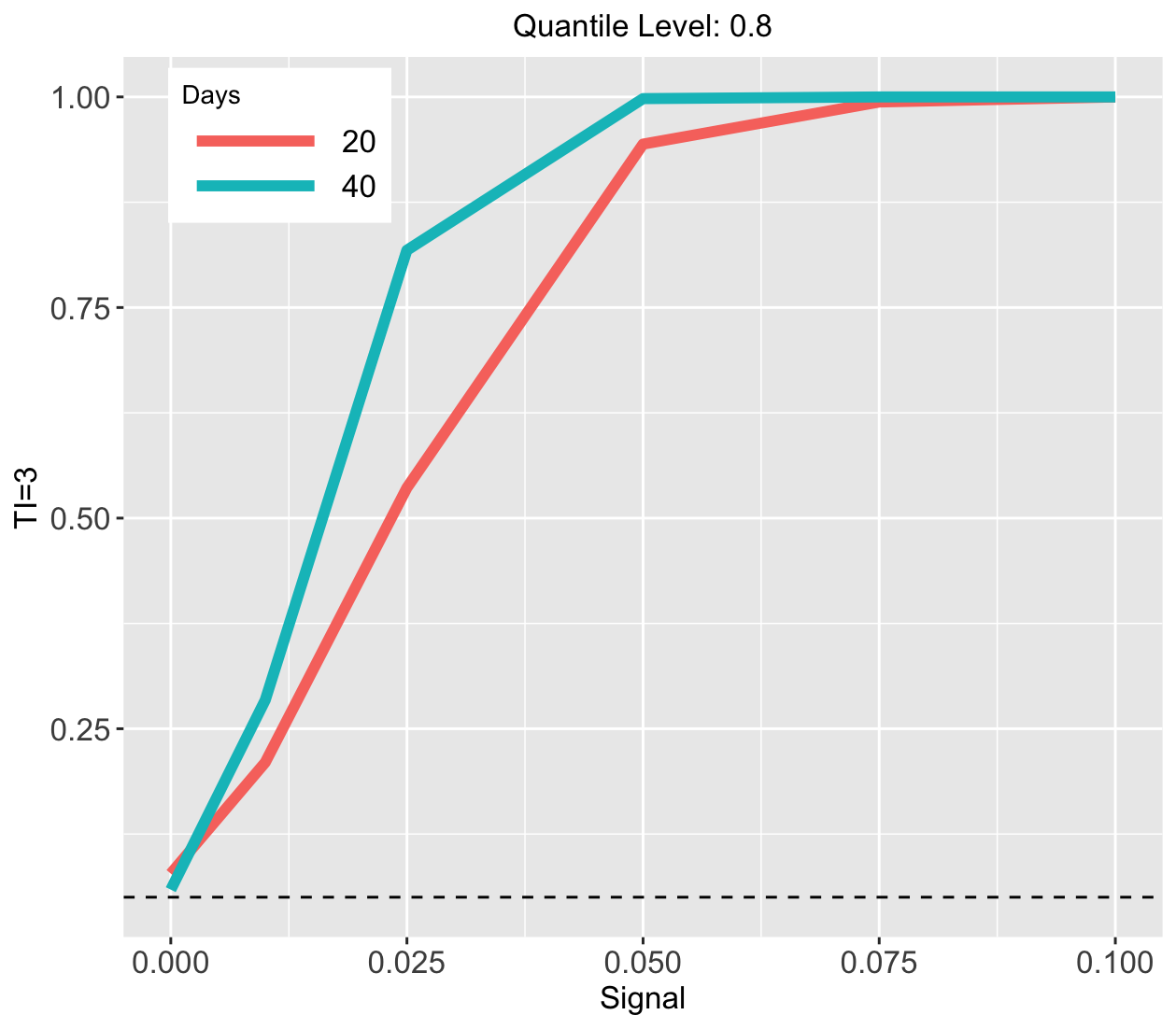}
	\caption{\small Empirical rejection rates for CQDE$_\tau$. TI equals $1$ for the top panels and $3$ for the bottom panels. The quantile level $\tau=0.2$, $0.5$ and $0.8$, from left to right.} 
\label{fig:QDE_simu}
\end{figure}

\begin{figure}
\centering
\includegraphics[height=4cm, width=5cm]{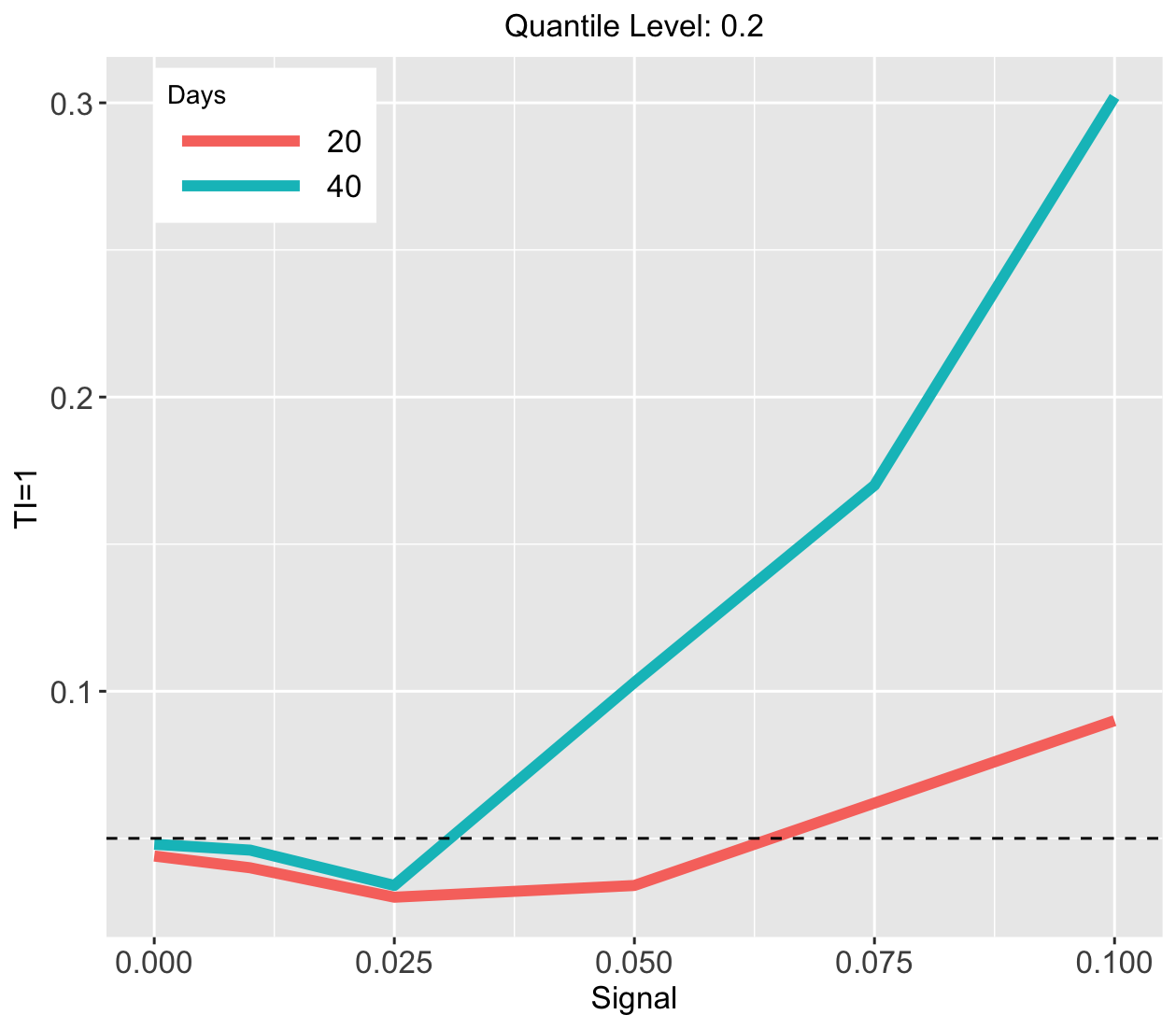}
\includegraphics[height=4cm, width=5cm]{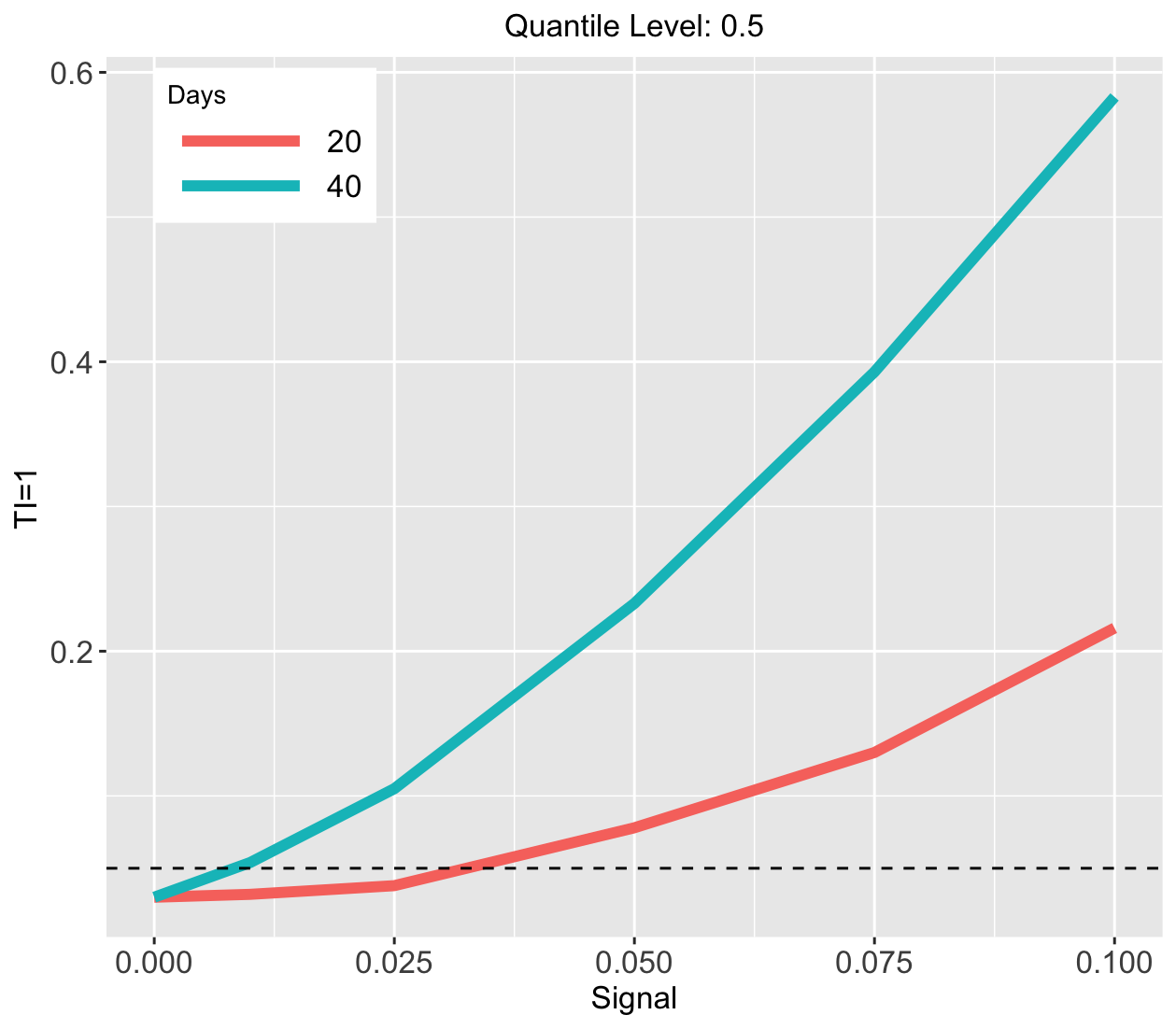}
\includegraphics[height=4cm, width=5cm]{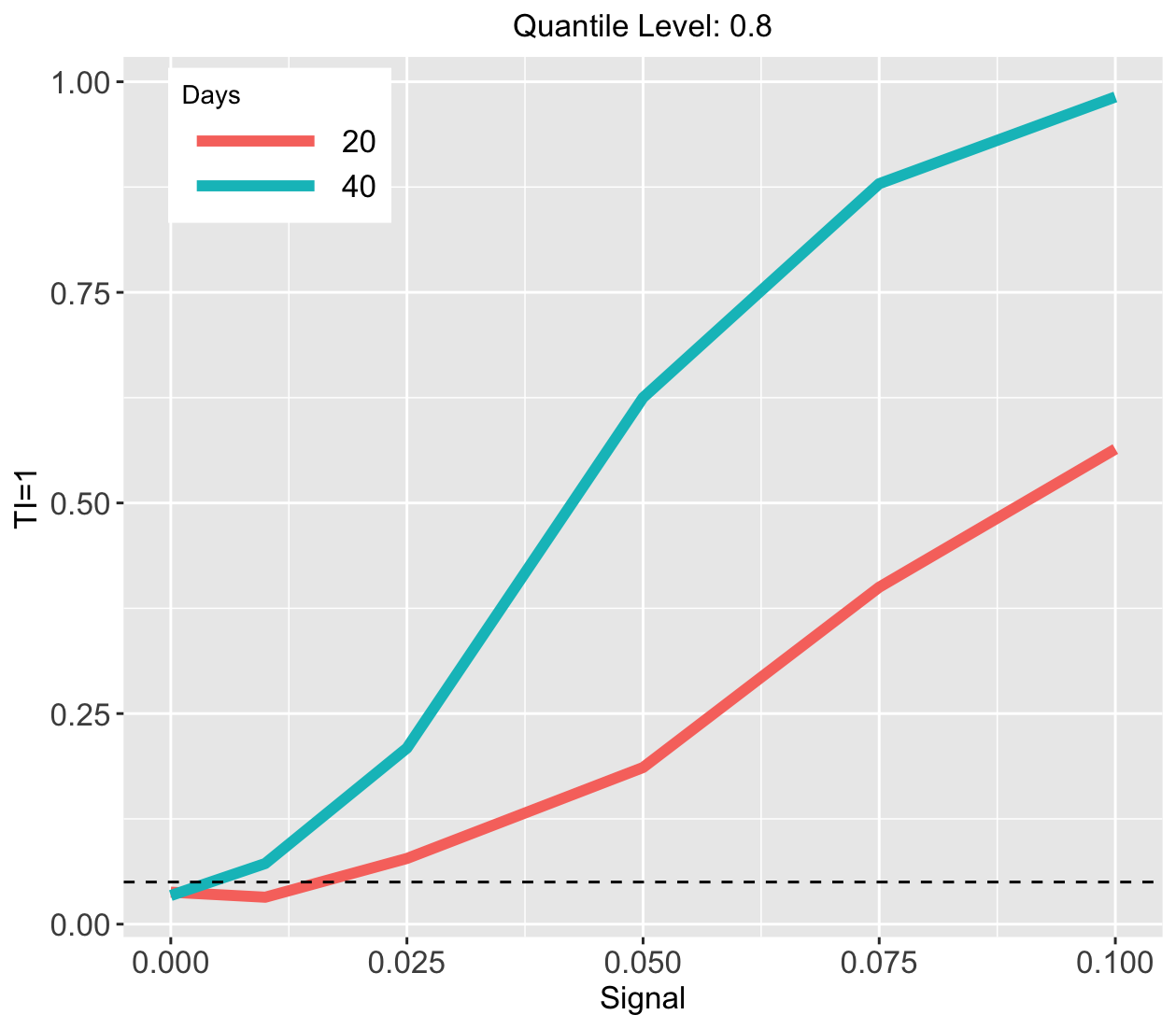} \\
\includegraphics[height=4cm, width=5cm]{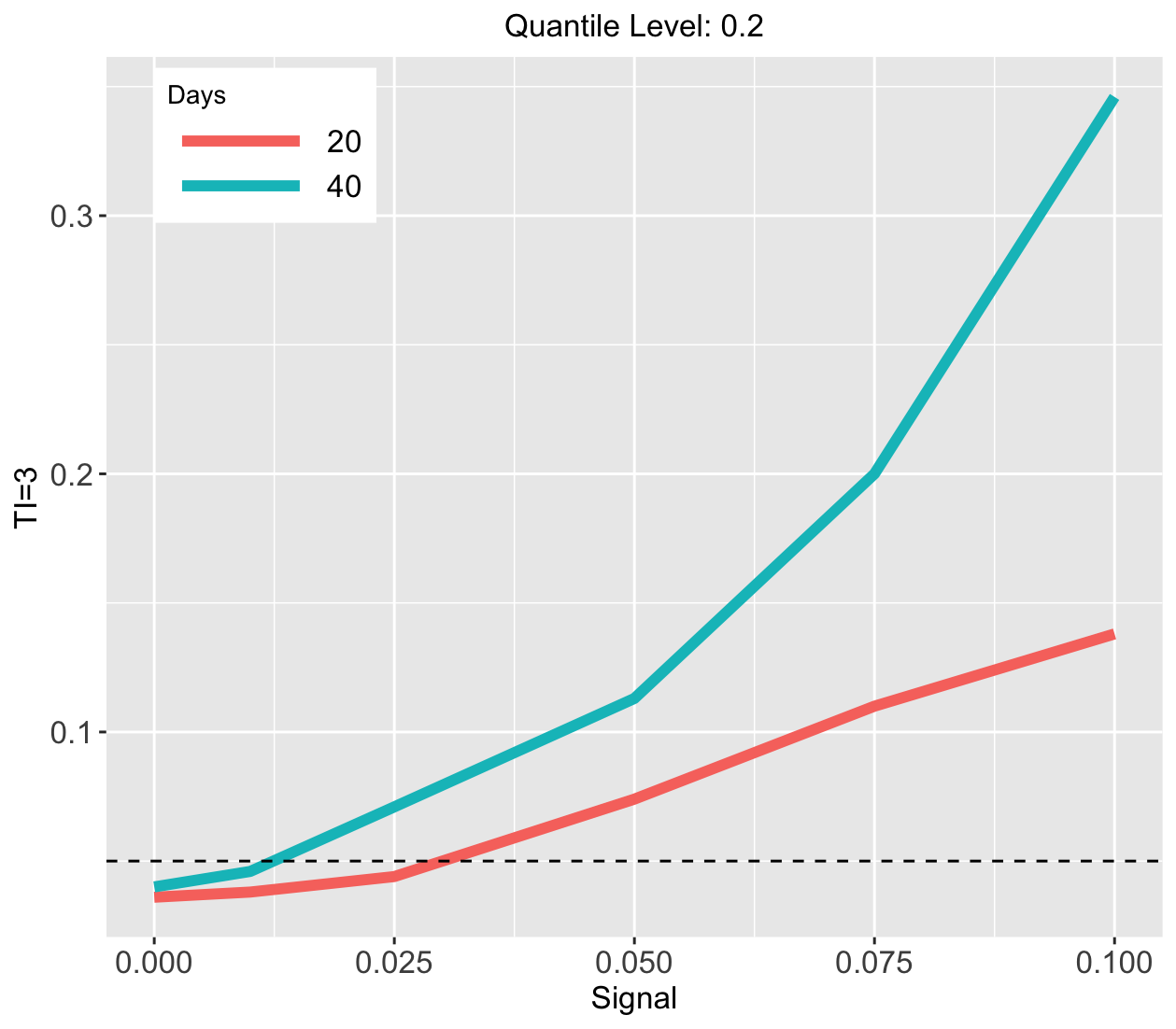}
\includegraphics[height=4cm, width=5cm]{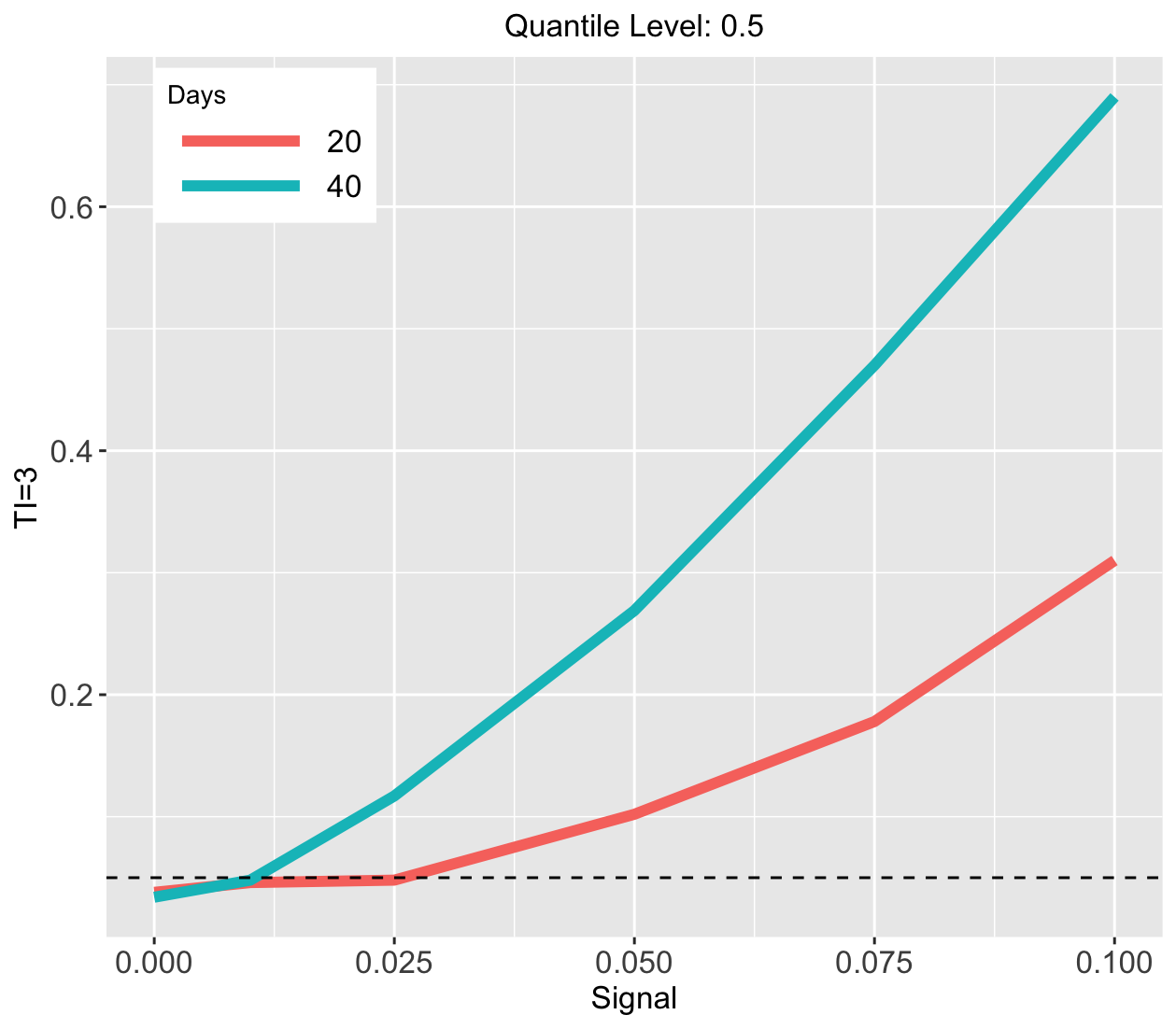}
\includegraphics[height=4cm, width=5cm]{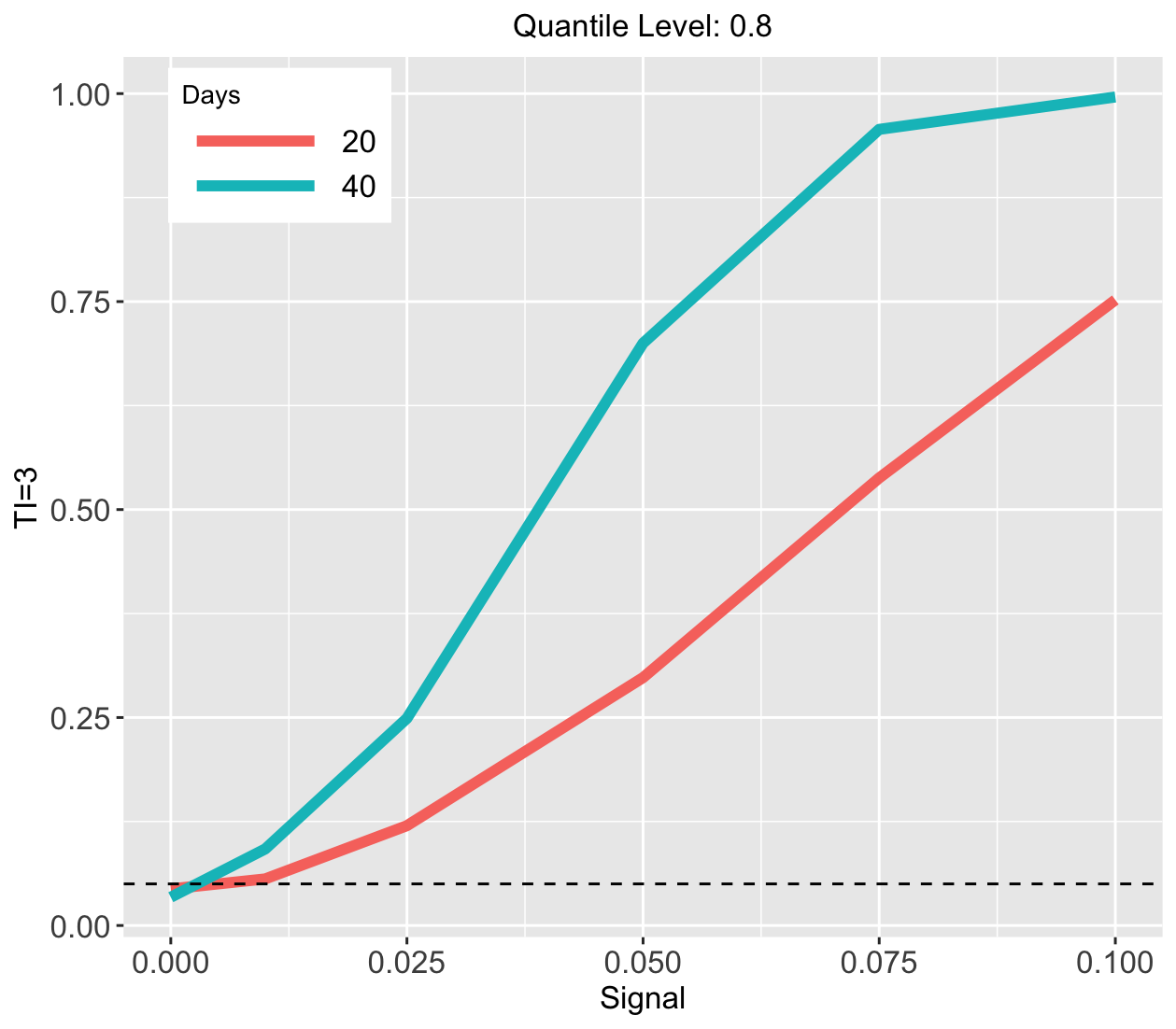}
\caption{\small Empirical rejection rates for CQIE$_\tau$. TI equals $1$ for the top panels and $3$ for the bottom panels. The quantile level $\tau=0.2$, $0.5$ and $0.8$, from left to right plots.} 
\label{fig:QIE_simu}
\end{figure}

\section{Regularity assumptions}
In this section, we present the assumptions for studying the asymptotic properties of the estimators. First, we introduce some notations. Let $F_e$ and $f_e$ denote the probability distribution and density function of the error process. We rescale the time and set $s=t/m$ for $t=1,\cdots,m$. It follows that $s\in [0,1]$. In temporal dependent experiments, we define
\begin{eqnarray*}
F_e(\kappa, s; \tau) &=& P( e_\tau (ms) \leq \kappa), \quad  	F_e(\kappa_1, \kappa_2, s_1, s_2; \tau_1, \tau_2)= P( e_ {\tau_1} (ms_1) \leq \kappa_1 , e_{\tau_2}(ms_2) \leq \kappa_2  ), \\
f_e(\kappa, s; \tau) &=& \frac{\partial F_e(\kappa, s; \tau)}{\partial \kappa}, \quad 
f_e(\kappa_1, \kappa_2, s_1, s_2; \tau_1, \tau_2) = \frac{ \partial^2 F_e(\kappa_1, \kappa_2, s_1, s_2; \tau_1, \tau_2)}{\partial \kappa_1 \partial \kappa_2 }.
\end{eqnarray*}
In spatiotemporal dependent experiments, we also 
set $l=\iota/r$ for $\iota=1,\cdots,r$ such that 
$l \in [0,1]$.
We define $F_e(\kappa, s, l; \tau) = P( e_\tau (ms, r\iota) \leq \kappa)$, $	f_e(\kappa, s, l; \tau) = { \partial F_e(\kappa, s, \iota; \tau)}/{\partial \kappa}, $  $F_e  (\kappa_1, \kappa_2, s_1, s_2, l_1, l_2; \tau_1, \tau_2)= P( e_ {\tau_1} (m s_1, r l_1) \leq \kappa_1 , e_{\tau_2}(m s_2, r l_2) \leq \kappa_2  )$,  and
\begin{eqnarray*}
\quad
f_e(\kappa_1, \kappa_2, s_1, s_2, l_1, l_2; \tau_1, \tau_2) = \frac{ \partial^2 F_e(\kappa_1, \kappa_2, s_1, s_2, l_1, l_2; \tau_1, \tau_2)}{\partial \kappa_1 \partial \kappa_2 }.
\end{eqnarray*}

Second, we list the following assumptions to guarantee the theoretical results of the proposed test in temporal dependent experiments. Notice that we allow $m$ to grow with $n$.

\begin{asmp}
\label{assump:kernel}
The kernel function $K(\cdot)$ is a symmetric probability density function defined over the interval $ [-1,1] $. It is Lipschitz continuous and satisfies the condition $\int_{-1}^1 | t K^\prime(t )| dt < \infty$, signifying its finiteness under the integral of its absolute derivative. 
\end{asmp}

\begin{asmp}
\label{assump:F_e}
The probability density function $f_e(\kappa, s; \tau)$ is strictly positive and continuously varies as a function of $\kappa$. It is twice differentiable at any $s$, possessing a second-order derivative that is both continuous and uniformly bounded. 
Similarly, the joint probability density function $f_e(\kappa_1, \kappa_2, s_1, s_2; \tau_1, \tau_2)$ is strictly positive and continuously varies as a function of $(\kappa_1,\kappa_2)$. It is twice differentiable at any pair $(s_1,s_2)$, maintaining a second-order derivative that is continuous and uniformly bounded.	
\end{asmp}

\begin{asmp}
\label{assump:Z}
The covariates $Z_i$ are drawn independently and identically from a sub-Gaussian process. Furthermore, for any integer $t$ within the range $1\leq t\leq m$, the smallest eigenvalue of the expected value matrix $\Mean(Z_{i,t} Z_{i,t}^\top)\in\M^{p\times p}$ remains greater than zero.  
\end{asmp}

\begin{asmp}
\label{assump:coe}
All components of $\theta_\tau (ms)$ and $\Theta(ms)$ possess second-order derivatives with respect to $s$ that are not only bounded but also continuous for each $\tau$. 
\end{asmp}	

\begin{asmp}
\label{asmp:st1}
For any $\tau$ in the interval [0, 1], there exist constants $0<q<1$, $M_{\Gamma},M_{\beta}>0$, and $M_{min}>0$ such that   $\Vert\Phi(t)\Vert_\infty\leq q$, $\Vert\Gamma(t)\Vert_\infty\leq M_{\Gamma}$, and $\Vert\beta_\tau(t)\Vert_\infty\leq M_\beta$.
\end{asmp}


Assumption \ref{assump:kernel}, which is often found in the literature on varying coefficient models and kernel smoothing, pertains to regular conditions in kernel methods \citep[see e.g.,][]{zhu2012multivariate, zhu2014spatially, cao2015regression}.
Assumption \ref{assump:F_e} addresses the distribution function of the error process, a condition frequently stipulated in the quantile regression literature \citep{koenker2001quantile, cai2012semiparametric,ma2019quantile}. This assumption is crucial for establishing the necessary conditions on the distribution and density functions of the error process.
Assumption \ref{assump:Z} ensures the well-defined nature of the first-step estimators, essentially guaranteeing the positive definiteness of the moment of the design matrix.
Assumption \ref{assump:coe} is a common requirement in the varying coefficient literature. It mandates the smoothness of the varying coefficients \citep{zhu2012multivariate,zhu2014spatially,zhou2021quantile}. 
Finally, Assumption \ref{asmp:st1} is put forth to ensure the stationarity of the observation vector. This type of condition is frequently imposed in the time series literature \citep[see e.g.,][]{brockwell2009time}. 

Finally, we present technical assumptions that establish the consistency of the proposed test in a spatiotemporal dependent experiment. We allow both $m$ and $r$ to grow with $n$.

\begin{asmp}
\label{assump:F_e_spatial}
The probability density function $f_e(\kappa, s, l; \tau)$ is strictly positive, continuous as a function of $\kappa$, and twice differentiable at any $(s, l)$ with continuous and uniformly bounded second-order derivative. Moreover, the joint probability density function $f_e(\kappa_1, \kappa_2, s_1, s_2, l_1, l_2; \tau_1, \tau_2)$ is strictly positive, continuous at $(\kappa_1,\kappa_2)$, and twice differentiable at any $(s_1, s_2, l_1, l_2)$ with continuous and uniformly bounded second-order derivative.
\end{asmp}

\begin{asmp}
\label{assump:Z_spatial}
The covariate ${Z}i$s are independently and identically distributed from a sub-Gaussian process. Furthermore, for $1\leq t\leq m$ and $1 \leq \iota \leq r$, the minimum eigenvalue of $\Mean(Z{i,t, \iota} Z_{i,t, \iota}^\top)\in\M^{p\times p}$ is bounded away from zero.
\end{asmp}

\begin{asmp}
\label{assump:coe_spatial}
All components of $\theta_\tau(m s, r l)$ and $\Theta( m s, r l)$ have bounded and continuous second-order derivatives with respect to $(s, l)$ for each $\tau$.
\end{asmp}

\begin{asmp}
\label{asmp:st1_spatial}
There exist some constants $0<q<1$, $M_{\Gamma},M_\beta>0$ such that $\Vert\Phi(t, \iota)\Vert_\infty\leq q$,
$\Vert \Gamma_{1,\tau} (t, \iota) \Vert_\infty\leq M_{\Gamma}$ and $\Vert\beta_\tau (t, \iota)\Vert_\infty\leq M_\beta$.
\end{asmp}



Assumptions \ref{assump:F_e_spatial}--\ref{asmp:st1_spatial}  
reduce to
those in Assumptions \ref{assump:F_e}--\ref{asmp:st1} 
when $r=1$.


\section{Asymptotic properties of the test procedures}
\label{sec: decomposition theoretical results}

In this section, we establish the asymptotic properties of the proposed test procedures in Section \ref{sec:DEIE}. We first establish the asymptotic normality of the one-step and second-step estimators, $\widehat{\bm{\theta}}(\tau) = \{ \widehat \theta_\tau(1)^\top, \dots, \widehat \theta_\tau(m)^\top\}^\top$ and $\widetilde{\bm{\theta}}(\tau) = \{\widetilde \theta_\tau(1)^\top, \dots, \widetilde \theta_\tau(m)^\top\}^\top$ in the temporal dependent experiment, and prove the validity of the test procedure for QDE$_{\tau}$.

Let's define $\Lambda(\tau) = \diag[f_e( 0,t_1; \tau), \dots, f_e( 0,t_m; \tau)]$, and $G$ as an $m \times m$ matrix where the $(i,j)$th entry equals $F_e(0,0, t_i, t_j; \tau)$. Furthermore, we define
\begin{eqnarray}\label{eq:V theta}
\bm{V}_{\widehat{\theta}\tau}=(\Mean \bm{Z}{i}^\top \Lambda(\tau) \bm{Z}i)^{-1}\Mean (\bm{Z}{i}^\top G \bm{Z}{i})
(\Mean \bm{Z}{i}^\top \Lambda(\tau) \bm{Z}i)^{-1}~~~~
\mbox{and}~~~~
\bm{V}_{\widetilde{\theta}\tau} = {m-1 \over m} \bm{V}_{\widehat{\theta}\tau}.
\end{eqnarray}
These two matrices represent the asymptotic covariance matrices of the one-step and two-step estimators, respectively.

\begin{thm}
\label{thm:theta_t asymp}
Suppose $\lambda_{\min}(\bm{V}_{\widetilde{\theta}, \tau})$ is bounded away from zero. 
Under Assumptions \ref{assump:kernel}--\ref{assump:coe} of the supplement material, for any $m(d+2)$-dimensional nonzero 
vectors $a_{n,1}$ and  $a_{n,2}$ with unit $\ell_2$ norm, 
we have as $n,m \rightarrow \infty$, $h\rightarrow0$, $mh \rightarrow \infty$,
\begin{itemize}
\item[(i)] $\sqrt{n} a_{n,1}^\top (\widehat{\bm{\theta}}_\tau - {\bm{\theta}}_\tau )/\sqrt{ {a}_{n,1}^\top \bm{V}_{\widehat{\theta}, \tau} {a}_{n,1}}\stackrel{d}{\longrightarrow} N(0,1)$. 

\item[(ii)] $\sqrt{n} {a}_{n,2}^\top (\widetilde{\bm{\theta}}_\tau - \bm{\theta}_\tau)/\sqrt{ {a}_{n,2}^\top\bm{V}_{\widetilde{\theta}, \tau} {a}_{n,2}}-b_n \stackrel{d}{\longrightarrow}N(0,1)$, where the bias $b_n=O(\sqrt{n}(h^2+m^{-1}))$. 

\item[(iii)]
Suppose $h=o(n^{-1/4})$ and $m \gg \sqrt{n}$. Then the conditional distribution 
of $\sqrt{n} ( \sum_{t=1}^m \tilde \gamma_\tau^b (t) - \sum_{t=1}^m \tilde \gamma_\tau (t)  )$ given the observed data weakly
converges to the null limit distribution of $\widehat{\textrm{CQDE}}_{\tau}= \sqrt{n}  (\sum_{t=1}^m \tilde \gamma_\tau (t) -\textrm{CQDE}_{\tau})$.

\item[(iv)]
If further Assumption \ref{asmp:st1} of the supplement material holds and $h=o(n^{-1/4})$, $m\asymp n^{c_2}$ for some $1/2< c_2<3/2$, 
then with probability approaching $1$,  there exist some constant $\varepsilon\in (0,1)$ and some positive constant $C$ such that
\begin{eqnarray*}
\sup_{ \tau \in [\varepsilon, 1- \varepsilon] }\sup_{z}|\prob(\widehat{\QIE}_\tau-\textrm{CQIE}_{\tau}  \leq z)-\prob( \widehat{\QIE}_\tau^b-\widehat{\QIE}_\tau \leq z |\textrm{Data}) \vert \\
\leq  C(\sqrt{n}h^2 + \sqrt{n} m^{-1} 
+n^{-1/8}). 
\end{eqnarray*}
\end{itemize}
\end{thm}

Part (i) of Theorem \ref{thm:theta_t asymp} is established using standard arguments in the quantile regression literature \citep[see e.g.,][]{koenker2001quantile}. This part of the theorem provides the limit normal distribution of the raw estimators. 
Part (ii) further shows that the smoothing step introduces only a negligible bias and reduces the asymptotic covariance by a factor of $1/m$. These two results imply that the proposed test statistic for CQDE $\sum_{t=1}^m \tilde \gamma_\tau (t)$ is asymptotically normal. However, as noted earlier, consistently estimating its asymptotic variance remains a challenge due to its complex form. 
Part (iii) of Theorem \ref{thm:theta_t asymp} validates the proposed bootstrap procedure for CQDE. Instead of using the empirical quantile of the bootstrapped statistics to determine the critical value, our test procedure relies on normal approximation and uses bootstrapped statistics to estimate the variance. 
Part (iv) validates the bootstrap method for CQIE, which directly follows from the bootstrap consistency for CQTE in Theorem \ref{thm:bootstrap_consistency_QTE}.

\section{Proofs of the theorems}
For any two positive sequences $\{a_n\}_n$ and $\{b_n\}_n$, the notation $a_n\lesssim b_n$ implies that there exists some constant $C>0$ such that $a_n\le Cb_n$ for any $n$. We use $C$ to denote some generic constant whose value is allowed to change from place to place. 
\begin{proof}[\bf Proof of Proposition \ref{prop0}]
Under the assumption of monotonicity, it is immediately evident that $\sum_{t=1}^m \phi_{t} (\mathcal{E}_t, \bar{a}_t, \tau )$ is a strictly increasing function of $\tau\in (0,1)$ when $\bar{a}_t=\bm{1}_t$ or $\bm{0}_t$. The independence between $U$ and $\mathcal{E}_t$ leads to the following: 
\begin{eqnarray*}
\mathbb{P}(   Y_t^*(\bar{a}_t) \leq   \phi_{t} ( \mathcal{E}_t, \bar{a}_t, \tau  )|\mathcal{E}_t )
=
\mathbb{P}(  \phi_{t} ( \mathcal{E}_t, \bar{a}_t, U  ) \leq   \phi_{t} ( \mathcal{E}_t, \bar{a}_t, \tau  )|\mathcal{E}_t ) = \mathbb{P}(U \leq \tau ) = \tau,
\end{eqnarray*}
when $\bar{a}_t=\bm{1}_t$ or $\bm{0}_t$. 
By definition, we know that $\mathcal{E}_1\subseteq \mathcal{E}_2\subseteq \cdots \subseteq \mathcal{E}_m$, which yields: 
\begin{eqnarray*}
\mathbb{P}( \sum_{t=1}^m  Y_t^*(\bar{a}_t) \leq  \sum_{t=1}^m  \phi_{t} ( \mathcal{E}_t, \bar{a}_t, \tau  )|\mathcal{E}_m)
=
\mathbb{P}( \sum_{t=1}^m  \phi_{t} ( \mathcal{E}_t, \bar{a}_t, U  ) \leq  \sum_{t=1}^m  \phi_{t} ( \mathcal{E}_t, \bar{a}_t, \tau  )|\mathcal{E}_m )
=
P( U \leq \tau ) = \tau.
\end{eqnarray*}
This leads to the conclusion that: 
$$
\textrm{CQTE}_\tau=\textrm{SCQTE}_\tau=  \sum_{t=1}^m  \phi_{t} ( \mathcal{E}_t, \bm{1}_{t}, \tau  ) - \sum_{t=1}^m  \phi_{t} ( \mathcal{E}_t, \bm{0}_{t}, \tau  ).
$$
With these steps, the proof is completed. 
\end{proof}

\begin{proof}[\bf Proof of Proposition \ref{prop:QTE_form}]
Consider the following varying coefficient models  
\begin{eqnarray}
\label{model:TQVCM_outcome_DE}
Y^*_{t}(\bar{a}_t)&=& \beta_{0 }(t, U) + {S^*_{t}}(\bar{a}_t)^\top\beta(t, U_i) + a_t\gamma(t, U_i),	\\
\label{model:TQVCM_outcome_IE}
S^*_{t+1}(\bar{a}_t) &=&\phi_{0 }(t)+\Phi (t) S^*_{i,t}(\bar{a}_{t-1})+a_t \Gamma(t) + E_{i}(t+1), 
\end{eqnarray}
where both $U$ and $E(t)$ are independent of the treatment history.  
Through simple calculations, we obtain: 
\begin{eqnarray*}
\phi_{t} ( \mathcal{E}_{t}, \bar a_{t}, U  )
&=&
\beta_{0 }(t, U)  + a_{t} \gamma(t, U) + \beta(t, U) ^\top 
\Big\{ \sum_{k=1}^{t-1} \Big( \prod_{l=k+1}^{t-1} \Phi(l) ( \phi_{0}(k) + a_{k} \Gamma(k)  )  \Big)  \Big\}  \\
&&+
\prod_{l=1}^{t-1}  \Phi(l)  S_{1} + \sum_{k=2}^t ( \prod_{l=k}^{t-1} \Phi(l) E(k) ).
\end{eqnarray*}
Given these assumptions and Proposition \ref{prop0}, Proposition \ref{prop:QTE_form} is therefore valid. 
\end{proof}

\begin{proof}[\bf Proof of Proposition \ref{prop:QTE_form_spatial}]
The proof of Proposition \ref{prop:QTE_form_spatial} is similar to that of Proposition \ref{prop:QTE_form}, and we omit it to save space. In the following, we introduce the two varying coefficient models for the potential outcomes, given by
\begin{eqnarray}
\label{model:STQVCM_outcome_DE}
Y^*_{t,\iota}(\bar{a}_{t,[1:r]}) = \beta_{0} (t,\iota,U) + \beta^\top (t,\iota,U){S^*_{t,\iota}}(\bar{a}_{t-1,[1:r]})+ a_{t,\iota}\gamma_{1}(t,\iota,U)+ \bar{a}_{i,t,\mathcal{N}_\iota}\gamma_{2}(t,\iota,U),\\ 
\label{model:STQVCM_outcome_IE}	
S^*_{t+1,\iota}(\bar{a}_{t,[1:r]}) = \phi_{0 }(t,\iota) + \Phi(t,\iota)S^*_{t,\iota}(\bar{a}_{t-1,[1:r]}) + a_{t,\iota}\Gamma_{1 }(t,\iota) + \bar{a}_{i,t,\mathcal{N}_\iota}\Gamma_{2 }(t,\iota) + E(t+1, \iota), 
\end{eqnarray}
where the rank variable $U$ and the errors $\{E(t,\iota)\}$ are independent of the treatment history. 
\end{proof}

\begin{proof}[\bf Proof of Theorem \ref{thm:bootstrap_consistency_QTE}]
First, we provide a sketch of the proof, which is divided into three parts.
In the initial step, we acquire a uniform Bahadur representation of the first-stage estimator, which is detailed in Lemma \ref{lem:uniform_Bahadur}.
In the second step, we dissect the difference between the distributions of the proposed statistic and the bootstrap statistic. These differences are subsequently bounded by employing the technique of comparison of distributions, as elaborated in \cite{chernozhukov_gaussian_2013}.
Finally, we investigate the power properties of the proposed test. 

We next detail each of the step.
Recall that $\tilde \theta_\tau (t) =  \sum_{\tilde t=1}^m \omega_{\tilde t,h}(t) \widehat\theta_\tau (\tilde t)$. According to the uniform Bahadur representation in  Lemma \ref{lem:uniform_Bahadur}, 
we have 
$$
\tilde \theta_\tau (t) = \theta_{ s, \tau}( t ) + {1 \over n} \sum_{i=1}^n ( \sum_{\tilde t=1}^m B_{i,\tilde t}(t) \psi_\tau( e_{i,\tau}( t) ) )
+ O_p( n^{-3/4} \log^3(nm) + n^{-1} \log^4(nm) ),
$$
where
$ 
B_{i,\tilde t}(t) = \omega_{\tilde t,h}(t) D_{\tilde t}^{-1} Z_{i, \tilde t},$ $\theta_{s, \tau}(t) = \sum_{\tilde t=1}^m \omega_{\tilde t,h}(t) \theta_\tau (\tilde t),$ 
$D_t =
n^{-1} \sum_{i=1}^n f_{e_i}(0, t; \tau) Z_{i,t}  Z_{i,t}^\top,
$ 
and the big-$O_p$ term is uniform in $\tau$ and $t$. 

Let $e_i^\theta (t; \tau)=\sum_{k=1}^mB_{i,k}(t)  \psi_\tau( E_{i}( k ) )=\{ e_{i}^{\beta_0} (t; \tau),(e_{t}^\beta (t; \tau) )^\top,e_{i}^\gamma(t; \tau) \}^\top$ and $e^\theta (t; \tau ) =n^{-1/2}\sum_{i=1}^n E_{i}^\Theta (t)$. 
Similarly, we can represent $\widetilde{\Theta}(t)$ as 
\begin{equation*}
\widetilde{\Theta}(t)=\Theta_{s}(t)+\frac{1}{n}\sum_{i=1}^n \left(\sum_{k=1}^{m-1}B_{i,k}(t)E_{i}(k) \right) + O_p( n^{-3/4} \log^3(nm) + n^{-1} \log^4(nm) ),
\end{equation*}
where $\Theta_{s}(t )=\sum_k \omega_{k,h}(t) \Theta(k)$. 

Let $E_{i}^\Theta(t )=\sum_{k=1}^mB_{i,k}(t)E_{i}(k)=\{E_{i}^{\phi_0}(t),(E_{i}^\Phi (t) )^\top,( E_{i}^\Gamma (t))^\top \}$ 
and $E^\Theta (t)=n^{-1/2}\sum_{i=1}^n E_{i}^\Theta (t) $. It follows that
\begin{eqnarray*}
\widetilde\beta_\tau(t)=\beta_{s, \tau}(t )+\frac{1}{\sqrt{n}}e^\beta(t; \tau),\
\widetilde\Gamma(t)=\gamma_{s, \tau}(t )+\frac{1}{\sqrt{n}}e^\gamma(t; \tau), \\
\widetilde\Phi (t)=\Phi_{s}( t )+\frac{1}{\sqrt{n}}E^\Phi (t),\
\widetilde\Gamma (t)=\Gamma_{s}(t)+\frac{1}{\sqrt{n}}E^\Gamma (t).
\end{eqnarray*}
For simplicity, let $\vc(\cdot)$ be the operator that reshapes a matrix into a vector by stacking its columns on top of one another. Define 
\begin{align}\label{eq:GARx}
&x_{i} (t; \tau ) = \Big[ (e_{i}^\gamma (t; \tau ) )^\top,
(e_{i}^\beta (t; \tau ) )^\top,
\{\vc(E_{i}^\Phi (t ) )\}^\top,
(E_{i}^\Gamma (t ) )^\top
\Big]^\top\in\R^{d(d+3)},\nonumber\\
&x_i(\tau)=\big(x_{i}(2; \tau )^\top,x_{i}(3; \tau )^\top,\ldots,x_{i}(m; \tau )^\top\big)^\top\in\R^{p_x},\quad p_x=(m-1)dp,\ d=p-3.
\end{align}
For any $\tau$, let $\left\{E_{i}^b (j),E_{i}^b(j) \right\}$ be randomly selected samples (with replacement) from $\{ \widehat E_{i}(j), \widehat E_{i}(j) \}$. 
We next define
\begin{align}\label{eq:Empirical_bootstrap}
&w_{i}(t; \tau) = \Big[(e_{i}^{\gamma,b} (t; \tau) )^\top,
(e_{i}^{\beta,b} (t; \tau) )^\top,
\{\vc(E_{i}^{\Phi,b} (t) )\}^\top,
(E_{i}^{\Gamma,b} (t) )^\top
\Big]^\top\in\R^{d(d+3)},\nonumber\\
&w_i(\tau)=\big(w_{i}(2; \tau )^\top, w_{i}(3; \tau )^\top,\ldots,w_{i}(m; \tau )^\top\in\R^{p_x}.
\end{align}
We also define $\left\{ \tilde e_{i,\tau} (j), \tilde E_{i} (j) \right\}$ be the empirical Gaussian analogs of $\left\{E_{i}(j ),E_{i}(j) \right\}$
such that
$
\tilde E_{i} (j )=\widehat E_{i}(j)\xi_i,\tilde  E_{i}(j )=\widehat E_{i}(j )\xi_i
$
where $\xi_1, \dots, \xi_n$ are i.i.d standard normal random variables,
and
\begin{align}\label{eq:GARw}
&\tilde w_{i}(t; \tau) = \Big[ (\tilde e_{i}^{\gamma} (t; \tau) )^\top,
(\tilde e_{i}^{\beta} (t; \tau) )^\top,
\{\vc( \tilde E_{i}^{\Phi} (t) )\}^\top,
(\tilde E_{i}^{\Gamma} (t) )^\top
\Big]^\top\in\R^{d(d+3)},\nonumber\\
&\tilde w_i(\tau)=\big( \tilde w_{i}(2; \tau )^\top, \tilde w_{i}(3; \tau )^\top,\ldots, \tilde w_{i}(m; \tau )^\top\in\R^{p_x}.
\end{align}

Let
\begin{align*}
&X(\tau)   = (X_2^\top(\tau),X_3^\top(\tau),\ldots,X_m^\top(\tau) )=\frac{1}{\sqrt{n}}\sum_{i=1}^n x_i(\tau),\\
&W( \tau ) = (W_2^\top(\tau),W_3^\top(\tau),\ldots,W_m^\top(\tau) )=\frac{1}{\sqrt{n}}\sum_{i=1}^n w_i(\tau),\\
&\tilde W( \tau ) = (\tilde W_2^\top(\tau), \tilde W_3^\top(\tau),\ldots, \tilde W_m^\top(\tau) )=\frac{1}{\sqrt{n}}\sum_{i=1}^n \tilde w_i(\tau).
\end{align*}
We remark that the term $X(\tau)$ represents the difference between the estimators and the true smoothed coefficients
in terms of the error processes,
while the term $W( \tau )$ represents the difference between the bootstrap estimators and the 
obtained estimators,
and $\tilde W( \tau ) $ is the Gaussian analogy of $X(\tau)$. 

Define the following functions
{\small\begin{eqnarray*}
\textrm{CQDE}_{\tau}(X;\theta,\Theta)&\equiv& 
\sum_{l=1}^m \left(\Gamma(l)+\frac{e^\gamma(l;\tau)}{\sqrt{n}}\right),\\
\textrm{CQIE}_{\tau}(X;\theta,\Theta)&\equiv& \sum_{l=2}^m
\left[\left(\beta_\tau(l)+\frac{e^\beta(l;\tau)}{\sqrt{n}}\right)^\top
\sum_{j=1}^{l-1}\left\{\prod_{k=j+1}^{l-1}\left(\Phi(k)+\frac{E^\Phi(k)}{\sqrt{n}}\right)\left(\Gamma(j)
+\frac{E^\Gamma(j )}{\sqrt{n}}\right)\right\}\right], \\
\textrm{CQTE}_{\tau}(X;\theta,\Theta) &\equiv& 
\textrm{CQDE}_{\tau}(X;\theta,\Theta) +	\textrm{CQIE}_{\tau}(X;\theta,\Theta).
\end{eqnarray*} }

To verify the proposed bootstrap procedure, we aim to provide an upper error bound for 
\begin{eqnarray*}
\rho^*(z;\tau)=\Big|\prob \left(  \kappa_{\tau} \leq z \right)
-\prob\left( \kappa_{\tau}^b \leq z \Big|\textrm{Data}\right) \Big|,
\end{eqnarray*}
where $\kappa_{\tau}=m^{-1}(T_{\tau}-\textrm{CQTE}_{\tau})$ and $\kappa_{\tau}^b=m^{-1} (T_{\tau}^b - T_{\tau})$. 
Notice that $\kappa_{\tau}$ and $\kappa_{\tau}^b$ can be represented by
\begin{eqnarray*}
\kappa_{\tau} &=& m^{-1}\QTE_{\tau}(X;\theta_{s, \tau},\Theta_{s})-m^{-1}\QTE_{\tau}(0;\theta_\tau,\Theta),\\
\kappa_{\tau}^b
&=&m^{-1}\QTE_{\tau}(W;\widetilde\theta_\tau,\widetilde\Theta)-m^{-1}\QTE_{\tau}(0;\widetilde\theta_\tau,\widetilde\Theta),
\end{eqnarray*}
respectively. 

Define $\widetilde \kappa_{\tau}^b = m^{-1}\QTE_{\tau}( \widetilde W;\widetilde\theta_\tau,\widetilde\Theta)-m^{-1} \QTE_{\tau}(0;\widetilde\theta_\tau,\widetilde\Theta)$, 
$\kappa_{01, \tau }^*=m^{-1}\QTE_{\tau}(X;\theta_{s, \tau},\Theta_{s})-m^{-1}\QTE_{\tau}(0;\theta_{s, \tau},\Theta_{s})$, 
$\kappa_{02, \tau }^*=m^{-1}\QTE_{\tau}(0;\theta_{s, \tau},\Theta_{s})-m^{-1}\QTE_{\tau}(0;\theta_\tau,\Theta)$. It follows that $\kappa_{\tau}=\kappa_{01, \tau}^*+\kappa_{02, \tau }^*$.
Therefore,
\begin{align}
\label{eq:rho*}
\rho^*(z;\tau)&=\Big|P\left( \kappa_{\tau}\le z\right)-P\left( \kappa_{\tau}^b\le z \Big|\textrm{Data} \right)\Big|  \nonumber \\
& \leq 
\Big| P\left( \kappa_{\tau}\le z\right)-P \left(  \kappa_{01, \tau}^*\le z\right) \Big| 
+ \Big|P\left( \kappa_{01, \tau}^*\le z\right)
-P\left( \widetilde \kappa_{\tau}^b\le z \Big|\textrm{Data} \right)\Big|   \nonumber \\
& + \Big|P\left( \widetilde \kappa_{\tau}^b\le z \Big|\textrm{Data}\right)-P\left( \kappa_{\tau}^b\le z  \Big|\textrm{Data} \right)\Big|   \nonumber \\
& =
I_1 + I_2 + I_3,
\end{align}
where $I_1$, $I_2$ and $I_3$ denote the above three components, respectively.


Define the event 
$$
\mathcal{E} = \{ Z_{i, t}: \max_i \max_t Z_{i, t} \leq C \log(nm) \}.
$$
Under the sub-Gaussian condition, it follows from Bonferroni's inequality that $\prob( \mathcal{E} ) \geq  1 - \exp( - C\log^2(nm)/2c + \log(nm) ) $ for some constant $c>0$. This probability approaches 1 as $nm$ diverges to infinity. We can derive the desired result by
showing that on the event $\mathcal{E}$, 
$\sup_\tau \sup_z \rho^*(z;\tau) \leq C n^{-1/8} $.

On the event $\mathcal{E}$, 
similar to the proofs of Lemma 3 and Theorem 2 in \citet[Section D]{Luo2021policy}, we can prove that there exits some constant $C > 0$ such that
\begin{align}
\label{eq:I_1 order}
\sup_z I_1 &\leq C ( n^{1/2} h^2 + n^{1/2} m ^{-1} + n^{-1/4} \log^3(nm) + n^{-1/2} \log^4(nm) ), \\
\label{eq:I_2 order}
\sup_z I_2 &\leq C n^{-1/8}.
\end{align}
This yields the upper error bounds for $I_1$ and $I_2$. 

To bound $I_3$, we introduce the following notations.
Let $\left\{ \check e_{i,\tau} (j), \tilde E_{i} (j) \right\}$ be the empirical Gaussian analogs of $\left\{e^b_{i, \tau}(j ),E^b_{i, \tau}(j) \right\}$
such that
$
\check e_{i,\tau} (j)= e^b_{i}(j; \tau)\xi_i,\check E_{i}(j )=E^b_{i}(j )\xi_i,
$
where $\xi$ is the same for $\tilde w_i(\tau$) in \eqref{eq:GARw}, and
\begin{align*}
&\check w_{i}(t; \tau) = \Big[(\check e_{i}^{\gamma} (t; \tau) )^\top,
(\check e_{i}^{\beta} (t; \tau) )^\top,
\{\vc( \check E_{i}^{\Phi} (t) )\}^\top,
(\check E_{i}^{\Gamma} (t) )^\top
\Big]^\top\in\R^{d(d+3)},\nonumber\\
&\check w_i(\tau)=\big( \check w_{i}(2; \tau )^\top, \check w_{i}(3; \tau )^\top,\ldots, \check w_{i}(m; \tau )^\top\in\R^{p_x},   \\
& \check W( \tau ) = (\check W_2^\top(\tau), \check W_3^\top(\tau),\ldots,\check W_m^\top(\tau) )=\frac{1}{\sqrt{n}}\sum_{i=1}^n \check w_i(\tau), \\
& \check \kappa_{\tau}^b \equiv m^{-1}\QTE_{\tau}( \check W;\widetilde\theta_\tau,\widetilde\Theta)-m^{-1}\QTE_{\tau}(0;\widetilde\theta_\tau,\widetilde\Theta) .
\end{align*}
It follows from triangle inequality that
\begin{align}
\label{eq:I_3 decomposition}
I_3 &= \Big|P\left\{ \widetilde \kappa_{\tau}^b\le z\right\}-P\left\{ \kappa_{\tau}^b\le z\right\}\Big| \nonumber \\
& \leq 
\Big|P\left\{ \widetilde \kappa_{\tau}^b\le z\right\}-P\left\{ \check \kappa_{\tau}^b\le z\right\}\Big|
+ \Big|P\left\{ \check \kappa_{\tau}^b\le z\right\}-P\left\{ \kappa_{\tau}^b\le z\right\}\Big|   \nonumber \\
&   \equiv I_{31} + I_{32},  
\end{align}
Similar to the derivation of Lemma 3 in \cite{Luo2021policy},
one can deduce that
\begin{align}
\label{eq:I_32 order}
\sup_z I_{32} \leq C n^{-1/8}.
\end{align}

As for $I_{31}$, notice that
\begin{eqnarray*}
&&\widetilde \kappa_{\tau}^b -   \check \kappa_{\tau}^b
=
\frac{1}{m} \QTE_{\tau}( \check W;\widetilde\theta_\tau,\widetilde\Theta ) -  \frac{1}{m} \QTE_{\tau}(  \widetilde W;\widetilde\theta_\tau,\widetilde\Theta ) \\
&=&
\frac{1}{m}\sum_{l=1}^m \left(\Gamma(l)+\frac{ \check e^{\gamma} (l;\tau)}{\sqrt{n}} \right) 
- 
\frac{1}{m}\sum_{l=1}^m \left(\Gamma(l)+\frac{ \tilde e^{\gamma} (l;\tau)}{\sqrt{n}} \right)  \\
&&+ \frac{1}{m}\sum_{l=2}^m
\left[\left(\beta_\tau(l)+\frac{ \check e^{\beta} (l;\tau)}{\sqrt{n}}\right)^\top
\sum_{j=1}^{l-1}\left\{\prod_{k=j+1}^{l-1}\left(\Phi(k)+\frac{ \check E^{\Phi} (k)}{\sqrt{n}}\right)\left(\Gamma(j)
+\frac{\check E^{\Gamma}(j )}{\sqrt{n}}\right)\right\}\right] \\
&&-
\frac{1}{m}\sum_{l=2}^m
\left[\left(\beta_\tau(l)+\frac{\widetilde e^\beta(l;\tau)}{\sqrt{n}}\right)^\top
\sum_{j=1}^{l-1}\left\{\prod_{k=j+1}^{l-1}\left(\Phi(k)+\frac{\widetilde E^\Phi(k)}{\sqrt{n}}\right)\left(\Gamma(j)
+\frac{\widetilde E^\Gamma(j )}{\sqrt{n}}\right)\right\}\right]  
\end{eqnarray*}
Denote 
\begin{eqnarray*}
\delta_\beta (l, \tau) 
&=&  \frac{1}{\sqrt{n}}( \check e^\beta(l;\tau) -\widetilde e^\beta(l;\tau) )   \\
&=&
{1 \over n m} \sum_{i=1}^n \xi_i ( \sum_{\tilde t=1}^m B_{i,\tilde t}(l) \psi_\tau( e^b_i( l; \tau ) ) ) - 
{1 \over n m} \sum_{i=1}^n \xi_i ( \sum_{\tilde t=1}^m B_{i,\tilde t}(l) \psi_\tau( \widehat e_i( l; \tau ) ) ) \\
&=&
{1 \over n m} \sum_{i=1}^n \xi_i  ( \sum_{\tilde t=1}^m B_{i,\tilde t}(l) \{ \psi_\tau( e^b_i( l; \tau ) ) ) - \psi_\tau( \widehat e_i( l; \tau ) ) )  \}
\end{eqnarray*}
Notice that $\sup_{l, \tau} | \{ \psi_\tau( e^b_i(l; \tau ) - \psi_\tau( \widehat e_i( l; \tau )   \} | \leq 1$. On the event $\mathcal{E}$, we have
$$
\max_{l, \tau} \delta_\beta (l, \tau)  \leq \max_l |  \sum_{\tilde t=1}^m B_{i,\tilde t}(l) | | {1 \over n m} \sum_{i=1}^n \xi_i |= O( n^{-1/2}  \log (nm)  )
$$
due to the sub-Gaussian property of $Z_{i,t}$ and that $| n^{-1/2} \sum_{i=1}^n \xi_i | = O_p(1)$.
Similarly, for
$
\delta_\gamma (l, \tau) 
=  n^{-1/2}{\sqrt{n}}( \check e^\gamma(l;\tau) -\widetilde e^\gamma(l;\tau) ),  
\delta_\Phi (l) = n^{-1/2}( \check E^\Phi(l) -\widetilde E^{\Phi, b} (l) )  ,
\delta_\Gamma (l) = n^{-1/2}( \check E^\Gamma(l) -\widetilde E^{\Gamma, b} (l) )  ,
$
we have that 
$$
\max_{l, \tau} \| \delta_\gamma (l; \tau) \| = \max_{l} \| \delta_\Phi (l) \| = \max_{l} \| \delta_\Gamma (l) \| = O_p( n^{-1/2}  \log (nm)  ).
$$

We can deduce that $\widetilde \kappa_{\tau}^b -   \check \kappa_{\tau}^b = J_1 +J_2 - J_3$,
where
\begin{eqnarray*}
J_1&=&
\frac{1}{m} \sum_{l=1}^m \delta_\gamma (l, \tau),  \\
J_2&=&  
\frac{1}{m}\sum_{l=2}^m
\left[\left(\beta_\tau(l)+\frac{\widetilde e^\beta(l;\tau)}{\sqrt{n}} +  \delta_\beta (l, \tau)  \right)^\top
\sum_{j=1}^{l-1}\left\{\prod_{k=j+1}^{l-1}\left(\Phi(k)+\frac{\widetilde E^\Phi(k)}{\sqrt{n}} + \delta_\Phi (l)    \right) \right.\right. \\
&& \left. \left. \left(\Gamma(j)
+\frac{\widetilde E^\Gamma(j )}{\sqrt{n}} + \delta_\Gamma (l) \right)\right\}\right],  \\
J_3 &=&
\frac{1}{m}\sum_{l=2}^m
\left[\left(\beta_\tau(l)+\frac{\widetilde e^\beta(l;\tau)}{\sqrt{n}}\right)^\top
\sum_{j=1}^{l-1}\left\{\prod_{k=j+1}^{l-1}\left(\Phi(k)+\frac{\widetilde E^\Phi(k)}{\sqrt{n}}\right)\left(\Gamma(j)
+\frac{\widetilde E^\Gamma(j )}{\sqrt{n}}\right)\right\}\right].  
\end{eqnarray*}
Similar to the steps of deriving the orders of $\kappa_{01}^*$ in \eqref{eq:kappa_01_*},
we can obtain that
$$
| \widetilde \kappa_{\tau}^b -   \check \kappa_{\tau}^b | = O_p( n^{-1/2}  \log (nm)  ),
Var( \kappa_{\tau}^b  ) = Var( \check \kappa_{\tau}^b ) = O_p( n^{-1}  \log^2 (nm) ).
$$
On the event $\mathcal{E}$, conditional on the data,
$\widetilde \kappa_{\tau}^b$ and $\check \kappa_{\tau}^b$ are Gaussian random variables,
according to Lemma 3.1 in \cite{chernozhukov_gaussian_2013},
$$
\label{eq:I_31 order}
\sup_z I_{31}  \leq C | Var\{ (\widetilde \kappa_{\tau}^b) \} - Var\{ (\check \kappa_{\tau}^b) \} |^{1/3} \leq C n^{-1/3}  \log^{2/3} (nm).
$$
Plugging \eqref{eq:I_32 order} and \eqref{eq:I_31 order} into \eqref{eq:I_3 decomposition},  we obtain that
$\sup_z I_{3} \leq C ( n^{-1/3}  \log^{2/3} (nm) + n^{-1/8})$.
Together with equations \eqref{eq:rho*}, \eqref{eq:I_1 order} and \eqref{eq:I_2 order}, we come to the desired assertion.

Finally, we show that the proposed test has good power properties. Under the alternative that $ \QTE_\tau /m  \gg n^{-1/2}  \log(n m)$ and according to 
\eqref{eq:I_1 order}, we have
\begin{eqnarray*}
\label{eq:power_decomposition}
P(T_\tau / m \geq z_\alpha ) = P(T_\tau  /m -\QTE_\tau  /m \geq  z_\alpha  -\QTE_\tau/m  )  
=
P(  \kappa^*_{01,\tau}    \geq  z_\alpha  -\QTE_\tau /m  ) +  o_p(1) \\
= 1-P(  -\kappa^*_{01,\tau}    >  \QTE_\tau /m -z_\alpha )+o_p(1)
\geq 1- E(  \kappa^*_{01,\tau}  )^2 / ( z_\alpha  - \QTE_\tau /m  )^2 + o_p(1),  \nonumber
\end{eqnarray*}
where $z_\alpha$ is the critical value. If we can show that 
$E(  \kappa^*_{01,\tau}  )^2 = O(n^{-1}  \log^2 (n m))$, then it follows that $P(T_\tau / m \geq z_\alpha)$ approaches one.

Direct calculations lead to 
$\kappa^*_{01,\tau} = G_1 + G_2 + G_3$ where
\begin{eqnarray*}
G_1 &=& {1 \over m} \sum_{l=1}^m \left( \frac{e^\gamma(l;\tau)}{\sqrt{n}}\right), \\
G_2 &=& {1 \over m}
\sum_{l=2}^m \left( \frac{e^\beta(l;\tau)}{\sqrt{n}}\right)^\top
\sum_{j=1}^{l-1}\left\{\prod_{k=j+1}^{l-1}\left(\Phi(k)+\frac{E^\Phi(k)}{\sqrt{n}}\right)\left(\Gamma(j)
+\frac{E^\Gamma(j )}{\sqrt{n}}\right)\right\} , \\
G_3 & =& {1 \over m}
\sum_{l=2}^m \beta_\tau(l)^\top
\sum_{j=1}^{l-1}\left\{ \prod_{k=j+1}^{l-1}\left(\Phi(k)+\frac{E^\Phi(k)}{\sqrt{n}}\right)\left(\Gamma(j)
+\frac{E^\Gamma(j )}{\sqrt{n}}\right) -  \prod_{k=j+1}^{l-1} \Phi(k) \Gamma(j)
\right\}.
\end{eqnarray*}
Recall that $e^\gamma(l;\tau), e^\beta(l;\tau), E^\Gamma(j ), E^\Phi(k)$ are Sub-Gassuain random vectors. It is straightforward to show that
$E(G_1^2) = O( \log^2(nm) /n )$.
Define the event 
$$
\mathcal{G} = 1 \left\{ 
\max_j |u_j / \sqrt{n} |  < (1- q)/2,
\right\}
$$
where 
$$u = (   (e^\beta (2;\tau))^\top,
( E^\Phi(2 ; \tau) )^\top, E^\Gamma(2 ; \tau), \dots, (e^\gamma(m;\tau))^\top,  (e^\beta (m;\tau))^\top,
( E^\Phi(m ) )^\top, E^\Gamma(m)).$$
On the event $\mathcal{G}$, define $\bar q = (1+q)/2, $we can obtain that
\begin{eqnarray*}
|G_2|
\leq 
\sum_{l=2}^m \left( \frac{e^\beta(l;\tau)}{\sqrt{n}}\right)
\sum_{j=1}^{l-1} \bar q^{t-j+1}( M_\Gamma + (1-q)/2) .
\end{eqnarray*}
Since the sequence $\bar q^{t-j+1}( M_\Gamma + (1-q)/2) $ is convergent, we have
$E(G_2^2) =O(  \log^2(nm) /n )$.
Meanwhile, denote $\delta = \max \{  \max_k \| E^\Phi(k) /\sqrt{n} \|_\infty,
\max_k   \| E^\Gamma(k; \tau) /\sqrt{n} \|_\infty  \}$, then
\begin{eqnarray*}
&& \left| \sum_{j=1}^{l-1}\left\{ \prod_{k=j+1}^{l-1}\left(\Phi(k)+\frac{E^\Phi(k)}{\sqrt{n}}\right)\left(\Gamma(j)
+\frac{E^\Gamma(j )}{\sqrt{n}}\right) -  \prod_{k=j+1}^{l-1} \Phi(k) \Gamma(j)
\right\} \right| \\
&\leq&
M_\Gamma \sum_{j=1}^{l-1} \left| \prod_{k=j+1}^{l-1}(\Phi(k) + \delta ) - \prod_{k=j+1}^{l-1} \Phi(k)  \right| 
+ 
\sum_{j=1}^{l-1} \prod_{k=j+1}^{l-1}|(\Phi(k) + (1-q)/2 ) \delta  | \\
& \leq &
M_\Gamma \sum_{j=1}^{l-1} \left|   \prod_{k=1}^{l-1-j} \delta^k 
\left( \begin{array}{c}
l-1-j  \\
k
\end{array} 
\right) q^{l-1-j-k} 
\right| + \delta  \sum_{j=1}^{l-1} \prod_{k=j+1}^{l-1} \bar q \\
& \lesssim &
\sum_{j=1}^{l-1} | (\delta+q)^{l-1-j} -q^{l-1-j} | + \delta  \lesssim  \delta.
\end{eqnarray*}
By the Sub-Gaussian property of $E^\Phi(k)$ and $E^\Gamma(k; \tau)$, one can obtain that
$E(G_3^2) = O( \log^2(nm) /n ) $.

Hence,  we have 
\begin{eqnarray}
\label{eq:kappa_01_*}
E(  \kappa^*_{01,\tau}  )^2 
&\lesssim&
E( (T^*_{01,\tau})^2  \mathcal G ) + P(1 - \mathcal{G})
\lesssim
E(G_1^2) +E(G_2^2) + E(G_3^2) + P(1 - \mathcal{G})  \nonumber \\
&\lesssim &
O(  \log^2(nm) /n ) +  \exp(- n + 2 n^{1/2} \log m - (\log m)^2),
\end{eqnarray}
where the last inequality follows from the fact that
$ P(1 - \mathcal{G})  \lesssim  \exp(- n + 2 n^{1/2} \log m - (\log m)^2) $, 
according to the Borell TIS inequality and Lemma 2.2.10 in \cite{van1996weak}.
This completes the proof.
\end{proof}

\begin{proof}[\bf Proof of Theorem \ref{thm:bootstrap_consistency_QTE_st}]
The proof of Theorem \ref{thm:bootstrap_consistency_QTE_st} is very similar to that 
of Theorem \ref{thm:bootstrap_consistency_QTE}. We omit it here to save space.
\end{proof}

\begin{proof}[\bf Proof of Theorem \ref{thm:theta_t asymp}]
The proof involves four parts. In part (i), we present the proof of the asymptotic normality of 
the first-stage estimator $\widehat{\bm{\theta}}_\tau$.
In part (ii), we give the proof of the asymptotic normality of 
the second-stage estimator $\widetilde{\bm{\theta}}_\tau$.
In parts (iii) and (iv), we give the proofs of the bootstrap consistency 
of the testing procedures with respect to $\QDE_\tau$ and $\QIE_\tau$, respectively.

For part (i), 	
we  use the Bahadur representation for ordinary quantile regression to derive the asymptotic normality of $\widehat{\theta}_\tau(t) $.
Specifically, according to \cite{koenker1987estimation}, for a given quantile level $\tau$, we have
\begin{eqnarray}
\label{eq:Bahadur_t}
\sqrt{n} D_t (\widehat{\theta}_\tau(t) - \theta_\tau (t) ) 
=
{1 \over \sqrt{n}} \sum_{i=1}^n Z_{i,t} \psi_\tau( e_{i,\tau}( t) ) + o_p(1), ~~
\end{eqnarray}
where $\psi_\tau(e) = \tau - I(e < 0)$ and $D_t 
=n^{-1} \sum_{i=1}^n f_{e_i}(0, t; \tau) Z_{i,t}  Z_{i,t}^\top$.

We need to verify the covariance structure of $\sqrt{n} (\widehat{\theta}_\tau(t) - \theta_\tau (t) ) $ 
and $\sqrt{n} (\widehat{\theta}_\tau(s)- \theta_\tau (s) ) $ for any $s, t=1, \dots, m$.
According to the Bahadur representation in \eqref{eq:Bahadur_t}, the asymptotic covariance of $\sqrt{n} D_t (\widehat{\theta}_\tau(t) - \theta_\tau (t) )$ and $\sqrt{n} D_s (\widehat{\theta}_\tau(s)- \theta_\tau (s) )$ equals
\begin{eqnarray}\label{eqn:asympcov}
\Mean\Big[   Z_{i,t}  \psi_\tau( e_{i,\tau}( t) )  \psi_\tau( E_{i}( s) )   Z^\top_{i,s}   \Big].
\end{eqnarray}	
Specifically, we can show 
\begin{eqnarray*}
\Mean[ \psi_\tau( e_{i,\tau}( t) ) \psi_\tau( E_{i}( s) ) ] 	
=
\Mean[ (\tau - I(e_i(t; \tau)< 0) ) (\tau - I(e_i(s; \tau)< 0) ) 
=
F_e(0,0, t, s; \tau)- \tau^2, 
\end{eqnarray*}
so 
\eqref{eqn:asympcov} is equal to
\begin{eqnarray*}
f^{-1}_{e}(0, t,;\tau) E^{-1}(Z_{i,t} Z_{i,t}^\top) 
E(Z_{i,t} [ F_e(0,0,t, s;\tau) - \tau^2] Z_{i,s}^\top)  f^{-1}_{e}(0, s,;\tau) E^{-1}(Z_{i,s} Z_{i,s}^\top) + o_p(1).
\end{eqnarray*}	
Therefore, we can prove 
the asymptotic normality of $\widehat{\bm{\theta}}_\tau$ by using some standard arguments in the quantile regression literature \citep{2005quantile}.

For part (ii), we 
establish the limiting distribution of $\widetilde{\theta}_\tau(t)$ by using 
\begin{eqnarray*}
\widetilde{\theta }_\tau (t )&=& \sum_{k=1}^m \omega_{k,h}  \widehat{\theta }_\tau (k )
=
\sum_{k=1}^m \omega_{k,h}(t) [  \widehat{\theta }_\tau (k ) - {\theta }_\tau (k ) + {\theta }_\tau (k ) - {\theta }_\tau (t) + {\theta }_\tau (t) ] \\
&=&
{\theta }_\tau (t) + 	\sum_{k=1}^m \omega_{k,h}(t) [\widehat{\theta }_\tau (k ) - {\theta }_\tau (k ) ]
+\sum_{k=1}^m \omega_{k,h} (t)  [ {\theta }_\tau (k ) - {\theta }_\tau (t) ].
\end{eqnarray*}

We assume the existence of a function $\theta_{0 \tau}$ such that $ \theta_{\tau} (k) = \theta_{0 \tau} (k / m)$. This function $\theta_{0 \tau} (k / m)$ is used to derive the bias. The calculation can be as follows: 
\begin{eqnarray}
\label{eq:bias}
E[ \widetilde{\theta }_\tau (t) -  {\theta }_\tau (t) ] &=& E[ 	\sum_{k=1}^m \omega_{k,h} (t) [{\theta }_\tau (k ) - {\theta }_\tau (t) ] ] 
=
E[ 	\sum_{k=1}^m \omega_{k,h} (t) [{\theta }_{0 \tau } (k /m ) - {\theta }_{0 \tau} (t /m) ] ] 
\nonumber \\
&=&
E\Big[ \sum_{k=1}^m { K(( k  -t )/(m h)) \over \sum_{j=1}^m K((j -t)/(mh)) } \{ {\theta }_{0 \tau } (k /m ) - {\theta }_{0 \tau } (t /m )\} \Big] \nonumber \\
&=&
E\Big[ \sum_{k=1}^m { (mh)^{-1} K(( k  -t )/(m h)) \over \sum_{j=1}^m (mh)^{-1} K((j -t)/(m h)) } {\theta }_{0 \tau } (k /m ) \Big] - {\theta }_{0 \tau } (k /m ).
\end{eqnarray}


Since the time grids are prefixed and equally spaced, we have 
$$
\sum_{k=1}^m (mh)^{-1} K\{( k  -t )/(m h) \} [ {\theta }_{0 \tau } (k /m ) -  {\theta }_{0 \tau } (t /m )]
=
\sum_{k=1}^m (mh)^{-1 } K\{( k  -t )/h \} [  {\theta }_{0 \tau } (k /m ) -  {\theta }_{0 \tau } (k /m )] \int_{{k-1}}^{k} du.
$$
Noticing that for any $u \in [k-1, k]$, $k - u \leq k - (k-1) \leq 1$, direct calculations lead to 
\begin{eqnarray*}
&& \sum_{k=1}^m (mh)^{-1} K \Big( { k  -t  \over m h } \Big) [  {\theta }_{0 \tau } ({k \over m} ) -  {\theta }_{0 \tau } ({t \over m} )  ] -
\sum_{k=1}^m \int_{k-1}^{k}  (mh)^{-1 } K \Big( { k  -t  \over m h } \Big) [ {\theta }_{0 \tau } ({k \over m} ) -  {\theta }_{0 \tau } ({t \over m} ) ]   du \\
&=&
\sum_{k=1}^m \int_{{k-1}}^{k}  \Big\{
(mh)^{-1 } K \Big( { k  -t  \over m h } \Big) [  {\theta }_{0 \tau } ({k \over m} ) -  {\theta }_{0 \tau } ({t \over m} )] 
- 
(mh)^{-1 } K \Big( { u  -t  \over m h } \Big) [ {\theta }_{0 \tau } ({u \over m} ) -  {\theta }_{0 \tau } ({t \over m} ) ]
\Big\}  du  \\
&=&
\sum_{k=1}^m \int_{{k-1}}^{k}  \Big\{
(m h)^{-1} { \partial(   K \Big( { u  -t  \over m h } \Big) [  {\theta }_{0 \tau } ({u \over m} ) -  {\theta }_{0 \tau } ({t \over m} ) ] )  \over \partial u} ( k -u) (1+o(1))
\Big\}  du  \\
&\leq &
\sum_{k=1}^m \int_{ {k-1}}^{k}  \Big\{
(m h)^{-2} K^\prime \Big( { u  -t  \over m h } \Big) [ {\theta }_{0 \tau } ({u \over m} ) -  {\theta }_{0 \tau } ({t \over m} )]  
+ m^{-2}h^{-1} K \Big( { u  -t  \over m h } \Big)  {\theta }^\prime_{0 \tau } ({u \over m} ) 
\Big\} (1+o(1)) du  \\
&=&
\sum_{k=1}^m \int_{(k-1)/m}^{k/m}  \Big\{
(m h)^{-1} K^\prime (\tilde u) [ {\theta }_{0 \tau } ( \tilde u h +{ t \over m} ) -  {\theta }_{0 \tau } ({t \over m} )]  
+ m^{-1} K ( \tilde u )  {\theta }_{0 \tau }^\prime ( \tilde u h +{ t \over m} )
\Big\} (1+o(1))  d \tilde u  \\
&=&
O(m^{-1}) ,
\end{eqnarray*}
where the last equality follows from Taylor expansion and 
Assumptions \ref{assump:kernel} and   \ref{assump:coe}.

Furthermore,  by applying a change of variables and considering Assumptions \ref{assump:kernel} and \ref{assump:coe}, we obtain the following: 
\begin{eqnarray*}
\int  (m h)^{-1 } K \Big( { k  -t  \over m h } \Big) [ {\theta }_{0 \tau } ({k \over m} ) -  {\theta }_{0 \tau } ({t \over m} )]   du  
=
\int   K(\tilde u) [  {\theta }_{0 \tau } (\tilde u h +{t \over m} ) -  {\theta }_{0 \tau } ({t \over m} ) ]  du  
=
O( h^2).
\end{eqnarray*}
Hence, for the fixed time grids,  we have
\begin{eqnarray}
\label{eq:bias_fixed}
\sum_{k=1}^m (mh)^{-1} K \Big( { k  -t  \over m h } \Big) [  {\theta }_{0 \tau } ({k \over m} ) -  {\theta }_{0 \tau } ({t \over m} )  ]
= O( h^2 + m^{-1} ).
\end{eqnarray}

Similarly, one can show that
\begin{eqnarray}
\label{eq:fixed design approximation error}
&&\sum_{k=1}^m (mh)^{-1} K\{( k  -t )/(mh) \}   - \int (m h)^{-1} K\{ (u - t) / (mh) \}  du  \\
& =&
{1 \over m h } \sum_{k=1}^m\int_{(k-1)}^{ k  }  
[ K\{( k  -t )/(mh) \} - K\{ (u - t) / (mh) \} ] d u  \nonumber \\
&=&
{1 \over m h } \sum_{k=1}^m\int_{(k-1)}^{ k } 
(mh)^{-1} K^\prime \{( u  -t )/(mh) \} du (1 + o(1))
=
O( (m h)^{-1} ). \nonumber
\end{eqnarray}

It is easy to derive that $ \int (mh)^{-1} K((u - t) / mh)  du 
=\int  K(\tilde u - t/(mh))  du =  1 + O(h^2) $ by using standard kernel smoothing techniques \citep{tsybakov2008introduction}. It leads to 
\begin{eqnarray}
\label{eq:bias_fixed_kernel_approximation}
\sum_{k=1}^m (mh)^{-1} K(( t_k  -t )/h) = 1 + O(h^2 + (m h)^{-1} ).
\end{eqnarray}

Plugging \eqref{eq:bias_fixed} and \eqref{eq:bias_fixed_kernel_approximation} into \eqref{eq:bias}, we derive the bias
for the fixed design
\begin{eqnarray*}
E[ \widetilde{\theta }_\tau (t) -  {\theta }_\tau (t) ]  = O( h^2 + m^{-1} ).
\end{eqnarray*}

For the asymptotic co-variance of $\widetilde{\theta }_\tau (t )$ and $\widetilde{\theta }_\tau (s )$, direct calculations lead to
\begin{eqnarray*}
&&	E[ \sum_{k=1}^m  \sum_{l=1}^m \omega_{k,h}(t) \omega_{l,h} (s) \{ \widehat{\theta }_\tau (k ) - {\theta }_\tau (k ) \} 
\{ \widehat{\theta }_\tau (l) - {\theta }_\tau (l ) \}   ] \\
&=&
\sum_{k=1}^m  \sum_{l=1}^m \omega_{k,h}(t) \omega_{l,h} (s) \bm{V}_{\widehat{\theta}, \tau} (k, l) \\
&=&
\sum_{k=1}^m  \sum_{l=1}^m \omega_{k,h}(t) \omega_{l,h} (s) f^{-1}_{e}(0, k,;\tau) E^{-1}(Z_{i,k} Z_{i,k}^\top) F_e(0,0, k, l;\tau)
E(Z_{i,k} Z_{i,s}^\top)  f^{-1}_{e}(0, l,;\tau) E^{-1}(Z_{i,l} Z_{i,l}^\top) \\
&=&
\sum_{k=1}^m  \sum_{l=1}^m \omega_{k,h}(t) \omega_{l,h} (s) g(k,l),	
\end{eqnarray*}
where $ g( k, l) = f^{-1}_{e}(0, k,;\tau) E^{-1}(Z_{i,k} Z_{i,k}^\top) F_e(0,0, k, l;\tau)
E(Z_{i,k} Z_{i,l}^\top)  f^{-1}_{e}(0,  l,;\tau) E^{-1}(Z_{i,l} Z_{i,l}^\top)$.
There also exists some function $g_0$ such that $g_0( k/m, l/m) =  g( k, l) $.
Then we can have
\begin{eqnarray*}
\sum_{k=1}^m  \sum_{l=1}^m \omega_{k,h}(t) \omega_{l,h} (s) g(k,l)
&=&	
\sum_{k=1}^m  \omega_{k,h}(t) \omega_{k,h} (s) g_0(k/m,k/m) 
+
\sum_{k \neq l} \omega_{k,h}(t) \omega_{l,h} (s) g_0(k/m,l/m) \\
&=&
H_1 + H_2,
\end{eqnarray*}	
where $H_1$ and $H_2$ denote the above two terms, respectively.

Denote $\widehat f(t) =  {1 \over m h} \sum_{k=1}^m K( { k-t \over m h } )$ , 
one gets
\begin{eqnarray*}
H_1 &=& {1 \over m^2 h^2} \sum_{k=1}^m K \Big( { k-t \over  m h } \Big) K\Big( { k-s \over m h } \Big) g_0( k/m, k/m) /( \widehat f(t) \widehat f(s) )=\tilde H_1 /( \widehat f(t) \widehat f(s) ), \\
H_2 &=& {1 \over m^2 h^2} \sum_{k \neq l} \omega_{k,h}(t) \omega_{l,h} (s)  K\Big( { k-t \over m h } \Big) K\Big( { l-s \over m h } \Big) g_0(k/m,l/m) /( \widehat f(t) \widehat f(s) )
=\tilde H_2 /( \widehat f(t) \widehat f(s) ).
\end{eqnarray*}	
If $g_0(t, s)$ has bounded and continuous second derivatives, 
following similar arguments of the proofs for Lemma 4 in \cite{zhu2012multivariate},
it is easy to derive that
\begin{eqnarray*}
\tilde H_1 &=& 	{1 \over m h^2}  \int K\Big( { u-t \over m h } \Big) K\Big( { u-s \over m h } \Big) g_0(u/m, u/m)  du ( 1 + O_p( (mh)^{-1/2} )  ) \\
&&
\left\{
\begin{array}{l}
= m^{-1} h^{-1} \int K(\nu )^2 [ g_0(t/m, t/m) + g_0^\prime (t/m) h \nu(1+o_p(1)) ]  d \nu (1 + O_p( (mh)^{-1/2}  ) ) \\
=
m^{-1} h^{-1} \int K(\nu )^2 d \nu  g_0(t/m,t/m) (1 + O_p( h+ (mh)^{-1/2}  ) ),
~~ \text{if } t=s, \\
\leq 
m^{-1} h^{-1} \int K(u)^2 du\sqrt{ g_0(t/m,t/m)  g_0(s/m,s/m) } (1 + O_p(h + (mh)^{-1/2}  ) ), ~~ \text{if } t \neq s. 
\end{array}
\right.
\end{eqnarray*}	
Meanwhile, we can obtain 
\begin{eqnarray*}
\tilde H_2 &=&
{m-1 \over m h^2} \int \int K( { u_1-t \over m h } ) K( { u_2-s \over m h } ) g_0(u_1/m,u_2/m)  du_1 du_2 ( 1 + O_p( (mh)^{-1/2} )  ) \\
&=&
{m-1 \over m } \int \int K( \nu_1 ) K( \nu_2 ) g_0(\nu_1 h +t/m, \nu_2 h + s/m)  d\nu_1 d\nu_2 ( 1 + O_p( (mh)^{-1/2} )  ) \\
&=&
{m-1 \over m } g_0(t/m, s/m)  (1 + o_p(1))= {m-1 \over m } g(t, s)  (1 + o_p(1)).
\end{eqnarray*}	
It is easy to see that $\tilde H_2$ is the leading term of $H_1$ by noticing that $\tilde H_1 $ is of lower order of $\tilde H_2$
according to $m h \rightarrow \infty$.
Because $\widehat f(t) = 1 + O(h^2 + (m h)^{-1})$, one can obtain
$$
E[ \sum_{k=1}^m  \sum_{l=1}^m \omega_{k,h}(t) \omega_{l,h} (s) \{ \widehat{\theta }_\tau (k ) - {\theta }_\tau (k ) \} 
\{ \widehat{\theta }_\tau (l) - {\theta }_\tau (l ) \}   ] 
= {m-1 \over m } 
\Cov[ \widehat{\theta }_\tau (t ),  \widehat{\theta }(s, \tau ) ] ( 1+o_p(1)).
$$
The smoothing step reduces $1/m$ of the variance.


With the preparations above, it is easy to calculate that for any vector $a_{n, 2}$ with unit norm, 
$$
E[ \sqrt{n} {a}_{n,2}^\top (\widetilde{ {\bm{\theta}}}_\tau  - \bm{\theta}_\tau   ) ] = O( \sqrt{n} h^2 + \sqrt{n} m^{-1}), ~~
Var [ \sqrt{n} {a}_{n,2}^\top (\widetilde{ {\bm{\theta}}}_\tau  - {\bm{\theta}}_\tau   ) ] = 
{a}_{n,2}^\top\bm{V}_{\widetilde{\theta}, \tau} {a}_{n,2}.
$$
By checking the Lindeberg condition, we can establish the asymptotic normality of $ \sqrt{n} {a}_{n,2}^\top(\widetilde{ {\bm{\theta}}}_\tau  - {\bm{\theta}}_\tau)$.

It is not hard to find that $\widetilde{ {\theta}}_\tau$ is a linear transformation of $\widehat{ {\bm{\theta}}}_\tau$.
Recall the Bahadur representation of $\widehat{ {\theta}}_\tau$, one can deduce that 
$$
\sqrt{n} {a}_{n,2}^\top (\widetilde{ {\theta}}_\tau  - {\theta}_\tau   ) 
=
n^{-1/2} \sum_{i=1}^n \sum_{j=1}^m \sum_{k=1}^m \omega_{k, h} (t_j) a_{n,2, j}^\top D_k^{-1} Z_{i, k} \psi( E_{i} (t_k) ) + o_p(1),
$$
where $  {a}_{n,2} = ( a_{n,2, 1}, \dots, a_{n,2, m} )$ with each $ a_{n,2, j}$ corresponds to a vector with dimension $(d+2)$. 
Observing that 
$$
P\Big( \Big| 
\sum_{j=1}^m \sum_{k=1}^m \omega_{k, h} (t_j) a_{n,2, j}^\top D_k^{-1} Z_{i, k} \psi( E_{i} (t_k) )  
\Big| 
>
\epsilon \sqrt{  n {a}_{n,2}^\top\bm{V}_{\widetilde{\theta}, \tau} {a}_{n,2} }
\Big)
\leq  { 1 \over \epsilon^2 n} \rightarrow 0
$$
holds for any $\epsilon > 0$ according to the Chebyshev's inequality. Then the Lindeberg condition follows such that
\begin{eqnarray*}
( {a}_{n,2}^\top\bm{V}_{\widetilde{\theta}, \tau} {a}_{n,2} )^{-1} 
E | \sum_{j=1}^m \sum_{k=1}^m \omega_{k, h} (t_j) a_{n,2, j}^\top D_k^{-1} Z_{i, k} \psi( E_{i} (t_k) )  |^2 \\
\times 
1 (
|  \sum_{j=1}^m \sum_{k=1}^m \omega_{k, h} (t_j) a_{n,2, j}^\top D_k^{-1} Z_{i, k} \psi( E_{i} (t_k) )   | 
>
\epsilon \sqrt{  n {a}_{n,2}^\top\bm{V}_{\widetilde{\theta}, \tau} {a}_{n,2} }
)  \rightarrow 0
\end{eqnarray*}
for any $\epsilon > 0$ .


(iii)
As for the bootstrap consistency of the bootstrap procedures for testing QDE, it suffices to establish the Bahadur representation of the bootstrap estimators
$\widehat{\theta}_\tau^b$ parallel to \eqref{eq:Bahadur_t} such that 
\begin{eqnarray}
\label{eq:Bahadur_t_bootstrap}
\sqrt{n} D_t (\widehat{\theta}_\tau^b(t) - \widehat{\theta}_\tau (t) ) 
=
{1 \over \sqrt{n}} \sum_{i=1}^n Z_{i,t} \psi_\tau( \widehat{e}^b_{i, \tau}( t ) ) + o_{p^*}(1), ~~
D_t 
=
{1 \over n} \sum_{i=1}^n f_{e_i}(0, t; \tau) Z_{i,t}  Z_{i,t}^\top,
\end{eqnarray}
where $\widehat{e}^b_{i, \tau}( t )$ is the bootstrap error.  
We use $p^*$ to denote ``convergence in probability'' conditional on the data.
Note that $E^*\psi_\tau(   \widehat{e}^b_{i, \tau}( t )  ) =0 $.

Recall that $Y_{it}^b = Z_{i,t}^\top \widetilde{\theta} + \widehat{e}^b_{i, \tau}( t )$.
The bootstrap estimates $\widehat \theta^b $ is obtained by minimizing the 
following loss function with respect to the bootstrap sample 
$$
\ell^*_{nt}   ( \theta ) = \sum_{i=1}^n \rho_\tau ( Y_{it}^b - Z_{i,t}^\top \theta ) 
= \sum_{i=1}^n \rho_\tau ( Z_{i,t}^\top \widetilde{\theta}  - Z_{i,t}^\top \theta + \widehat{e}^b_{i, \tau}( t ) ).
$$
It is easy to see that  $ \sum_{i=1}^n  E^*[  \rho_\tau ( Z_{i,t}^\top (  \theta_0  -  \theta ) +  {e}^b_i( t; \tau ) )    - \rho_\tau (   {e}^b_i( t; \tau ) ) ]$ is globally minimized at 
$ \theta = \theta_0$. Similar to the steps in \cite{feng2011wild}, according to the Sub-Gaussian assumption of the covariates and the error process,
we have that  
\begin{eqnarray*}
|  \sum_{i=1}^n [ \rho_\tau ( Z_{i,t}^\top (  \widetilde \theta  -  \theta ) +  \widehat{e}^b_{i, \tau}( t )  )  - \rho_\tau ( Z_{i,t}^\top (   \widehat{e}^b_{i, \tau}( t )  ) 
- E^*[  \rho_\tau ( Z_{i,t}^\top (  \theta_0  -  \theta ) +  {e}^b_i( t; \tau ) )    - \rho_\tau (   {e}^b_i( t; \tau ) ) ] | = o_{p^*}(1).
\end{eqnarray*}
uniformly on any compact set of $\theta_0$.
Thus, $\widehat \theta^b   -  \widetilde \theta  \rightarrow 0$ in probability. 

Now consider 
$$
S_{it}   ( \theta )  = Z_{i,t} \psi_\tau(  Z_{i,t}^\top \widetilde{\theta}  - Z_{i,t}^\top \theta + \widehat{e}^b_{i, \tau}( t )  )
-
Z_{i,t}  \psi_\tau(   \widehat{e}^b_{i, \tau}( t )  ), \quad S_{nt}   ( \theta ) = n^{-1} \sum_{i=1}^n S_{it}   ( \theta ) . 
$$
Following equation (2) in the supplementary material of \cite{feng2011wild}, we can obtain that
$$
S_{nt}   ( \theta ) = E^* S_{it}   ( \theta ) + o_{p^*} (n^{-1/2}).
$$

Now we calculate the explicit form of $E^* S_{it}   ( \theta ) $. 
Notice that $ \widehat{e}^b_{i, \tau}( t )  = Y^b_{it} - Z_{i,t}^{b \top} \widetilde \theta = e^b_{i}(t; \tau) -  Z_{i,t}^{b \top} (\widetilde \theta -  \theta_0 )$, this implies
\begin{eqnarray*}
E^* S_{it}   ( \theta ) &=&Z_{i,t}   E^* \Big[ 
\psi_\tau(   e^b_{i}(t; \tau) - Z_{i,t}^\top ( \theta- \widetilde{\theta} ) - Z_{i,t}^{b \top} (\widetilde \theta -  \theta_0 )  )  -
\psi_\tau(   e^b_{i}(t; \tau) - Z_{i,t}^{b \top} (\widetilde \theta -  \theta_0 )  )  	\Big]  \\
& =&
Z_{i,t}  E^* \Big[   I( e^b_{i}(t; \tau)  < Z_{i,t}^\top ( \theta- \widetilde{\theta} ) + Z_{i,t}^{b \top} (\widetilde \theta -  \theta_0 )  )  
- I (   e^b_{i}(t; \tau)  < Z_{i,t}^{b \top} (\widetilde \theta -  \theta_0 )  )   \Big]   \\
&=&
Z_{i,t} Z_{i,t}^\top ( \theta- \widetilde{\theta} )  f_{e_i}(0, t; \tau)  +  
f^\prime_{e_i}(0, t; \tau)  [ O_p (  | Z_{i,t} |^{3} (\widetilde \theta -  \theta_0 ) ^2   )
+ O_{p^*}(  | Z_{i,t} |^{3} ( \theta- \widetilde{\theta} )  ^2   ) ],
\end{eqnarray*}
where the last equality follows from the Taylor's expansion.
Direct calculations lead to 
$$
n^{-1} \sum_{i=1}^n Z_{i,t}  \psi_\tau(   \widehat{e}^b_{i, \tau}( t )  ) =
D_t ( \theta- \widetilde{\theta} )   +o_{p^*}(1) .
$$

(iv) The proof of the bootstrap consistency for QIE follows directly from the proof of Theorem \ref{thm:bootstrap_consistency_QTE} and we omit here.
This completes the proof. 
\end{proof}

\begin{lemma}
\label{lem:uniform_Bahadur}
Suppose that Assumptions \ref{assump:F_e} and \ref{assump:Z} hold,
and
$\log^3(nm)/n^{1/4} \rightarrow 0$ as $n \rightarrow \infty$. Then for any $\varepsilon > 0$
we have
\begin{eqnarray}
\label{eq:uniform_Bahadur}
\sqrt{n} D_t (\widehat{\theta}_\tau(t) - \theta_\tau (t) ) 
=
{1 \over \sqrt{n}} \sum_{i=1}^n Z_{i,t} \psi_\tau( e_{i,\tau}( t) ) + O_p( n^{-1/4} \log^3(nm) + n^{-1/2} \log^4(nm)),
\end{eqnarray}
where the big-$O$ term is uniform in $\tau \in [\varepsilon, 1- \varepsilon]$ and $t \in \mathcal{T}$.
\end{lemma}

\begin{proof}[\bf Proof of Lemma \ref{lem:uniform_Bahadur}.]
Define $g(\delta_t, \tau) = \sum_{i=1}^n [ I( E_{i}(t ) \leq F_e^{-1}(t,\tau) + Z_{i,t}^\top \delta_t  ) - \tau ] Z_{i,t}$, 
$T(\delta_t, \tau) = g(\delta_t, \tau) - g(0, \tau)$, $\tilde T(\delta_t, \tau) = T(\delta_t, \tau) - E T(\delta_t, \tau)$.
The proof 
consists of three parts, as detailed below.  
\begin{itemize}
\item Part 1 shows that there exists some $K > 0$ such that with probability approaching $1$, 
$$
\sup_{t \in [0, 1], \tau \in [\varepsilon, 1- \varepsilon]} \| \widehat{\theta}_\tau(t) - \theta_\tau (t) \| \leq K( \log(nm)/n )^{1/2}.
$$
This part is very similar to that Lemma A.1 in \cite{koenker1987estimation}, and we omit the proof for brevity. 

\item
Part 2 shows that for any $j$, the $j$-th component of $\tilde T$, denoted by $\tilde T_j$, satisfies that for any $\lambda > 0 $ and  $ \delta_t \in \bigtriangleup = \{ \delta_t \in \mathbb{R}^{d+2} ~| ~\| \delta_t \| \leq K( \log(nm) /n )^{1/2}  \}$,
\begin{eqnarray}
\label{eq:tilde_T_j}
P( |\tilde T_j  |  \geq \lambda n^{1/4} \log^3(nm)  ) \leq 2 \exp( - \lambda \log^2(nm) ( 1+o(1))).
\end{eqnarray}

\item
Part 3 shows that
$$
\sup_{t \in [0, 1], \tau \in [\varepsilon, 1- \varepsilon]} \| \tilde T(\widehat{\theta}_\tau(t) - \theta_\tau (t), \tau)  \| = O_p( n^{1/4} \log^3(nm) ),
$$	
which follows from \eqref{eq:tilde_T_j} and the proof of Lemma A.2 in \cite{koenker1987estimation}. We omit the proof to save space. 
\end{itemize}

In the rest of the proof, we show that \eqref{eq:tilde_T_j} holds.
By the Markov inequality, we have for any $\epsilon > 0$ and any $\lambda_n >0$ that
\begin{eqnarray}
\label{eq:tilde_T_j_markov}
P( | \tilde T_{j} | \geq \lambda_n ) \leq 2 e^{- \epsilon \lambda_n} [ M_{jt}(\epsilon) + M_{jt}(-\epsilon) ],
\end{eqnarray}
where $M_{jt}(\epsilon) = \Pi_i M_{ijt} (\epsilon) $ where
\begin{eqnarray*}
M_{ijt} (\epsilon) &=& E \exp\{ t Z_{ijt} [ J_i( \delta_t, \tau ) - E J_i(\delta_t, \tau) ]  \}, \\
J_i( \delta_t, \tau ) &=& I( E_{i} (t) \leq F_e^{-1}(t,\tau) + Z_{ijt} \delta_t) - I(  E_{i} (t) \leq F_e^{-1}(t,\tau) ).
\end{eqnarray*}
Notice that $E J_i(\delta_t, \tau) = \text{sgn}(   Z_{ijt} \delta_t  ) p_{it} $ for
$p_{it} =  P( F_e^{-1}(t,\tau) \leq E_{i} (t) \leq F_e^{-1}(t,\tau) + Z_{ijt} \delta_t )$, we have
$$
M_{ijt} (\epsilon) = p_{it} \exp\{ \epsilon Z_{ijt} ( 1 - p_{it} ) \text{sgn}(  Z_{ijt} \delta_t )  \} + 
( 1 -  p_{it} ) \exp\{ -\epsilon Z_{ijt} p_{it}   \text{sgn}(  Z_{ijt} \delta_t )   \}.
$$
Under the sub-Gaussian assumption, we have that
$E(\max_i \max_t \| Z_{i,t} \| ) \lesssim \log(nm)$,
there exist constants $c, B >0$ such that
\begin{eqnarray*}
\max_t \log  M_{jt}
& \leq &
\sup_t  \sum_{i=1}^n c \| \delta_t \| \| Z_{i,t} \|^3 \epsilon^2 \exp ( B \epsilon \log n m) \\
& \leq &
\sup_t  \sum_{i=1}^n \| Z_{i,t} \|^3   c  K( \log n m /n )^{1/2} \epsilon^2  \exp ( B \epsilon \log n m) \\
& \lesssim &
n^{1/2} ( \log nm )^{7/2} \epsilon^2  \exp ( B \epsilon \log n m).
\end{eqnarray*}
plugging the above result into \eqref{eq:tilde_T_j_markov} and taking $\epsilon=  n^{-1/4} \log(nm)^{-1} $ 
and $\lambda_n = n^{1/4} \log(n m)^3  $
it is easy to see that
\begin{eqnarray*}
P( | \tilde T_{jt} | \geq \lambda n^{1/4} \log^3(n m) ) & \lesssim & 2 \exp\{ - \lambda \log^2(nm) + \log^{3/2}(nm)  \exp ( B n^{-1/4} )  \} \\
&=&
2 \exp\{ - \lambda \log^2(nm) (1 + o(1))  \}.
\end{eqnarray*}

Notice that for $ \delta_t \in \bigtriangleup $, 
\begin{eqnarray*}
E g(\delta_t, \tau) &=& \sum_{i=1}^n Z_{i,t} ( Z_{i,t}^\top \delta_t ) f_e( F^{-1}_e( t, \tau ), t; \tau )  + \sum_{i=1}^n Z_{i,t} ( Z_{i,t}^\top \delta_t ) f^\prime_e( F^{-1}_e( t, \tau ), t; \tau ) \\
&=&
n D_t \delta_t f_e( F^{-1}_e( t, \tau ), t; \tau ) + \sum_i \|Z_{i,t} \|^3  \delta_t^2 f^\prime_e( F^{-1}_e( t, \tau ), t ; \tau ) \\
&=&
n D_t \delta_t f_e( F^{-1}_e( t, \tau ), t; \tau ) + O( \log^4(nm) )
\end{eqnarray*}	
holds uniformly for $t$ by using the fact that $E(\max_i \max_t \| Z_{i,t} \| ) = \log(nm)$ and $f^\prime_e( F^{-1}_e( t, \tau ), t; \tau )$ is uniformly bounded.
Let $\delta_t = \widehat \theta_\tau(t) - \theta_\tau(t)$, then
\begin{eqnarray*}
\tilde T(\delta_t, \tau) = T(\delta_t, \tau) - E T(\delta_t, \tau) = O_p( n^{1/4} \log^3(nm) ).
\end{eqnarray*}	
Using $\|g(\delta_t, \tau)\| = \log(nm)$ by the sub-Gaussian property of $Z_{i,t}$, direct calculations lead to 
\begin{eqnarray*}
n D_t \delta_t f_e( F^{-1}_e( t, \tau ), t; \tau ) = 
\sum_{i=1}^n [ \tau  - I( E_{i}(t ) \leq F_e^{-1}(t,\tau) + Z_{i,t}^\top \delta_t  ) ] Z_{i,t} 
+O_p(\log^4(nm) +  n^{1/4} \log^3(nm)   ).
\end{eqnarray*}	
This completes the proof.
\end{proof}

\end{document}